%% file: main.tex
\newif\ifcomments   
\newif\ifanon       
\newif\ifcrypto     
\newif\ifllncs      
\newif\ifpublish    

\ifllncs
  \documentclass[runningheads,a4paper]{llncs}

\else
  \documentclass[11pt,pdfa]{article}
  \usepackage[in]{fullpage}
\fi

\usepackage{iftex}
\ifPDFTeX
  \usepackage[utf8]{inputenc}
  \usepackage[noTeX]{mmap}
  \usepackage[T1]{fontenc}
\fi
\ifLuaTeX
  \usepackage{luatex85}
  \usepackage[noTeX]{mmap}
\fi

\usepackage{comment}
\usepackage{mmap}
\usepackage[shortlabels]{enumitem}
\usepackage{bm}
\usepackage{amsfonts}
\usepackage{amsmath}
\usepackage{amssymb}
\usepackage{amsthm}
\usepackage{xcolor}
\usepackage{graphicx}
\usepackage{qtree}
\usepackage{tree-dvips}
\usepackage{float}
\usepackage{hyperref}
\usepackage[nameinlink]{cleveref}
\usepackage{braket}
\usepackage{mathrsfs}
\usepackage{tikz}
\usepackage{qcircuit}

\allowdisplaybreaks 

\ifllncs
  
  \newtheorem{construction}[theorem]{Construction} 
\else
  \hypersetup{colorlinks={true},linkcolor={blue},citecolor=magenta}

  \theoremstyle{plain}
  \newtheorem{theorem}{Theorem}[section]
  \newtheorem{lemma}[theorem]{Lemma}
  \newtheorem{claim}[theorem]{Claim}
  \newtheorem{corollary}[theorem]{Corollary}
  \newtheorem{definition}[theorem]{Definition}
  \newtheorem{remark}[theorem]{Remark}
  \newtheorem{conjecture}[theorem]{Conjecture} 
  \newtheorem{construction}[theorem]{Construction} 

  \usepackage[style=alphabetic,minalphanames=3,maxalphanames=4,maxnames=99,backref=true]{biblatex}

  \DeclareFieldFormat{eprint:iacr}{Cryptology ePrint Archive: \href{https://ia.cr/#1}{\texttt{#1}}}
  \DeclareFieldFormat{eprint:iacrarchive}{Cryptology ePrint Archive: \href{https://eprint.iacr.org/archive/#1}{\texttt{#1}}}
  \AtEveryBibitem{
     \clearlist{address}
     \clearfield{date}
     \clearfield{isbn}
     \clearfield{issn}
     \clearlist{location}
     \clearfield{month}
     \clearfield{series}
     \ifentrytype{book}{}{
      \clearlist{publisher}
      \clearname{editor}
     }
  }

\fi
\newtheorem{fact}[theorem]{Fact}

\ifdefined\authorNote
\newcommand{\YTnote}[1]{{\color{blue} [YT: #1]}}
\newcommand{\pnote}[1]{{\color{red} [Prab: #1]}}
\newcommand{\AG}[1]{{\color{green} [AG: #1]}}
\newcommand{\fatih}[1]{{\color{magenta} [F: #1]}}
\else
\newcommand{\YTnote}[1]{}
\newcommand{\pnote}[1]{}
\newcommand{\AG}[1]{}
\newcommand{\fatih}[1]{}
\fi

\input{headers}

\usepackage[style=alphabetic,minalphanames=3,maxalphanames=4,maxnames=99,backref=true]{biblatex}

  \DeclareFieldFormat{eprint:iacr}{Cryptology ePrint Archive: \href{https://ia.cr/#1}{\texttt{#1}}}
  \DeclareFieldFormat{eprint:iacrarchive}{Cryptology ePrint Archive: \href{https://eprint.iacr.org/archive/#1}{\texttt{#1}}}
\newcommand{\comp}{{\sf Comp}}

\newcommand{\single}{{\sf Single}}

\newcommand{\pri}{{\sf PRI}}
\newcommand{\prfsi}{{F}}
\newcommand{\haarisometry}{{\cal I}}
\newcommand{\prp}{{\sf PRP}}

\newcommand{\prf}{{\sf PRF}}

\newcommand{\sphere}{\mathcal{S}}

\date{}
\title{Pseudorandom Isometries}

\ifllncs
\author{}
\institute{}
\else
\author{Prabhanjan Ananth\thanks{\texttt{prabhanjan@cs.ucsb.edu}} \\{\small UCSB} \and Aditya Gulati\thanks{\texttt{adityagulati@ucsb.edu}} \\ {\small UCSB} \and Fatih Kaleoglu\thanks{\texttt{kaleoglu@ucsb.edu}} \\ {\small UCSB} \and Yao-Ting Lin\thanks{\texttt{yao-ting\_lin@ucsb.edu}} \\ {\small UCSB}} 
\fi

\begin{document}

\maketitle

\abstract{
\noindent We introduce a new notion called ${\cal Q}$-secure pseudorandom isometries (PRI). A pseudorandom isometry is an efficient quantum circuit that maps an $n$-qubit state to an $(n+m)$-qubit state in an isometric manner. In terms of security, we require that the output of a $q$-fold PRI on $\rho$, for $ \rho \in {\cal Q}$, for any polynomial $q$, should be computationally indistinguishable from the output of a $q$-fold Haar isometry on $\rho$. 
\par By fine-tuning ${\cal Q}$, we recover many existing notions of pseudorandomness. We present a construction of PRIs and assuming post-quantum one-way functions, we prove the security of ${\cal Q}$-secure pseudorandom isometries (PRI) for different interesting settings of ${\cal Q}$. 
\par We also demonstrate many cryptographic applications of PRIs, including, length extension theorems for quantum pseudorandomness notions, message authentication schemes for quantum states, multi-copy secure public and private encryption schemes, and succinct quantum commitments. 
}

\ifllncs

\else
\newpage
\tableofcontents
\newpage
\fi

\input{intro}
\input{prelims} 
\input{definition}
\input{construction/invar_to_sec}
\input{construction_root_of_unity}

\input{applications}

\section*{Acknowledgements}
We thank Fermi Ma for useful discussions. 

\printbibliography
\appendix

\ifllncs
    \input{haar_ortho_single}
    \input{Appommittedproof}

\input{MultiEnc}\input{SuccinctQSC}
\else
    \input{PRIimpliedPRFS}
\fi
\end{document}

%% file: headers.tex
\newcommand{\swaptest}{\mathsf{SwapTest}}
\newcommand{\permtest}{\mathsf{PermTest}}
\newcommand{\hybrid}{\mathsf{H}}
\newcommand{\adversary}{\mathcal{A}}

\newcommand{\ketbra}[2]{\ket{#1}\!\bra{#2}}
\renewcommand{\cal}[1]{\mathcal{#1}}

\newcommand{\C}{\mathbb{C}}
\newcommand{\N}{\mathbb{N}}
\newcommand{\E}{\mathop{\mathbb{E}}}
\newcommand{\Tr}{\mathrm{Tr}}

\newcommand{\distinct}{\mathsf{distinct}}
\newcommand{\sym}{\mathsf{sym}}
\newcommand{\eps}{\varepsilon}

\newcommand{\Haar}{\mathscr{H}}

\newcommand{\pqstates}[1]{{\cal S}\left( \C^{2^{#1}} \right)}

\DeclareMathOperator{\TD}{TD}

\newcommand{\norm}[1]{{\left\lVert#1\right\rVert}}

\newcommand{\bit}{{\{0, 1\}}}

\newcommand{\bfX}{\mathbf{X}}

\newcommand{\bfZ}{\mathbf{Z}}

\newcommand{\secparam}{\lambda}

\DeclareFontFamily{U}{skulls}{}
\DeclareFontShape{U}{skulls}{m}{n}{ <-> skull }{}

\newcommand{\prob}{{\sf Pr}}

\newcommand{\nn}{\nonumber}

\newcommand{\negl}{\mathsf{negl}}
\newcommand{\poly}{\mathsf{poly}}
\newcommand{\secp}{\secparam}

\newcommand{\unique}{\mathsf{uni}}
\newcommand{\dis}{\mathsf{dis}}

\newcommand{\Ex}{\E}

\newcommand{\cD}{\mathcal{D}}

\newcommand{\cH}{\mathcal{H}}
\newcommand{\cI}{\mathcal{I}}
\newcommand{\cL}{\mathcal{L}}

\newcommand{\cS}{\mathcal{S}}
\newcommand{\cU}{\mathcal{U}}

\newcommand{\qprf}{\mathsf{QPRF}}
\newcommand{\qprp}{\mathsf{QPRP}}

\newcommand{\expand}{\mathsf{Expand}}

\newcommand{\veps}{\varepsilon}

\newcommand{\ugets}{\xleftarrow{\$}}
\renewcommand{\set}[1]{\{#1\}}

\newcommand{\expbraket}[3]{\langle#1|#2|#3\rangle}
\newcommand{\inner}[2]{\langle#1|#2\rangle}

\newcommand{\bracS}[1]{\left[ #1 \right]}
\newcommand{\bracC}[1]{\left\{ #1 \right\}}

\newcommand{\alice}{\mathcal{A}}

\newcommand{\tuni}[3]{\mathcal{T}_{\unique^{ #1 }_{ { #2 } , { #3 } }}}
\newcommand{\tdis}[3]{\mathcal{T}_{\dis^{ #1 }_{ { #2 } , { #3 } }}}

\newcommand{\hamming}{\mathsf{size}}
\newcommand{\type}{\mathsf{type}}

\newcommand{\setT}{\mathsf{set}}

\newcommand{\freq}[2]{\mathsf{freq}_{#2} (#1)}
\newcommand{\supp}{\mathsf{supp}}

\newcommand{\inv}{\mathsf{Inv}}
\newcommand{\cmp}{\comp}

\newcommand{\setup}{\mathsf{Setup}}
\newcommand{\enc}{\mathsf{Enc}}
\newcommand{\dec}{\mathsf{Dec}}
\newcommand{\distinguisher}{{\cal D}}
\newcommand{\sk}{\mathsf{sk}}
\newcommand{\ct}{\mathsf{ct}}
\newcommand{\pk}{\mathsf{pk}}
\newcommand{\smsetup}{\mathsf{setup}}
\newcommand{\smenc}{\mathsf{enc}}
\newcommand{\smdec}{\mathsf{dec}}
\newcommand{\sign}{\mathsf{Sign}}
\newcommand{\ver}{\mathsf{Ver}}
\newcommand{\aux}{\mathsf{Aux}}

%% file: intro.tex
\newcommand{\qclass}{\mathcal{Q}}
\newcommand{\haarunitary}{\mathcal{U}}

\section{Introduction}

\noindent Pseudorandomness has played an important role in theoretical computer science. In classical cryptography, the notions of pseudorandom generators and functions have been foundational, with applications to traditional and advanced encryption schemes, signatures, secure computation, secret sharing schemes, and proof systems. On the other hand, we have only just begun to scratch the surface of understanding the implications pseudorandomness holds for quantum cryptography, and there is still a vast uncharted territory waiting to be explored.
\par When defining pseudorandomness in the quantum world, there are two broad directions one can consider. 
\ \\
\paragraph{Quantum States.} Firstly, we can study pseudorandomness in the context of quantum states. Ji, Liu, and Song (JLS)~\cite{JLS18} proposed the notion of a pseudorandom quantum state generator, which is an efficient quantum circuit that on input a secret key $k$ produces a quantum state (referred to as a pseudorandom quantum state) that is computationally indistinguishable from a Haar state as long as $k$ is picked uniformly at random and moreover, the distinguisher is given many copies of the state. JLS and the followup works by Brakerski and Shmueli~\cite{BS19,BrakerskiS20} presented constructions of pseudorandom quantum state generators from one-way functions. Ananth, Qian, and Yuen~\cite{AQY21} defined the notion of a pseudorandom function-like quantum state generator, which is similar to pseudorandom quantum state generators, except that the same key can be used to generate multiple pseudorandom quantum states. These two notions have many applications, including in quantum gravity theory~\cite{BFV19,ABFGVZZ23}, quantum machine learning~\cite{HBCLMNBKP22}, quantum complexity~\cite{Kretschmer21}, and quantum cryptography~\cite{AQY21,MY22}. Other notions of pseudorandomness for quantum states have also been recently explored~\cite{ABFGVZZ23,ABKKK23,GLGEYQ23}.

\paragraph{Quantum Operations.} Secondly, we can consider pseudorandomness in the context of quantum operations. This direction is relatively less explored. One prominent example, proposed in the same work of~\cite{JLS18}, is the notion of pseudorandom unitaries, which are efficient quantum circuits such that any efficient distinguisher should not be able to distinguish whether they are given oracle access to a pseudorandom unitary or a Haar unitary. Establishing the feasibility of pseudorandom unitaries could have ramifications for quantum gravity theory (as noted under open problems in~\cite{GLGEYQ23}),  quantum complexity theory~\cite{Kretschmer21}, and cryptography~\cite{GJMZ23}. Unfortunately, to date, we do not have any provably secure construction of pseudorandom unitaries, although some candidates have been proposed in~\cite{JLS18}. A recent independent work by by Lu, Qin, Song, Yao, and Zhao~\cite{LQSYZ23} takes an important step towards formulating and investigating the feasibility of pseudorandomness of quantum operations. They define a notion called pseudorandom state scramblers that isometrically maps a quantum state $\ket{\psi}$ into another state $\ket{\psi'}$ such that $t$ copies of $\ket{\psi'}$, where $t$ is a polynomial, is computationally indistinguishable from $t$ copies of a Haar state. They establish its feasibility based on post-quantum one-way functions. In the same work, they also explored cryptographic applications of pseudorandom state scramblers. \\

\noindent Although pseudorandom state scramblers can be instantiated from one-way functions, the definition inherently allows for scrambling only a single state. On the other extreme, pseudorandom unitaries allow for scrambling polynomially many states but unfortunately, establishing their feasibility remains an important open problem. Thus, we pose the following question: 

\begin{quote}
\begin{center}
{\em Is there a pseudorandomness notion that can scramble polynomially many states and\\ can be provably instantiated based on well studied cryptographic assumptions?}    
\end{center}
\end{quote}

\paragraph{Our Work in a Nutshell.} We address the above question in this work. Our contribution is three-fold:
\begin{enumerate}
    \item \underline{\textsc{New definitions}}: We introduce a new notion called ${\cal Q}$-secure pseudorandom isometries that can be leveraged to scramble many quantum states coming from the set ${\cal Q}$.  
    \item \underline{\textsc{Construction}}: We present a construction of pseudorandom isometries and investigate its security for different settings of ${\cal Q}$.
    \item \underline{\textsc{Applications}}: Finally, we explore many cryptographic applications of pseudorandom isometries. 
\end{enumerate}

\subsection{Our Results}
\label{sec:ourresults}
\noindent Roughly speaking, a pseudorandom isometry is an efficient quantum circuit, denoted by $\pri_k$, parameterized by a key\footnote{We denote $\secparam$ to be the security parameter.} $k \in \{0,1\}^{\secparam}$ that takes as input an $n$-qubit state and outputs an $(n+m)$-qubit state with the guarantee that $\pri_k$ is functionally equivalent to an isometry. In terms of security, we require that any efficient distinguisher should not be able to distinguish whether they are given oracle access to $\pri_k$ or a Haar isometry\footnote{The Haar distribution of isometries is defined as follows: first, sample a unitary from the Haar measure, and then set the isometry, that on input a quantum state $\ket{\psi}$, first initializes an ancilla register containing zeroes and then applies the Haar unitary on $\ket{\psi}$ and the ancilla register.} ${\cal I}$. We consider a more fine-grained version of this definition in this work, where we could fine-tune the set of allowable queries. 
\par More precisely, we introduce a concept called \emph{$(n,n+m)$-$\qclass$-secure-pseudorandom isometries} (PRIs). Let us first consider a simplified version of this definition. Suppose $n(\secparam),q(\secparam)$ are polynomials and $\qclass_{n,q,\secparam}$ is a subset of $nq$-qubit (mixed) states. Let $\qclass = \{\qclass_{n,q, \secparam}\}_{\secparam \in \mathbb{N}}$. The definition states that it should be computationally infeasible to distinguish the following two distributions: for any polynomials $q$,
\begin{itemize}
    \item $\left( \rho,\ \pri_k^{\otimes q} \left( \rho \right) \right)$, 
    \item $\left( \rho,\ {\cal I}^{\otimes q} \left( \rho \right) ({\cal I}^{\dagger})^{\otimes q}\right)$,
\end{itemize}
where $\rho \in \qclass_{n,q,\secparam}$ and ${\cal I}$ is a Haar isometry. 
\par Let us consider some examples. 
\begin{enumerate}
    \item If $\qclass_{n,q,\secparam}=\{\ket{0^n}^{\otimes q} \}$ then this notion implies a pseudorandom state generator (PRSG)~\cite{JLS18}.
    \item If $\qclass_{n,q,\secparam}$ consists of all possible $q$ computational basis states then this notion implies a pseudorandom function-like state generator  (PRFSG) \cite{AQY21,AGQY22}.
    \item If $\qclass_{n,q,\secparam}$ consists of $q$-fold tensor of all possible $n$-qubit states then this notion implies a pseudorandom state scrambler (PSS)~\cite{LQSYZ23}. 
\end{enumerate}
\noindent We can generalize this definition even further. Specifically, we allow the adversary to hold an auxiliary register that is entangled with the register on which the $q$-fold isometry ($\pri_k$ or Haar) is applied and we could require the stronger security property that the above indistinguishability should hold even in this setting. 
\par In more detail, $\rho$ is now an $(nq+\ell)$-qubit state and the distinguisher is given either of the following: 
\begin{itemize}
\item $\left( \rho, \left( I_{\ell} \otimes \pri_k^{\otimes q} \right) (\rho) \right)$, 
\item $\left( \rho, \left( I_{\ell} \otimes {\cal I}_k^{\otimes q} \right) \rho \left( I_{\ell} \otimes {{\cal I}^{\dagger}_k}^{\otimes q} \right) \right)$
\end{itemize}
where $I_{\ell}$ is an $\ell$-qubit identity operator. We can correspondingly define $\qclass$ to be instead parameterized by $n, q, \ell, \secparam$, and we require $\rho \in \qclass_{n,q,\ell,\secparam}$.
\par The above generalization captures the notion of pseudorandom isometries (discussed in the beginning of~\Cref{sec:ourresults}) against selective queries. Specifically, if $\pri_k$ is a ${\cal Q}$-secure pseudorandom isometry (according to the above-generalized definition), where ${\cal Q}$ is the set of all possible $nq+\ell$-qubit states then indeed it is infeasible for an efficient distinguisher making selective queries\footnote{Roughly speaking, the selective query setting is one where all the queries are made at the same time. In contrast, in the adaptive query setting, each query could depend on the previous queries and answers.} to distinguish whether it has oracle access to $\pri_k$ or a Haar isometry oracle. 
\par Thus, by fine-tuning ${\cal Q}$, we recover many notions of pseudorandomness in the context of both quantum states and operations. 

\paragraph{Construction.} We first study the feasibility of PRIs.
\par We present a construction of PRIs and investigate its security for different settings of $\qclass$. On input an $n$-qubit state $\ket{\psi} = \sum_{x \in \{0,1\}^n} \alpha_{x } \ket{x}$, define $\pri_k \ket{\psi}$ as follows: 
$$\boxed{\pri_k \ket{\psi} = \frac{1}{\sqrt{2^m}}\sum_{x \in \{0,1\}^{n},y \in \{0,1\}^m} \alpha_x \cdot {\omega}_p^{f_{k_1}(x||y)} \ket{g_{k_2}(x||y)}}$$

\noindent In the above construction, we parse $k$ as a concatenation of two $\secparam_1$-bit strings $k_1$ and $k_2$, where $\secparam = 2\secparam_1$. The first key $k_1$ would serve as a key for a pseudorandom function $f:\{0,1\}^{\secparam_1} \times \{0,1\}^{n+m} \rightarrow \mathbb{Z}_p$, where $p\sim 2^{\secparam_1}$ is an integer. The second key $k_2$ would serve as a key for a pseudorandom permutation $g:\{0,1\}^{\secparam_1} \times \{0,1\}^{n+m} \rightarrow \{0,1\}^{n+m}$. Both $f$ and $g$ should satisfy quantum query security. Moreover, both of them can be instantiated from post-quantum one-way functions~\cite{Zha12,Zha16}. We require $n$ to be a polynomial in $\secparam$, larger than $\secparam$, and similarly, we set $m$ to be a polynomial in $\secparam$, larger than $\secparam$. 

The above construction was first studied by~\cite{BBBS23,ABFGVZZ23}, perhaps surprisingly, in completely different contexts. Brakerski, Behera, Sattath, and Shmueli~\cite{BBBS23} introduced a new notion of PRSG and PRFSG and instantiated these two notions using the above construction. Aaronson, Bouland, Fefferman, Ghosh, Vazirani, Zhang, and Zhou~\cite{ABFGVZZ23} introduced the notion of pseudo-entanglement and instantiated this notion using the above construction. 
\noindent An important property of this construction is that it is \emph{invertible}, that is, given the key $k$, it is efficient to implement $\inv_k$ such that $\inv_k \pri_k$ is the identity map. 
\par It is natural to wonder if it is possible to modify the above construction to have binary phase as against $p^{th}$ roots of unity, for a large $p$. There is some recent evidence to believe since~\cite{HBK23} showed that pseudorandom unitaries cannot just have real entries.

\paragraph{Security.} We look at different possible settings of $\qclass$ and study their security\footnote{We only consider a simplified version of these settings here and in the technical sections, we consider the most general version.}. \\

\noindent \underline{\textsc{I. Haar states.}} Our main contribution is showing that the output of $\pri_k$ on many copies of many $n$-qubit Haar states, namely, $(\ket{\psi_1}^{\otimes t},\ldots,\ket{\psi_s}^{\otimes t})$ with $t$ being a polynomial and $\ket{\psi_1},\ldots,\ket{\psi_s}$ are Haar states, is computationally indistinguishable from a Haar isometry on $(\ket{\psi_1}^{\otimes t},\allowbreak \ldots,\allowbreak \ket{\psi_s}^{\otimes t})$. Moreover, the computational indistinguishability should hold even if $(\ket{\psi_1}^{\otimes t},\allowbreak \ldots,\allowbreak \ket{\psi_s}^{\otimes t})$ is given to the QPT adversary. In other words, $\pri_k$ can be used to map maximally mixed states on smaller dimensional symmetric subspaces onto pseudorandom states on larger dimensional symmetric subspaces. We consider the following setting:
\begin{itemize}
    \item Let $t(\secparam)$ and $s(\secparam)$ be two polynomials. Let $q= s \cdot t$ and $\ell=n \cdot q$. 
    \item We define $\qclass_{{\sf Haar}}=\left\{ \qclass_{n,q,\ell,\secparam} \right\}_{\secparam \in \mathbb{N}}$, where $\qclass_{n,q,\ell,\secparam}$ is defined as follows\footnote{$\Haar_n$ denotes the Haar distribution on $n$-qubit Haar states.}:
    $$\qclass_{n,q,\ell,\secparam} =\left\{\mathbb{E}_{\ket{\psi_1},\ldots,\ket{\psi_s} \leftarrow \Haar_{n} }\left[ {\color{red} \bigotimes_{i=1}^{s}\ketbra{\psi_i}{\psi_i}^{\otimes t}} \otimes {\color{blue} \bigotimes_{i=1}^{s}\ketbra{\psi_i}{\psi_i}^{\otimes t}} \right]\right\}$$
    Recall that the first $\ell$ qubits (in the above case, it is the first $t$ {\color{red} red}-colored copies of $n$-qubit Haar states $\ket{\psi_1},\ldots,\ket{\psi_s}$) are not touched. On the next $q$ $n$-qubit states (colored in {\color{blue} blue}), either $\pri_k^{\otimes q}$ or ${\cal I}^{\otimes q}$ is applied.   
\end{itemize}

\noindent We prove the following. 

\begin{theorem}[Informal]
    Assuming post-quantum one-way functions exist, $\pri_k$ is a $\qclass_{{\sf Haar}}$-secure pseudorandom isometry. 
\end{theorem}

\noindent This setting is reminiscent of {\em weak} pseudorandom functions~\cite{DN02,ABGKR14} studied in the classical cryptography literature, where we require the pseudorandomness to hold only on inputs chosen from the uniform distribution on binary strings.\\

\noindent \underline{\textsc{Application: Length Extension Theorem.}} As an application, we demonstrate a length extension theorem for PRSGs and PRFSGs. Specifically, we show how to extend the output length of both these pseudorandomness notions assuming PRIs secure against Haar queries\footnote{An $(n,n+m)$-pseudorandom isometry secure against any $\qclass$ trivially gives a PRSG or PRFSG on $n+m$ qubits. However, our length extension theorem requires the underlying PRI to only be secure against Haar queries.}. Specifically, we show the following. 

\begin{theorem}[Length Extension Theorem; Informal]
Assuming $\qclass_{{\sf Haar}}$-secure pseudorandom isometry, mapping $n$ qubits to $n+m$ qubits, and an $n$-qubit $\mathrm{PRSG}$, there exists an $n+m$-qubit $\mathrm{PRSG}$. 
\par Similarly, assuming $\qclass_{{\sf Haar}}$-secure pseudorandom isometry, mapping $n$ qubits to $n+m$ qubits, and $n$-qubit $\mathrm{PRFSG}$, there exists an $(n+m)$-qubit $\mathrm{PRFSG}$. \
\end{theorem}

\noindent Prior to our work, the only known length extension theorem was by Gunn, Ju, Ma, and Zhandry~\cite{GJMZ23} who demonstrated a method to increase the output length of pseudorandom states and pseudorandom unitaries but at the cost of reducing the number of copies given to the adversary. That is, the resulting PRSG in their transformation is only secure if the adversary is given one copy. On the other hand, in the above theorem, the number of copies of the PRSG is preserved in the above transformation.  \\

\noindent \underline{\textsc{II. Many copies of an $n$-qubit state.}} We also consider the setting where we have multiple copies of a single state. Specifically, we consider the following setting: 
\begin{itemize}
    \item Let $q=q(\secparam)$ be a polynomial. Let $\ell = n \cdot q$.  
    \item We define $\qclass_{{\sf Single}} = \left\{ \qclass_{n,q,\ell,\secparam} \right\}_{\secparam \in \mathbb{N}}$, where $\qclass_{n,q,\ell,\secparam}$ is defined as follows:
    $$\qclass_{n,q,\ell,\secparam} = \left\{ {\color{red} \ket{\psi}^{\otimes q}} \otimes {\color{blue} \ket{\psi}^{\otimes q}} \ :\  \ket{\psi} \in {\cal S}(\C^{2^n}) \right\}$$
\end{itemize}

\noindent We prove the following. 

\begin{theorem}[Informal]
Assuming post-quantum one-way functions exist, $\pri_k$ is a $\qclass_{{\sf Single}}$-secure pseudorandom isometry. 
\end{theorem}

\noindent Informally, the above theorem ensures that even if an efficient distinguisher is given polynomially many copies of $\ket{\psi}$, for an arbitrary $n$-qubit state $\ket{\psi}$, it should not be able to efficiently distinguish $q$ copies of $\pri_k \ket{\psi}$ versus $q$ copies of ${\cal I} \ket{\psi}$, for any polynomial $q(\secparam)$. \\

\noindent \underline{\textsc{Application: Pseudorandom State Scamblers.}} A recent work~\cite{LQSYZ23} shows how to isometrically scramble a state such that many copies of the scrambled state should be computationally indistinguishable from many copies of a Haar state. Our notion of $\qclass_{{\sf Single}}$-secure pseudorandom isometry is equivalent to pseudorandom state scramblers. Thus, we have the following. 

\begin{theorem}[Informal]
$\qclass_{{\sf Single}}$-secure pseudorandom isometry exists if and only if pseudorandom state scramblers exist. 
\end{theorem}

\noindent The work of~\cite{LQSYZ23} presents an instantiation of pseudorandom scramblers from post-quantum one-way functions. While our result does not give anything new for pseudorandom scramblers in terms of assumptions, we argue that our construction and analysis are (in our eyes) much simpler than~\cite{LQSYZ23}. In addition to pseudorandom permutations and functions, they also use rotation unitaries in the construction. Their analysis also relies on novel and sophistical tools such as Kac random walks whereas our analysis is more elementary. \\

\noindent \underline{\textsc{Application: Multi-Copy Secure Public-Key Encryption.}} There is a simple technique to encrypt a quantum state, say $\ket{\psi}$: apply a quantum one-time pad on $\ket{\psi}$ and then encrypt the one-time pad keys using a post-quantum encryption scheme. However, the disadvantage of this construction is that the security is not guaranteed to hold if the adversary receives many copies of the ciphertext state. A natural idea is to apply a unitary $t$-design on $\ket{\psi}$ rather than a quantum one-time pad but this again only guarantees security if the adversary receives at most $t$ queries. On the other hand, we formalize a security notion called multi-copy secure public-key and private-key encryption schemes, where the security should hold even if the adversary receives arbitrary polynomially many copies of the ciphertext. 

\begin{theorem}[Informal]
Assuming $\qclass_{{\sf Single}}$-secure pseudorandom isometry\footnote{We additionally require that the pseudorandom isometry satisfy an invertibility condition. We define this more formally in the technical sections.}, there exists multi-copy secure private-key and public-key encryption schemes. 
\end{theorem}

\noindent The investigation of multi-copy security was independently conducted by~\cite{LQSYZ23}. However, they only studied multi-copy security in the context of one-time encryption schemes whereas we introduce the definition of multi-copy security for private-key and public-key encryption schemes and establish their feasibility for the first time. 

\ \\
\noindent \underline{\textsc{Conjecture.}} Unfortunately, we currently do not know how to prove that $\pri_k$ is a $\qclass$-secure pseudorandom isometry for every $\qclass$. We leave the investigation of this question as an interesting open problem. 

\begin{conjecture}
For every $\qclass=\{\qclass_{n,q,\ell,\secparam}\}_{\secparam \in \mathbb{N}}$, where $\qclass_{n,q,\ell,\secparam}$ consists of $nq$-qubit states, $\pri_k$ is a $\qclass$-secure pseudorandom isometry. 
\end{conjecture}

\paragraph{Other Applications.} We explore other applications of PRIs that were not covered before.\\

\noindent \underline{\textsc{Application: Quantum MACs.}} We explore novel notions of message authentication codes (MAC) for quantum states. Roughly speaking, in a MAC for quantum states, there is a signing algorithm using a signing key $sk$ that on input a state, say $\ket{\psi}$, outputs a tag that can be verified using the same signing key $sk$. Intuitively, we require that any adversary who receives tags on message states of their choice should not be able to produce a tag on a challenge message state. For the notion to be meaningful, we require that the challenge message state should be orthogonal (or small fidelity) to all the message states seen so far. 
\par There are different settings we consider: 
\begin{itemize}
    \item In the first setting, the verification algorithm gets as input multiple copies of the message state $\ket{\psi}$ and the tag state. In this case, we require the probability that the adversary should succeed is negligible. 
    \item In the second setting, the verification algorithm gets as input many copies of the message state but only a single copy of the tag. In this case, we weaken the security by only requiring that the adversary should only be able to succeed with inverse polynomial probability. 
    \item Finally, we consider the setting where we restrict the type of message states that can be signed. Specifically, we impose the condition that for every message state $\ket{\psi}$, there is a circuit $C$ that on input an all-zero state outputs $\ket{\psi}$. Moreover this circuit $C$ is known to the verification algorithm.  In this case, we require that the adversary only be able to succeed with negligible probability. 
\end{itemize}
\noindent We show how to achieve all of the above three settings using PRIs. \\

\noindent \underline{\textsc{Application: Length Extension Theorem.}} Previously, we explored a length extension theorem where we showed how to generically increase the output length of pseudorandom (function-like) state generators assuming only PRIs secure against Haar queries. We explore a qualitatively different method to extend the output length of pseudorandom states. Specifically, we show the following. 

\begin{theorem}[Informal]
Assuming the existence of $(n,n+m)$-secure pseudorandom isometry and an $(2n)$-output PRSG secure against $o(m)$ queries, there exists a $(2n+m)$-output PRSG secure against the same number of queries. Moreover, the key of the resulting PRSG is a concatenation of the $(2n)$-output PRSG and the $(n,n+m)$-secure PRI. 
\end{theorem}

\noindent One might be tempted to conclude that a unitary $o(m)$-design can be used to get the above result. The main issue with using a $o(m)$-design is that it increases the key size significantly~\cite{BCHKP21}. However, in the above theorem, if we start with a PRI with short keys (i.e., $\secparam \ll m$) then the above transformation gets a PRSG with a much larger stretch without increasing the key size by much. 

\subsection{Technical Overview} 

\subsubsection{Haar Unitaries: Observations}

Before we talk about proving security of our construction, we point out some useful properties of Haar unitaries. Note that Haar isometries are closely related to Haar unitaries since the former can be implemented by appending suitably many zeroes\footnote{The state being appended and the position of the new qubits is not important.} followed by a Haar random unitary.

\paragraph{Behavior on Orthogonal Inputs.}
In the classical world, a random function $f$ with polynomial output length is indistinguishable from the corresponding random permutation $g$ against a query-bounded black-box adversary $\alice$. One can prove this fact in three simple steps: \begin{enumerate}
    \item Without loss of generality one can assume $\alice$ only makes distinct queries $\bracC{x_1,\dots,x_q}$. \label{item:distinctquery}
    \item $f$ is perfectly indistinguishable from $g$ conditioned on the fact that $f(x_i) \ne f(x_j)$ for $i \ne j$.
    \item If the number $q$ is polynomial, then the probability that $f$ has a collision on $\bracC{x_1,\dots,x_q}$ is negligible. \label{item:nocollision}
\end{enumerate}

Now consider the quantum analogue of the same problem. Namely, consider two oracles $O_1, O_2$ that can only be queried on classical inputs, where: (1) $O_1$ on input ${x}$ outputs $\haarunitary\ket{x}$, where $\haarunitary$ is a Haar unitary; and (2) $O_2$ for each distinct input $x$, outputs an i.i.d. Haar-random state $\ket{\psi_x}$. Our goal is to show that $O_1,O_2$ are indistinguishable against a query-bounded quantum adversary $\alice$. If we try to replicate the classical proof above, we run into problems: we can no longer assume distinct queries due to the principle of no-cloning, and we need to generalize step \ref{item:nocollision} in a non-trivial to an almost-orthogonality argument. Instead, we consider an alternative proof for the classical case.
\par Fix the set of queries $\bracC{x_1,\dots,x_q}$ and for $0 \le i \le q$ define a hybrid oracle $O_i$ as follows: \begin{itemize}
    \item For $1\le j \le q$, if $x_j\in \{x_1,\dots,x_{q-1}\}$, then output consistently as the previous instance of the same query.
    \item Otherwise, for $1 \le j \le i$: On input $x_j$, sample $y_j \notin \{y_1,\dots,y_{j-1}\}$ uniformly at random and output $y_j$. For $i+1 \le j \le q$, sample an i.i.d. random answer $y_j$ and output $y_j$.
\end{itemize}

Now, one can argue that $O_{i}$ is perfectly indistinguishable from $O_{i+1}$ conditioned on the answer $y_{i+1}$ sampled by $O_i$ satisfying $y_{i+1} \notin \{y_1,\dots,y_i\}$. It turns out this argument is more easily generalizable to the quantum case, where we can define oracle $\widetilde{O}_i$ as answering $x_1,\dots,x_i$ using a random isometry and answering $x_{i+1},\dots,x_{q}$ using i.i.d. Haar-random states (while maintaining consistency). Indistinguishability of $\widetilde{O}_i$ and $\widetilde{O}_{i+1}$ follows from an analysis comparing the dimensions of the subspaces the hybrid oracles sample outputs from.

\paragraph{Almost-Invariance Property.}

\noindent The security definition for a pseudorandom unitary, and similarly isometry, can be cumbersome to work with. Let us focus on the information-theoretic setting first, i.e. when there is no computational assumption on the adversary besides a query bound. We investigate what it means for a candidate pseudorandom unitary $F_k$ to be information theoretically indistinguishable from a Haar unitary $\haarunitary$ for different query sets $\qclass$; in other words, we consider \emph{statistical $\qclass$-security} of $F_k$. Rather than attempting to directly calculate the trace distance between the output of $F_k$ on a given query $\rho$ and the output of a Haar unitary $\mathcal{U}$ on the same input, which may look significantly different for different values of $\rho$, we are naturally drawn to look for a simpler condition that suffices for security. 
\par Accordingly, we show that $F_k$ is statistically $\qclass$-secure if and only if for every $\rho \in \qclass$ which describes $q$ queries to $F_k$, we have that $F_k^{\otimes q}\rho (F_k^\dagger)^{\otimes q}$ changes only negligibly (in trace distance) under the action of $q$-fold Haar unitary $\haarunitary^{\otimes q}(\cdot)(\haarunitary^\dagger)^{\otimes q}$. We prove this fact for any quantum channel $\Phi$ (in particular for $\Phi(\cdot) = F_k(\cdot)F_k^{\dagger}$) as long as $\Phi$ is a mixture of unitary maps, and the proof follows by the unitary invariance of the Haar measure.
\par We note that the argument above can be easily generalized to a pseudorandom isometry (PRI), since an isometry can be decomposed into appending zeroes followed by applying a unitary.
The detailed proofs of the almost-invariance property can be found in \Cref{sec:haar_invariance}.
\\
\par Next, we will describe our construction, then discuss its security and applications in more detail.

\newcommand{\isom}{\widetilde I}
\subsubsection{Construction}
We describe how to naturally arrive at our construction of pseudorandom isometry, which was recently studied by~\cite{BBBS23,ABFGVZZ23} in different contexts. Given an input state $\ket{\psi} = \sum \alpha_x \ket{x}$, we will first apply an isometry $\isom$ to get a state $\ket{\varphi} = \sum \theta_z \ket{z}$, followed by unitary operations. A commonly used technique to scramble a given input state $\ket{\varphi}$ is to apply a random binary function $f$ with a phase kickback \cite{JLS18}, i.e. apply the unitary $O_{f}\ket{\psi} = \sum (-1)^{f(z)}\theta_z \ket{z} $. The action of $O_f$ on a mixed state $q$-query input $\rho = \sum_{\Vec{z},\Vec{z}'} \beta_{\vec{z},\vec{z}'} \ket{\vec{z}}\bra{\vec{z}'}$ can be calculated as \begin{align*}
    \E_f \bracS{ O_f^{\otimes q} \rho (O_f^\dagger)^{\otimes q} } & = \E_f \bracS{ \sum_{\Vec{z},\Vec{z}'} (-1)^{\sum_i f(z_i) + f(z_i')} \beta_{\vec{z},\vec{z}'} \ket{\vec{z}}\bra{\vec{z}'} } \\
    &= \sum_{\Vec{z},\Vec{z}'}\beta_{\vec{z},\vec{z}'} \ket{\vec{z}}\bra{\vec{z}'}  \E_f \bracS{(-1)^{\sum_i f(z_i) + f(z_i')} }.
\end{align*}
Observe that if $\Vec{z}$ and $\Vec{z}'$ are related by a permutation\footnote{This condition will later be referred to as $\Vec{z}$ and $\vec{z}'$ having the same \emph{type}.}, then $(-1)^{\sum_i f(z_i) + f(z_i')} = 1$. Otherwise, if there exists $z$, which occurs odd number of times in $\vec{z}$ and even number of times in $\vec{z}'$ (or vice versa), we get $(-1)^{\sum_i f(z_i) + f(z_i')} = 0$. Ideally we would like all terms $\ket{\vec{z}}\bra{\vec{z}'}$ to vanish when $\vec{z}$ and $\vec{z}'$ are not related by a permutation. We can easily fix this by switching to $p$-th root of unity phase kickback, i.e. apply $\widetilde{O}_f$ for a random function $f$ with codomain $\mathbb{Z}_p$, where $\widetilde{O}_f \ket{\psi} = \sum_x \omega_p^{f(x)}\ket{x}$ and $\omega_p = e^{2\pi i/p}$. As long as $q \ll p$ (e.g. $q$ is polynomial and $p$ is super-polynomial), we get that \begin{align*}
    \E_f \bracS{ \widetilde{O}_f^{\otimes q} \rho (\widetilde{O}_f^\dagger)^{\otimes q} } & = \sum_{\substack{\vec{z},\vec{z}' \\ \exists \sigma :\; \vec{z}' = \sigma(\vec{z})}} \beta_{\vec{z},\vec{z}'} \ket{\vec{z}}\bra{\vec{z}'}.
\end{align*}

\noindent Now we would like to scramble the remaining terms $\ket{\Vec{z}}\bra{\Vec{z}'}$ in the equation above. A natural try is to apply a random permutation $\pi$ in the computational basis, denoted by $O_\pi$ as a unitary operation. Such an operation would scramble the term above as $O_\pi^{\otimes q}\ket{\Vec{z}}\bra{\Vec{z}'} (O_\pi^\dagger)^{\otimes q}$, which only depends on $\sigma$ as long as $\vec{z}$ has distinct entries. Hence, to achieve maximal scrambling we would like $\ket{\varphi}$ to have negligible weight on states $\ket{\vec{z}}$ with collisions of the form $z_i = z_j$.

\par In order to make sure that the weight on $\ket{\Vec{z}}$ with distinct entries is close to 1, we pick $\isom$ to append a uniform superposition of strings\footnote{Note that this step crucially relies on the fact that we are constructing a pseudorandom isometry, not a pseudorandom unitary.},
which brings us to the information-theoretic inefficient construction 
\begin{align}
    G_{(f,\pi)} \ket{\psi} = \frac{1}{\sqrt{2^m}} \sum_{x \in \bit^n, y \in \bit^m} \alpha_x \cdot {\omega}_p^{f(x||y)} \ket{\pi(x||y)}, \label{eq:gfpi}
\end{align} 
To make the construction efficient, we instantiate $f$ and $g$ with a post-quantum pseudorandom function and a post-quantum pseudorandom permutation, respectively, hence reaching our construction \begin{align*}
     F_{(k_1,k_2)} \ket{\psi} = \frac{1}{\sqrt{2^m}} \sum_{x \in \bit^n, y \in \bit^m} \alpha_x \cdot {\omega}_p^{f_{k_1}(x||y)} \ket{g_{k_2}(x||y)}.
\end{align*}

\subsubsection{Security Proof}
As a first step, we argue that a QPT adversary cannot distinguish the PRF ($f_{k_1}$) and the PRP ($g_{k_2}$) from a random function and a random permutation, respectively. To show this we use a $2q$-wise independent hash function as an intermediate hybrid for $f_{k_1}$ to get an efficient reduction, following \cite{Zha12} who showed that such a hash function is indistinguishable from a random function under $q$ queries. Combining this with \cite{Zha16} who showed how to instantiate the PRP ($g_{k_2}$) from post-quantum one-way functions, we successfully invoke computational assumptions.
\par Now that we have invoked the computational assumptions as per the existence of quantum-secure PRF and PRP, we are left with the information theoretic construction given by $G_{(f, \pi)}$ (\cref{eq:gfpi}),
which is parametrized by a random function $f$ and a random permutation $\pi$. Below, we write $\rho \in \qclass$ as a short-hand to mean $\rho \in \qclass_{n,q,\ell,\secparam}$ for some $\secparam \in \mathbb{N}$. To show that $G_{(f,\pi)}$ is \emph{statistically} $\qclass$-secure for different query sets $\qclass$, we will show that the output of $G_{(f,\pi)}$ under any query $\rho \in \qclass$ is \emph{almost-invariant} under $q$-fold Haar unitary as per our second observation above. We achieve this in two steps: \begin{itemize}
    \item[]\textbf{Step 1:} Find a particular mixed state $\rho_\unique$, to be defined later, which is almost-invariant under $q$-fold Haar unitary. Conclude that if the output of $G_{(f,\pi)}$ under any query $\rho \in \qclass$ is negligibly close (in trace distance) to $\rho_\unique$, then it is $q$-fold Haar almost invariant, hence $G_{(f,\pi)}$ satisfies statistical $\qclass$-security.
    \item[]\textbf{Step 2:} For 3 different instantiations of $\qclass$, prove that the condition in Step 1 is satisfied, hence $G_{(f,\pi)}$ is statistically $\qclass$-secure.
\end{itemize}
Note that our proof-strategy outlined above is a top-down approach, and the first two steps can be viewed as reducing the problem of PRI-security to a simpler condition that is easier to check for different query sets, and is independent of the action of Haar isometry on $\qclass$. In Step 3, we show instantiations of $\qclass$ that satisfy the simpler condition. Next, we delve into the details of each step.


\paragraph{Step 1: An Almost-Invariant State: $\rho_{\unique}$.} 
\par Having established $q$-fold Haar almost-invariance as a sufficient condition for statistical security of $G_{(f,\pi)}$, it is natural to ask the question: \begin{quote}
    {\em Can we find a state $\rho^*$ which is both:\\ (a) close to the output of $G_{(f,\pi)}$ on certain inputs, and\\ (b) $q$-fold Haar almost-invariant?}
\end{quote}
\noindent This would allow us to use negligible closeness to $\rho^*$ as a sufficient condition for $q$-fold Haar almost-invariance, hence for statistical security of $G_{(f,\pi)}$. We start by analyzing condition (a).
\par We restrict our attention to queries with a particular, yet quite general, structure. Namely, suppose $\qclass = \bracC{\qclass_{n,q,\ell,\secparam}}$ is such that every $\rho \in \qclass$ is a mixture of pure states of the form $\bigotimes_{i=1}^s \ket{\psi_i}^{\otimes t}$, where $q = st$. In other words, the adversary makes queries in the form of $s$ states with $t$-copies each, or formally queries from the $s$-fold tensor product of symmetric subspaces, denoted by $\cH = \left( \vee^t \mathbb{C}^N \right)^s$. For such inputs, the output of the isometry will belong to the corresponding tensor product of symmetric subspaces $\cH' := \left( \vee^t \mathbb{C}^{NM} \right)^s$, where $N = 2^n$ and $M = 2^m$. It is known~\cite{Harrow13church} that $\cH$ is spanned by $s$-fold tensor product of \emph{type states} $\ket{\psi_{T_1,\dots,T_s}} = \bigotimes_{i=1}^s \ket{\type_{T_i}} $, where $\ket{\type_{T_i}}$ is a uniform superposition over computational basis states $\ket{\vec{x}} \in \C^{Nt}$ of the same \emph{type} ($T_i$), where $\vec{x}$ and $\vec{y}$ are said to have the same type if $\vec{y} = \sigma \vec{x}$ for some permutation $\sigma \in S_t$ over $t$ elements.
\par To understand the action of $G_{(f,\pi)}$ on $\qclass$, we consider its action on a basis state $\ket{\psi_{T_1,\dots,T_s}}$ of $\cH$. We first look at the action of a random isometry $\mathcal{I}$ on $\ket{\psi_{T_1,\dots,T_s}}$ and see that $$ \E_{\haarisometry} \bracS{ \haarisometry^{\otimes q} \ketbra{\psi_{T_1,\dots,T_s}}{\psi_{T_1,\dots,T_s}} \haarisometry^{\otimes q}} = \E_{T_1',\dots,T_s'} \bracS{ \ketbra{\psi_{T_1',\dots,T_s'}}{\psi_{T_1',\dots,T_s'}} } $$
is maximally mixed over $\cH'$, where $T_1',\dots,T_s'$ are types over $\C^{NMt}$. The same fact is not quite true for $G_{(f,\pi)}$ due to cross terms. Nonetheless, such terms cancel out whenever $(T_1,\dots,T_s)$ form a set of \emph{unique} types, denoted by $(T_1,\dots,T_s) \in \tuni{n}{s}{t}$, meaning collectively they span $st$ distinct computational basis states $\ket{x} \in \C^N$, thanks to the nice algebraic structure of the image of $f$, i.e. $\mathbb{Z}_p$. As a result, we get \begin{align} 
& \E_{f,\pi} \bracS{ G_{(f,\pi)}^{\otimes q} \ketbra{\psi_{T_1,\dots,T_s}}{\psi_{T_1,\dots,T_s}} G_{(f,\pi)}^{\otimes q}} \nn \\
& = \E_{(T_1',\dots,T_s') \leftarrow \tuni{n+m}{s}{t}} \bracS{ \ketbra{\psi_{T_1',\dots,T_s'}}{\psi_{T_1',\dots,T_s'}} } 
=: \rho_\unique
\label{eq:1}
\end{align}

for any $(T_1,\dots,T_s) \in \tuni{n}{s}{t}$. 
\noindent Fortunately, $\rho_\unique$ satisfies\footnote{We note that $\rho_\unique = \rho_{\unique_{s,t}}$ is parametrized by $s,t$ in the tecnhical sections, which we omit here for simplicity of notation.} property (b) as well. The reason is that the $q$-fold unique type states $\ket{\psi_{T_1,\dots,T_s}}$ constitute the vast majority\footnote{This follows from the fact that a random type will contain no repetitions with overwhelming probability as long as $t=\poly(\secparam)$.} of the basis for $\cH'$, so that $\rho_\unique$ is negligibly close to the maximally mixed state over $\cH'$, which is invariant under $q$-fold unitary operations. Therefore, if $G_{(f,\pi)}^{\otimes q}\rho (G_{(f,\pi)}^\dagger)^{\otimes q}$ is negligible close to $\rho_\unique$, then it is $q$-fold Haar almost-invariant, hence we have a simpler sufficient condition to check for PRI security as desired. Note that so far we have ignored the $\ell$-qubit (purification) register held by the adversary, but the arguments generalize without trouble. The detailed proofs of this step can be found in \Cref{sec:invariance_rho_uni}.

\paragraph{Step 2: Closeness to $\rho_\unique$.}
In the final step of our security proof, we show that $G_{(f,\pi)}$ is statistically $\qclass$-secure for three instantiations of $\qclass$ by showing that the output of $G_{(f,\pi)}$ is close to $\rho_\unique$ in each case. \\
\ \\
\noindent \underline{\textsc{Distinct Types:}}
By $\cref{eq:1}$, it follows that $G_{(f,\pi)}$ is $\qclass$-secure for\footnote{The reader may observe that we can also consider the convex closure of $\tuni{n}{s}{t}$.} $\qclass = \tuni{n}{s}{t}$. We can generalize this to \emph{distinct} type states $\ket{\psi_{T_1,\dots,T_s}}$, which are defined by the condition that the computational basis states spanned by the types $T_i$ are mutually disjoint, denoted by $(T_1, \dots, T_s) \in \tdis{n}{s}{t}$. Note that $\tuni{n}{s}{t} \subset \tdis{n}{s}{t}$ since for types $(T_1,\dots,T_s) \in \tdis{n}{s}{t}$ each $T_j$ may contain repetitions. Fortunately, a careful analysis shows that the output of $G_{(f,\pi)}$ on a distinct type state acquires a nice form and is close to $\rho_\unique$ as well. Intuitively, the reason for this is that the first step in our construction appends a random string $\vec{a}$ to the input query, and after this step the internal collisions in $\tdis{n}{s}{t}$ get eliminated except with negligible weight. Accordingly, we get security for the query set $$\qclass_{\distinct_{t,s}} = \bracC{\bigotimes_{i=1}^{s}\ketbra{\type_{T_i}}{\type_{T_i}}:(T_1,\cdots,T_s)\in\tdis{n}{s}{t}}.$$

\noindent As a corollary, we conclude that our construction is secure against computational basis queries. \\

\noindent \underline{\textsc{Many Copies of an $n$-Qubit State:}}
Next, we show security for many copies of the same pure state, defined by the query set 
    $$\qclass_{{\sf Single}} = \left\{ {\color{red} \ket{\psi}^{\otimes t}} \otimes {\color{blue} \ket{\psi}^{\otimes t}} \ :\  \ket{\psi} \in {\cal S}(\C^{2^n}) \right\},$$
which allows for the adversary to keep $t$ copies of the state that are not fed into the PRI, with $\ell = q = t$. We can write the input state in the type-basis of the symmetric subspace as $$\ketbra{\psi}{\psi}^{\otimes t} = \sum_{T,T'}\alpha_{T,T'}\ketbra{\type_T}{\type_{T'}}.$$

Thanks to the algebraic structure of $\mathbb{Z}_p$, the terms with $T \ne T'$ vanish under the application of $G_{(f,\pi)}^{\otimes q}(\cdot)(G_{(f,\pi)}^{\dagger})^{\otimes q}$. The rest of the terms are approximately mapped to $\rho_\unique$ as we showed in $\qclass_{\distinct_{t,s}}$-security above (by taking $s = 1$). Hence, the result follows.

\ \\
\noindent \underline{\textsc{Haar States:}}
Finally, we consider the case when the query contains a collection of $s$ i.i.d. Haar states, with $t$ copies of each kept by the adversary and $t$ copies given as input to the PRI, i.e. the query set is
$$ \qclass_{{\sf Haar}} = \left\{\mathbb{E}_{\ket{\psi_1},\ldots,\ket{\psi_s} \leftarrow \Haar_{n} }\left[ {\color{red}\bigotimes_{i=1}^{s}\ketbra{\psi_i}{\psi_i}^{\otimes t}} \otimes {\color{blue} \bigotimes_{i=1}^{s}\ketbra{\psi_i}{\psi_i}^{\otimes t}} \right]\right\} .$$
Note that without the {\color{red} red part}, the security would simply follow by taking an expectation over unique types in \cref{eq:1}. Since the adversary will keep $t$ copies of each Haar state to herself, she holds an entangled register (purification) to the query register, hence we need to work more.
We first recall that the query $\rho_{{\sf Haar}} \in \qclass_{{\sf Haar}}$ is negligibly close to the uniform mixture of unique $s$-fold type states (for $2t$ copies). We combine this with the useful expression    
\begin{align}
\ketbra{\type_T}{\type_T} = \frac{1}{(2t)!}\sum_{\sigma\in S_{2t}}\sum_{\substack{\vec{v}\in[N]^{2t}\\ \type(\vec{v}) = T}}\ketbra{\vec{v}}{\sigma(\vec{v})}. \label{eq:2}
\end{align} 
to express the output as
$$\rho \propto \E_{\substack{(f,\pi)\\ T_1,\dots,T_s \\ (\vec{x_1},\cdots,\vec{x_s})\in (T_1,\cdots,T_s) \\ \sigma_1,\cdots,\sigma_s\in S_{2t}}} \left[\bigotimes_{i=1}^{s} \left(\left(I_{nt} \otimes \left(G_{(f,\pi)}\right)^{\otimes t}\right)\ketbra{\vec{x_i}}{\sigma_i(\vec{x_i})} \left(I_{nt}\otimes \left(G_{(f,\pi)}^{\dagger}\right)^{\otimes t}\right)\right)\right].$$

\noindent Above, due to the nice structure of $G_{(f,\pi)}$, the only terms that do not vanish are those with permutations $\sigma_i$ that act separately on the first and the last $n$ qubits, i.e. $\sigma_i(\vec{x_i}) = \sigma_i^1(\vec{x_i^1}) || \sigma_i^2(\vec{x_i^2})$ with $\sigma_i^b \in S_n, x_i^b \in \bit^n$. With this observation, and using \cref{eq:2} in reverse, we see that the $q$-fold application of $G_{(f,\pi)}$ effectively \emph{unentangles} the state, which was the only barrier against security. 

\ \\
\noindent The detailed proofs of this step for all three query sets can be found in \Cref{sec:closeness_to_rho_uni}.

\subsubsection{Applications.}
We discuss several applications of PRIs, giving an overview of \Cref{sec:applications}.
\paragraph{Multi-Copy Secure Encryption.}
As a first application, we achieve multi-copy secure public-key and private-key encryption for quantum messages. Multi-copy security is defined via a chosen-plaintext attack (CPA) with the modification that the CPA adversary gets polynomially many copies of the ciphertext in the security experiment. This modification only affects security in the quantum setting due to the no-cloning principle, with the ciphertexts being quantum states. We note that using $t$-designs one can achieve multi-copy security if the number of copies is fixed a-priori before the construction, whereas using PRI we can achieve it for \emph{arbitrary} polynomially many copies. Multi-copy security was independently studied by~\cite{LQSYZ23} albeit in the one-time setting.
\par We will focus on the public-key setting, for the private-key setting is similar. Formally, we would like an encryption scheme $(\setup, \enc, \dec)$ with the property that no QPT adversary, given $\rho^{\otimes t}$, where $\rho \leftarrow \enc(\ket{\psi_b})$, can distinguish the cases $b=0$ and $b=1$ with non-negligible advantage, for any quantum messages $\ket{\psi_0}, \ket{\psi_1}$. In the construction, we will use a post-quantum public-key encryption scheme $(\smsetup, \smenc, \smdec)$ and a secure pseudorandom isometry $\pri$. The public-secret keys are those generated by $\smsetup(1^\secparam)$. To encrypt a quantum message $\ket{\psi}$, we sample a PRI key $k$ and output $(\ct, \varphi)$, where $\ct$ is encryption of $k$ using $\smenc$, and $\varphi \leftarrow \pri_k(\ket{\psi})$. Note that for correctness we need the ability to efficiently invert the PRI, which is a property satisfied by our PRI construction.
\par To show security, we deploy a standard hybrid argument where we invoke the security of $(\smsetup, \smenc, \smdec)$ as well as the $\qclass_{{\sf Single}}$-security of $\pri$. This suffices since we only run $\pri$ on copies of the same pure-state input (the quantum message). 

\paragraph{Succinct Commitments.}
\cite{GJMZ23} showed how to achieve succinct quantum commitments using pseudorandom unitaries (PRU) by first achieving one-time secure quantum encryption, and then showing that one-time secure quantum encryption implies succinct commitments. We adapt their approach to achieve succinct quantum commitments from PRIs. \cite{LQSYZ23} uses the work of \cite{GJMZ23} in a similar fashion to achieve succinct commitments from quantum pseudorandom state scramblers.
\par To one-time encrypt a quantum message, we apply in order: (1) inverse Schur transform, (2) PRI, and (3) Schur transform. Note that in contrast with \cite{GJMZ23}, the Schur transforms in (1) and (3) have different dimensions. The security proof follows that of \cite{GJMZ23} closely and relies on Schur's Lemma.

\paragraph{Quantum MACs.}
We show how to achieve a restricted version of quantum message authentication codes (QMACs) using an invertible pseudorandom isometry $\pri$. We face definitional challenges in this task. 
\par Similar to an injective function, an isometry does not have a unique inverse\footnote{We remind the reader that the map $\haarisometry^\dagger$ is not a physical map (quantum channel) for a general isometry $\haarisometry$.}. We discuss this and give a natural definition of the inverse in \Cref{sec:invertibility}.
\par There is extensive literature~\cite{BCG+02FOCS,DNS12,GYZ17,AM17} on \emph{one-time}, private-key quantum state authentications, i.e., the honest parties can detect whether the signed quantum state has been tempered. However, defining \emph{many-time} security, such as existentially unforgeable security under a chosen-message attack, is quite challenging. In particular, defining QMACs is non-trivial for several reasons, explicitly pointed out by~\cite{AGM18}. Firstly, one needs to carefully define what constitutes a \emph{forgery}, and secondly, verification may require multiple copies of the message and/or the tag. We give a new syntax which differs from the classical setting in that the verification algorithm outputs a message instead of Accept/Reject. 
\par In our construction, the signing algorithm simply applies $\pri$ to the quantum message, whereas the verification applies the inverse of $\pri$. Given this syntax, we show that our construction satisfies three different security notions: \begin{itemize}
    \item In the first setting, the verification algorithm is run polynomially many times in parallel on fresh (message, tag) pairs, and the outputs of the verifier is compared with the message using a SWAP test. We argue that during a forgery, each swap test succeeds with constant probability, hence the forgery succeeds with exponentially small probability due to independent repetition of SWAP tests.
    \item In the second setting, the verification is run once on the tag, and the output is compared to polynomially many copies of the message using a generalized SWAP test called \emph{the permutation test} \cite{BBA+97,KNY08,GGH+15,BS20}. The upside of this security notion is that it requires only one copy of the tag, yet the downside is that the it yields inverse polynomial security rather than negligible security.
    \item In the third setting, the adversary is asked to output the description of an invertible quantum circuit that generates the forgery message on input $\ket{0^n}$, together with the tag. In this setting, the verification is run on the tag, and the inverse of the circuit is computed on the output to see if the outcome is $\ket{0^n}$. We show that negligible security in this setting follows as a direct consequence of PRI security.
\end{itemize}

Now we will describe the security proof for the first and the second settings. Firstly, we can replace the PRI with a Haar isometry $\haarisometry$ using PRI security. Next, suppose the adversary $\alice$ makes $q$ queries $\ket{\psi_1}, \dots, \ket{\psi_q}$ to the signing oracle, receiving tags $\ket{v_1}, \dots, \ket{v_q}$ in return. Let the forgery output by $\alice$ be $(\ket{\psi^*}, \ket{\phi^*})$. It is forced by definition that $\ket{\psi^*}$ is orthogonal to $V := \mathsf{span}(\ket{\psi_1}, \dots, \ket{\psi_q})$. From $\alice$'s point of view, $\haarisometry\ket{\psi^*}$ is a Haar-random state sampled from $V^\perp$. Therefore, any $\ket{\phi^*} \in V$ will be mapped to a state orthogonal to $\ket{\psi^*}$ by the verification, whereas a forgery satisfying $\ket{\phi^*} \in V^\perp$ is as good as any other such forgery. Putting these together, a straightforward calculation using the fact that $\dim V \le q \ll 2^\secparam$ suffices for the proof in both settings.

\paragraph{PRS Length Extension.}
We show how to generically extend the length of a Haar-random state using a small amount of randomness assuming the existence of PRIs. Formally, we show that if $\pri$ is a secure $(n,n+m)$-pseudorandom isometry, then given $t$ copies of a $2n$-qubit Haar-random state $\ket{\theta}$, the state $(I_n \otimes \pri_k)^{\otimes t} \ket{\theta}^{\otimes t}$, obtained by applying $\pri_k$ to the last $n$ qubits, is computationally indistinguishable from $t$ copies of a $(2n+m)$-qubit Haar-random state $\ket{\gamma}^{\otimes t}$.
\par In the proof, we can replace $\pri$ with a random isometry $\haarisometry$ up to negligible loss invoking security. After writing $\ket{\theta}\bra{\theta}^{\otimes t}$ as a uniform mixture of type states, we obtain the expression $$ \rho' = \Ex_{T, \haarisometry} \left[ (I_n \otimes \haarisometry)^{\otimes t} \ketbra{\type_T}{\type_T} (I_n \otimes \haarisometry^\dagger)^{\otimes t} \right], $$
where by a collision-bound we can assume (up to a negligible loss) that $T$ is sampled as a \emph{good} type, meaning if it contains strings $\{x_1||y_1 \dots x_t || y_t \}$, then $x_i \ne x_j$ and $y_i \ne y_j$ for $i \ne j$. For such good types $T$, we can show that the state $\rho'$ is close to the uniform mixture of type states $\ket{\type_{T'}}\bra{\type_{T'}}$ spanning states of the form $\ket{\vec{x}}\ket{\vec{z}}$, where $\vec{z} \in \bit^{(n+m)t}$ is a random vector with pairwise distinct coordinates. This is because the mapping $(I_n \otimes \haarisometry)^{\otimes t}$ \emph{scrambles} $\vec{y}$ and leaves $\vec{x}$ untouched. In the proof we use our (first) observation about how $t$-fold Haar unitary acts on orthogonal inputs. 

\par For technical reasons, our loss in this step is proportional to $t!$, which necessitates the assumption that $t$ must be sublinear in the security parameter (e.g. $t = \poly\log(\secparam)$. In more detail, we expand $\rho'$ by expressing the type state $
\ket{\type_T}$ as superposition of computational basis states pairwise related by a permutation to get \begin{align*}
    \rho' & = \frac{1}{t!} \sum_{\sigma,\pi\in S_t} \ketbra{\sigma(\Vec{x})}{\pi(\Vec{x})} \otimes \Ex_{\haarisometry}[ \haarisometry^{\otimes t} \ketbra{\sigma(\Vec{y})}{\pi(\Vec{y})} (\haarisometry^\dagger)^{\otimes t}] \\
    & = \frac{1}{t!} \sum_{\sigma,\pi\in S_t} \ketbra{\sigma(\Vec{x})}{\pi(\Vec{x})} \otimes P_\sigma \Ex_{\haarisometry}[ \haarisometry^{\otimes t} \ketbra{\Vec{y}}{\Vec{y}} (\haarisometry^\dagger)^{\otimes t}] P_\pi^\dagger,
\end{align*}
where we used the fact that the permutation operators $P_\sigma, P_\pi$ commute\footnote{Technically the permutation operator acts on a larger Hilbert space after applying the isometry, but it applies the same permutation to the order of $t$ copies.} with the $t$-fold isometry $\haarisometry^{\otimes t}$. We can show that the term between the permutation operators $P_\sigma, P_\pi^\dagger$ is maximally scrambled for any given $\sigma, \pi$, which can be combined with a union bound over $\sigma, \pi$ that yields a factor of $t!$ in the loss. Unfortunately we do not know how to relate the terms across different $\sigma, \pi$ to avoid this loss.
Finally, the uniform mixture we obtained is negligibly close to the distribution of $\ket{\gamma}^{\otimes t}$ by another collision-bound.

%% file: prelims.tex
\section{Preliminaries} \label{sec:prelim}
We denote the security parameter to be $\secparam$. We assume that the reader is familiar with the fundamentals of quantum computing covered in~\cite{nielsen_chuang_2010}. 
\par We define $\sphere(\C^{N})$ to be the set of $N$-dimensional vectors with unit norm. An element in  $\sphere(\C^{N})$ is denoted using the ket notation $\ket{\cdot}$. We use ${\cal D}(\C^{N})$ to denote the set of $N$-dimensional density matrices. Let $H_A,H_B$ be finite-dimensional Hilbert spaces, we use $\cL(H_A,H_B)$ to denote the set of all linear operators from $H_A$ to $H_B$. If $H_A \cong H_B$, then we write $\cL(H_A)$ instead of $\cL(H_A,H_B)$ for short. Sometimes we abuse the notation and denote a density matrix of the form $\ketbra{\psi}{\psi}$ to be $\ket{\psi}$. We denote the trace distance between quantum states $\rho, \rho'$ by $\TD(\rho, \rho') := \frac{1}{2} \|\rho - \rho\|_1$. We denote the operator norm of $A$ by $\norm{A}_\infty$.
\par We refer to Section 2.1 in~\cite{AQY21} for the definition of quantum polynomial-time (QPT) algorithms adopted in this work. 

\subsection{Notation}
\begin{itemize}
\item Let $n,p\in\N$, we use $[n]$ to denote the set $\{0,\ldots,n-1\}$. 
\item We denote by $S_{n}$ the symmetric group on $n$ elements. 
\item We denote by $\mathcal{F}_{n,p}$ the set of all functions from $[n]$ to $[p]$. 
\item For a set $A$ and $t \in \N$, we define $A^t := \{(a_1,\ldots,a_t)\ :\ \forall i, a_i \in A\}$. 
\item Let $n,m,t\in\N$, $\vec{x} = (x_1,\cdots,x_t) \in \{0,1\}^{nt}$, $\vec{y} = (y_1,\cdots,y_t)\in \{0,1\}^{mt}$, we define $\vec{x}||\vec{y}:= (x_1||y_1,\cdots,x_t||y_t)\in\{0,1\}^{(n+m)t}$.
\item Let $\sigma\in S_t$, we define $\sigma(\vec{x}) := (x_{\sigma^{-1}(1)},\cdots,x_{\sigma^{-1}(t)})\in\{0,1\}^{nt}$. 
\item Let $\pi\in S_{2^n}$, we define $\vec{x}_{\pi} := (\pi(x_{1}),\cdots,\pi(x_{t}))\in\{0,1\}^{nt}$.
\item Let $F:\bit^{n}\to\mathbb{Z}$, we define $F(\vec{x}):=\sum_{i=1}^{t} F(x_i)$.
\item Let $X_{AB} \in \cL(H_A\otimes H_B)$. By $\Tr_{B}(X_{AB})$ we mean the partial trace over $B$.
\end{itemize}

\subsection{Haar Measure, Symmetric Subspaces, and Type States}
\paragraph{Haar Unitaries, Haar States, and Haar Isometries.}
\begin{definition}[Haar Unitaries and Haar States]
We denote by $\overline{\Haar_n}$ the \emph{Haar measure} over $2^n\times 2^n$ unitaries. We call a $2^n\times 2^n$ unitary $U$ a \emph{Haar unitary} if $U\gets \overline{\Haar_n}$. Let $V$ be a finite-dimensional Hilbert space, we denote by $\Haar(V)$ the \emph{uniform spherical measure} on the unit sphere $\sphere(V)$. If $V\cong \C^{2^n}$, then we write $\Haar_n$ instead of $\Haar(\C^{2^n})$ for short. Moreover, $\Haar_n$ is equivalent to the distribution of $U\ket{0^n}$ induced by $U\leftarrow\overline{\Haar_n}$. We call a state $\ket{\vartheta}\in\pqstates{n}$ a \emph{Haar state} if $\ket{\vartheta}\gets \Haar_n$. We refer the readers to~\cite[Chapter~7]{WatrousBook} and \cite[Chapter~1]{Meckes19} for formal definitions.
\end{definition}

\begin{definition}[Haar Isometries] 
We call an isometry $\haarisometry: \C^{N} \to \C^{NM}$ a \emph{Haar isometry} if $\haarisometry\ket{x} = U\ket{x}\ket{\hat{0}}$, where $U: \C^{NM} \to \C^{NM}$ is a Haar unitary and $\ket{\hat{0}} \in \C^{M}$ is an arbitrary\footnote{Note that the choice of $\ket{\hat{0}}$ does not affect the distribution of $\haarisometry$ because $U$ is distributed according to the Haar distribution.} and fixed pure state. Equivalently, $\haarisometry$ is obtained by truncating an $NM\times NM$ Haar unitary to its first $N$ columns. We denote by $\overline{\Haar_{n,n+m}}$ the distribution of a Haar isometry from $n$ qubits to $n+m$ qubits. We refer the readers to~\cite{ZS00,KNPPZ21} for more details.
\end{definition}

\paragraph{An Explicit Geometric Construction of Haar Unitaries.}
According to~\cite[page~19]{Meckes19}, sampling a Haar unitary $U$ has a nice geometric interpretation. Intuitively, the procedure goes by ``uniformly'' sampling $U$ column-by-column conditioned on being orthogonal to all the previously sampled columns.

\begin{fact}[Sampling Haar Unitaries] \label{fact:sampling_Haar_unitary}
For any $d\in \N$, the following procedures output a $d\times d$ Haar unitary $U$.
\begin{enumerate}
    \item Let $V_0 := \set{0}$.
    \item For $i = 1,2,\dots,d$, samples $\ket{v_i} \gets \Haar(V_{i-1}^\perp)$ and let $V_i := \mathsf{span}\set{\ket{v_1},\ket{v_2},\dots,\ket{v_i}}\subseteq \C^d$.
    \item Output $U := \sum_{i=1}^d \ketbra{v_i}{i}$.
\end{enumerate}
\end{fact}

\noindent Similarly, a Haar random isometry from $\C^{d'}$ to $\C^d$ ($d'\leq d$) is identically distributed to running the above procedure right after $d'$ columns are sampled (or equivalently, truncating the last $d-d'$ columns of $U$~\cite{ZS00}).

\begin{fact}[Sampling Haar Isometries] \label{fact:sampling_Haar_isometry}
For any $d',d\in \N$ such that $d'\leq d$, the following procedures output a $d\times d'$ Haar isometry $\haarisometry$.
\begin{enumerate}
    \item Let $V_0 := \set{0}$.
    \item For $i = 1,2,\dots,d'$, samples $\ket{v_i} \gets \Haar(V_{i-1}^\perp)$ and let $V_i := \mathsf{span}\set{\ket{v_1},\ket{v_2},\dots,\ket{v_i}}\subseteq \C^d$.
    \item Output $\haarisometry := \sum_{i=1}^{d'} \ketbra{v_i}{i}$.
\end{enumerate}
\end{fact}
\noindent Hence, we can view PRIs as a relaxation of PRUs from the following perspective: By leveraging computational assumptions, PRIs approximate the \emph{marginal} distribution of the first few columns of a Haar unitary, whereas PRUs need to approximate the whole matrix.

\paragraph{Symmetric Subspace and Type States.}
\noindent The proof of facts and lemmas in the rest of this subsection can be found in~\cite{Harrow13church,mele2023introduction}. Let $v = (v_1,\ldots,v_t)\in [N]^t$ for some $N,t\in\N$, we define $\hamming(v) := \sum_{i=1}^t v_i$.
Define $\type(v)$ to be a vector in $[t+1]^N$ where the $i^{th}$ entry in $\type(v)$ denotes the frequency of $i$ in $v$. For each type vector $T\in [t+1]^{N}$, with $T=(t_1,\cdots,t_{N})$, we define $\freq{T}{i} := t_i$ and $\setT(T)$ to be the multiset of size $t+1$ containing $t_i$ copies of $i$ for $1\leq i\leq N$. We define the support of $T$ by $\supp(T):=\set{i\in[N]: \freq{T}{i}>0}$. We sometimes write $\vec{v}\in T$ to mean $\vec{v}\in[N]^t$ with $\type(\vec{v}) = T$. Similarly, we write $T'\subset T$ to mean $\setT(T')\subset \setT(T)$.
\begin{definition}[Type States]
\label{def:type_states}
    Let $T\in [t+1]^N$ with $\hamming(T) = t$ for some $N,t\in \N$, define the type state: $$\ket{\type_T} := \sqrt{\frac{\prod_{i\in\supp(T)} \freq{T}{i}!}{t!}}\sum_{\vec{v}\in T}\ket{\vec{v}}.$$ 
\end{definition} 

\begin{lemma}
\label{lem:type_struc}
    Let  $T\in [t+1]^N$ with $\hamming(T) = t$ for some $N,t\in \N$, then $$\ketbra{\type_T}{\type_T} = \frac{1}{t!}\sum_{\sigma\in S_t}\sum_{\substack{\vec{v}\in[N]^t:\\ \type(\vec{v}) = T}}\ketbra{\vec{v}}{\sigma(\vec{v})}.$$ 
\end{lemma}

\begin{proof}
Notice that by \Cref{def:type_states}, 
$$\ket{\type_T} = \sqrt{\frac{\prod_{i\in\supp(T)} \freq{T}{i}!}{t!}}\sum_{\substack{\vec{v}\in[N]^t:\\ \type(\vec{v}) = T}}\ket{\vec{v}}.$$
Hence, 
$$\ketbra{\type_T}{\type_T} = \frac{\prod_{i\in\supp(T)} \freq{T}{i}!}{t!}\sum_{\substack{\vec{v},\vec{v'}\in[N]^t:\\ \type(\vec{v}) = \type(\vec{v'}) = T}}\ketbra{\vec{v}}{\vec{v'}}.$$
Note that since $\type(\vec{v'}) = \type(\vec{v})$, $\vec{v'} = \sigma(\vec{v})$ for some $\sigma\in S_{t}$. Notice that the number of $\vec{v'}$ with $\type(\vec{v'}) = T$ is $\frac{t!}{\prod_{i\in \supp(T)} \freq{T}{i}!}$.  Summing over all permutations $\sigma$, each $\vec{v'}$ is repeated exactly $\prod_{i\in \supp{T}} \freq{T}{i}!$ times. 
Hence, 
$$\ketbra{\type_T}{\type_T} = \frac{1}{t!}\sum_{\sigma\in S_t}\sum_{\substack{\vec{v}\in[N]^t:\\ \type(\vec{v}) = T}}\ketbra{\vec{v}}{\sigma(\vec{v})}.$$
\end{proof}

\noindent Next, we define permutation operators and discuss a few properties of Haar states.

\begin{definition}[Permutation Operator]
\label{fact:sym}
Let $N,t\in\N$. For any permutation $\sigma \in S_t$, let $P_{N}(\sigma)$ denote the unitary that permutes the $t$ tensor factors according to $\sigma$, i.e., $P_N(\sigma) := \sum_{\vec{x}\in[N]^t}\ketbra{\sigma(\vec{x})}{\vec{x}} \in \cL((\C^N)^{\otimes t})$. When the dimension $N$ is clear from the context, we sometimes omit it and write $P_\sigma$ for brevity.
\end{definition}

\begin{definition}[Symmetric Subspace] \label{def:sym_subspace} Let $t \in \N$ and $H$ be a finite-dimensional Hilbert space. The symmetric subspace $\vee^t H\subseteq H^{\otimes t}$ is defined as \begin{align*}
    \vee^t H := \span \bracC{\ket{\psi}^{\otimes t} \; : \; \ket{\psi} \in H},
\end{align*}
and the orthogonal projector onto $\vee^t H$ is denoted by $\Pi^{H,t}_\sym$. In particular, if $H \cong \C^N$, then we write $\Pi^{N,t}_\sym$ rather than $\Pi^{\C^N,t}_\sym$ for brevity.
\end{definition}

\begin{fact}[Dimension of symmetric subspace]
\label{fact:dim-sym}
For $N,t\in\N$, $\dim(\vee^t\C^N) = \Tr(\Pi^{N,t}_\sym) = \binom{N+t-1}{t}$.
\end{fact}

\begin{fact}[Average of $t$-copies Haar states]
\label{fact:avg-haar-random}
Let $H$ be a finite-dimensional Hilbert space and $N :=\dim(H)$. For all $t \in \N$,
\[
    \E_{\ket{\vartheta} \leftarrow \Haar(H)} \ketbra{\vartheta}{\vartheta}^{\otimes t} 
    = \frac{\Pi^{H,t}_\sym}{\Tr(\Pi^{H,t}_\sym)}
    = \E_{\substack{T\leftarrow [t+1]^N\\ \hamming(T)=t}} \ketbra{\type_T}{\type_T} = \Ex_{\sigma\gets S_t} \left[ P_{N}(\sigma) \right].
\]
\end{fact}

\begin{fact}[Projection onto symmetric subspace stabilizes type states] \label{fact:sym_stab_type}
For all $N,t \in \N$ and $T\in [t+1]^N$ such that $\hamming(T) = t$,
\begin{align*}
    \Pi^{N,t}_\sym \ket{\type_T} = \ket{\type_T}.
\end{align*}
\end{fact}

\begin{fact}[Average inner product with Haar states] \label{fact:average_inner_product}
For any $N\in\N$ and fixed $\ket{\psi}\in\cS(\C^N)$, $\Ex_{\ket{\vartheta}\gets\Haar(\C^N)}[|\inner{\psi}{\vartheta}|^2] = 1/N$.
\end{fact}

\begin{fact} \label{fact:random_type_halves_collision}
Let $T$ be sampled uniformly from $[t+1]^{2^{\ell+k}}$ conditioned on $\hamming(T) = t$, where $\setT(T) = \set{x_1||y_1,x_2||y_2,\dots,x_t||y_t}$ and $x_i\in\bit^\ell$, $y_j\in\bit^k$. Then
$\Pr[\exists i\neq j\ s.t.\ x_i = x_j \lor y_i = y_j ] = O(t^2/2^\ell) + O(t^2/2^k)$.
\end{fact}

\noindent

\begin{lemma}[{\cite[Theorem~7.5]{WatrousBook}}, restated] \label{lem:multi_span_sym}
For all $N,t \in \N$, there exists a finite set $A\subseteq \sphere({\C^N})$ such that $\vee^t\C^N = \mathsf{span}\set{\ket{\psi}^{\otimes t}: \ket{\psi}\in A}$.
\end{lemma}

\subsection{Pseudorandom Primitives} \label{app:pseudo_defs}

\noindent We recall existing post-quantum secure pseudorandom primitives as well as quantum pseudorandom primitives. 

\paragraph{Pseudorandom Functions.}

\begin{definition}[Quantum-Query Secure Pseudorandom Functions]
\label{def:qqprfs}
We say that a deterministic polynomial-time algorithm $F:\{0,1\}^{\ell(\secparam)} \times \{0,1\}^{d(\secparam)} \rightarrow \{0,1\}^{n(\secparam)}$ is a {\em quantum-query $\varepsilon$-secure pseudorandom function (QPRF)} if for all QPT (non-uniform) distinguishers $A=(A_{\secparam},\rho_{\secparam})$ there exists a function $\eps(\cdot)$ such that the following holds:
 \[
    \left|\Pr_{k \leftarrow \{0,1\}^{\ell(\secparam)} }\left[A_\lambda^{\ket{{\cal O}_{\sf prf}(k,\cdot)}}(\rho_\lambda) = 1\right] - \Pr_{{\cal O}_{\sf Rand}}\left[A_\lambda^{\ket{{\cal O}_{\sf Rand}(\cdot)}}(\rho_\lambda) = 1\right]\right| \le \eps(\lambda),
  \]
  where:
\begin{itemize}
    \item ${\cal O}_{\sf prf}(k,\cdot)$ on input a $(d+n)$-qubit state on registers ${\bf X}$ (first $d$ qubits) and ${\bf Y}$, applies an $(n+d)$-qubit unitary $U$ described as follows: $U\ket{x}\ket{a} = \ket{x}\ket{a \oplus F(k,x)}$. It sends back the registers ${\bf X}$ and ${\bf Y}$. 
    \item ${\cal O}_{\sf Rand}(\cdot)$ on input a $(d+n)$-qubit state on registers ${\bf X}$ (first $d$ qubits) and ${\bf Y}$, applies an $(n+d)$-qubit unitary $R$ described as follows: $R \ket{x}\ket{a} = \ket{x}\ket{a \oplus y_x}$, where $y_x \leftarrow \{0,1\}^{n(\secparam)}$. It sends back the registers ${\bf X}$ and ${\bf Y}$.   
\end{itemize}
We denote the fact that $A_{\secparam}$ has quantum access to an oracle ${\cal O}$ by $A_{\secparam}^{\ket{{\cal O}}}$. 
\par We also say that $F$ is an $(\ell(\secparam),d(\lambda),n(\lambda),\eps)$-$\qprf$ to succinctly indicate that its input length is $d(\lambda)$ and its output length is $n(\lambda)$. When $\ell(\secparam)=\secparam$, we drop $\ell(\secparam)$ from the notation. Similarly, when $\eps(\secparam)$ can be any negligible function, we drop $\eps(\secparam)$ from the notation.    
\end{definition}

\noindent Zhandry~\cite{Zha12} showed how to instantiate quantum-query secure pseudorandom functions from post-quantum one-way functions.  

\paragraph{Pseudorandom Permutations.} 

\begin{definition}[Quantum-Query Secure Pseudorandom Permutation]
\label{def:qqprps}
We say that a deterministic polynomial-time algorithm $F:\{0,1\}^{\ell(\secparam)} \times \{0,1\}^{n(\secparam)} \rightarrow \{0,1\}^{n(\secparam)}$ is a {\em quantum-query $\varepsilon$-secure pseudorandom permutation (QPRP)} if for all QPT (non-uniform) distinguishers $A=(A_{\secparam},\rho_{\secparam})$ there exists a function $\eps(\cdot)$ such that the following holds:
\begin{multline*}
        \left|\Pr_{k \leftarrow \{0,1\}^{\ell(\secparam)} }\left[A_\lambda^{\ket{{\cal O}_{\sf prp}(k,\cdot)},\ket{{\cal O}_{{\sf prp}^{-1}}(k,\cdot)}}(\rho_\lambda) = 1\right]\right. - \left.\Pr_{g \xleftarrow{\$} S_{2^{n(\lambda)}}}\left[A_\lambda^{\ket{{\cal O}_{g}(\cdot)},\ket{{\cal O}_{{g}^{-1}}(\cdot)}}(\rho_\lambda) = 1\right]\right| \le \eps(\lambda),
\end{multline*}
  where:
\begin{itemize}
    \item ${\cal O}_{\sf prp}(k,\cdot)$ on input a $(2n)$-qubit state on registers ${\bf X}$ (first $n$ qubits) and ${\bf Y}$, applies an $(2n)$-qubit unitary $U$ described as follows: $U\ket{x}\ket{a} = \ket{x}\ket{a \oplus F(k,x)}$. It sends back the registers ${\bf X}$ and ${\bf Y}$. 
    \item ${\cal O}_{{\sf prp}^{-1}}(k,\cdot)$ on input a $(2n)$-qubit state on registers ${\bf X}$ (first $n$ qubits) and ${\bf Y}$, applies an $(2n)$-qubit unitary $U$ described as follows: $U\ket{x}\ket{a} = \ket{x}\ket{a \oplus F^{-1}(k,x)}$. It sends back the registers ${\bf X}$ and ${\bf Y}$. 
    \item ${\cal O}_{g}(\cdot)$ on input a $(2n)$-qubit state on registers ${\bf X}$ (first $d$ qubits) and ${\bf Y}$, applies an $(n+d)$-qubit unitary $R$ described as follows: $R \ket{x}\ket{a} = \ket{x}\ket{a \oplus g(x)}$. It sends back the registers ${\bf X}$ and ${\bf Y}$.
    \item ${\cal O}_{{g}^{-1}}(\cdot)$ on input a $(2n)$-qubit state on registers ${\bf X}$ (first $d$ qubits) and ${\bf Y}$, applies an $(n+d)$-qubit unitary $R$ described as follows: $R \ket{x}\ket{a} = \ket{x}\ket{a \oplus g^{-1}(x)}$. It sends back the registers ${\bf X}$ and ${\bf Y}$.   
\end{itemize}
\par We also say that $F$ is an $(\ell(\secparam),n(\lambda),\eps)$-$\qprp$ to succinctly indicate that its input and output length is $n(\lambda)$. When $\ell(\secparam)=\secparam$, we drop $\ell(\secparam)$ from the notation. Similarly, when $\eps(\secparam)$ can be any negligible function, we drop $\eps(\secparam)$ from the notation.    
\end{definition}

\noindent Zhandry~\cite{Zha16} showed how to instantiate quantum-query secure pseudorandom permutations from post-quantum one-way functions. Moreover, Zhandry~\cite{Zha12} showed that no algorithm making $q$ queries can distinguish between a random function and a $2q$-wise independent function.
\begin{theorem}[\cite{Zha12}]
    \label{thm:zha12}
    Let $A$ be a quantum algorithm making $q$ quantum queries to an oracle $H:X\to Y$. If we draw $H$ from uniformly random functions from $X$ to $Y$ versus if we draw $H$ uniformly from $2q$-wise independent functions, then for every $z$, the quantity $\Pr_{H}[A^{H}() = z]$ is the same for both the cases.
\end{theorem}

\paragraph{Pseudorandom State Generators (PRSGs).} 
\begin{definition}[PRS Generator]
\label{def:vanilla-prs}
We say that a QPT algorithm $F$ is a \emph{pseudorandom state (PRS) generator} if the following holds. 
\begin{enumerate}
    
    \item \textbf{State Generation}. For all $\lambda$ and for all $k \in \{0,1\}^\lambda$, the algorithm $F$ behaves as
    \[
        F_\lambda(k) = \rho_k.
    \]
    for some $n(\lambda)$-qubit (possibly mixed) state $\rho_k$. 
    
    \item \textbf{Pseudorandomness}. For all polynomials $t(\cdot)$ and (non-uniform) QPT distinguisher $A$ there exists a negligible function $\eps(\cdot)$ such that for all $\lambda$,  we have
    \[
        \left | \Pr_{k \leftarrow \{0,1\}^\lambda} \left [ A_\lambda (F_\lambda(k)^{\otimes t(\lambda)}) = 1 \right] - \Pr_{\ket{\vartheta} \leftarrow \Haar_{n(\lambda)}} \left [ A_\lambda (\ket{\vartheta}^{\otimes t(\lambda)}) = 1 \right] \right | \leq \eps(\lambda)~.
    \]
\end{enumerate}
We also say that $F$ is a $n(\lambda)$-PRS generator to succinctly indicate that the output length of $F$ is $n(\lambda)$.
\end{definition}

\noindent Ji, Liu and Song~\cite{JLS18} and Brakerski and Shmueli~\cite{BS20} presented instantiatiations of PRSGs from post-quantum secure one-way functions. 
 
\paragraph{Pseudorandom Function-Like State Generators.}

\begin{definition}[Selectively Secure PRFS Generator]
\label{def:prfs}
We say that a QPT algorithm $F$ is a (selectively secure) \emph{pseudorandom function-like state (PRFS) generator} if for all polynomials $s(\cdot), t(\cdot)$, QPT (nonuniform) distinguishers $A$ and a family of indices $\left(\{x_1,\ldots, x_{s(\lambda)}\} \subseteq \{0,1\}^{d(\lambda)}\right)_\lambda$, there exists a negligible function $\eps(\cdot)$ such that for all $\lambda$,
    \begin{multline*}
        \Big | \Pr_{k \leftarrow \{0,1\}^\lambda} \left [ A_\lambda( x_1,\ldots,x_{s(\lambda)},F_\lambda(k,x_1)^{\otimes t(\lambda)},\ldots, F_\lambda(k,x_{s(\lambda)})^{\otimes t(\lambda)}) = 1 \right] \\
        - \Pr_{\ket{\vartheta_1}, \ldots,\ket{\vartheta_{s(\lambda)}} \leftarrow \Haar_{n(\lambda)}} \left [ A_\lambda( x_1,\ldots,x_{s(\lambda)}, \ket{\vartheta_1}^{\otimes t(\lambda)},\ldots, \ket{\vartheta_{s(\lambda)}}^{\otimes t(\lambda)}) = 1 \right] \Big | \leq \eps(\lambda)~.
    \end{multline*}
We say that $F$ is a $(d(\lambda),n(\lambda))$-PRFS generator to succinctly indicate that its input length is $d(\lambda)$ and its output length is $n(\lambda)$.
\end{definition}

\noindent Ananth, Qian, and Yuen~\cite{AQY21} presented instantiations of PRFSGs either assuming post-quantum secure one-way functions or PRSGs (in the setting when the input length was logarithmic).  

%% file: definition.tex
\section{Pseudorandom Isometry: Definition}    \label{sec:defs}
\noindent For a given class of inputs $\qclass$, we propose the following definition of $\qclass$-secure psuedorandom isometries. Throughout the rest of the paper, for a polynomial $p(\cdot)$, we denote $p$ to be $p(\secparam)$, where $\secparam$ is the security parameter.

\begin{definition}[$\qclass$-Secure Pseudorandom Isometry (PRI)]
\label{def:Q-pri}
Let $n,m,q,\ell$ be polynomials in $\secparam$. Suppose $\qclass=\{\qclass_{n,q,\ell,\secparam}\}_{\secparam \in \mathbb{N}}$, where $\qclass_{n,q,\ell,\secparam} \subseteq {\cal D}(\C^{2^{nq+\ell}})$. 
We say that $\pri=\left\{F_{\secparam} \right\}_{\secparam \in \mathbb{N}}$ is an $(n,n+m)$-$\qclass$-secure pseudorandom isometry if the following holds: 
\begin{itemize}
    \item For every $k \in \{0,1\}^{\secparam}$, $F_{\secparam}(k,\cdot)$ is a QPT algorithm implementing a quantum channel such that it is functionally equivalent to ${\cal I}_k$, where ${\cal I}_k$ is an isometry that maps $n$ qubits to $n+m$ qubits. 
    \item For sufficiently large $\secparam \in \mathbb{N}$, any QPT distinguisher $\adversary$, the following holds: for every $\rho \in \qclass_{n,q,\ell,\secparam}$,
    $$\left| \Pr\left[ \adversary\left((I_{\ell} \otimes F_k^{\otimes q}) \left( \rho \right) \right) = 1 \right] 
    - \Pr \left[ \adversary\left((I_{\ell} \otimes \haarisometry^{\otimes q}) \left(\rho\right) \right) = 1 \right]  \right| 
    \leq \negl(\secparam),$$
    where:
    \begin{itemize}
        \item $\haarisometry(\cdot)$ is the channel implementing a Haar-random isometry that takes an $n$-qubit input $\ket{\psi}$ and outputs an $(n+m)$-qubit output $\haarisometry(\ket{\psi})$,
        \item $I_{\ell}$ is an identity operator on $\ell$ qubits. 
    \end{itemize}
\end{itemize}
\end{definition}



\noindent We sometimes write $\qclass$-secure with $m,n$ being implicit. We consider the following set of queries. We color the part of the query given to $I_\ell$ with {\color{red} red} and color the part of the query given to $F_k$ or $\haarisometry$ with {\color{blue} blue}.

\paragraph{Computational basis queries.} We define $\qclass^{(\comp)}_{n,q,\ell,\secparam}$ as follows. 
$$\qclass_{n,q,\ell,\secparam}^{(\comp)} = {\color{red}{\cal D}(\C^{2^{\ell}})} \otimes \left\{ {\color{blue}\left(\ketbra{x_1}{x_1}\otimes\ldots\otimes\ketbra{x_q}{x_q}\right)}\ :\ x_1,\ldots,x_q \in \{0,1\}^n \right\}.$$

\noindent Let $n(\cdot),q(\cdot),\ell(\cdot)$ be polynomials. We also define $\qclass_{\comp}$ (implicitly parameterized by $n(\cdot),q(\cdot),\ell(\cdot)$) to be $\qclass_{\comp}=\left\{ \qclass_{n,q,\ell,\secparam}^{(\comp)} \right\}_{\secparam \in \mathbb{N}}$.

\paragraph{Multiple copies of a single pure state.} We define $\qclass_{n,q,\ell,\secparam}^{(\single)}$ as follows: $$\qclass_{n,q,\ell,\secparam}^{(\single)} = {\color{red}{\cal D}(\C^{2^{\ell}})} \otimes \left\{ {\color{blue}\left(\ketbra{\psi}{\psi}^{\otimes q}\right)}\ : \ket{\psi}\text{ is an  $n$-qubit pure state} \right\}.$$

\noindent Let $n(\cdot),q(\cdot),\ell(\cdot)$ be polynomials. We also define $\qclass_{\single}$ (implicitly parameterized by $n(\cdot),q(\cdot),\ell(\cdot)$) to be $\qclass_{\single}=\left\{ \qclass_{n,q,\ell,\secparam}^{(\single)} \right\}_{\secparam \in \mathbb{N}}$. 

\paragraph{Haar queries.} We first define $\qclass^{({\sf Haar})}_{n,s,t,\ell',\secparam}$ as follows, for some polynomials $s(\cdot),t(\cdot),\ell'(\cdot)$,
$$\qclass_{n,s,t,\ell',\secparam}^{({\sf Haar})} = {\color{red}{\cal D}(\C^{2^{\ell'(\secparam)}})} \otimes \left\{\mathbb{E}_{\ket{\psi_1},\ldots,\ket{\psi_{s(\secparam)}} \leftarrow \Haar_{n} }\left[ {\color{red}\bigotimes_{i=1}^{s(\secparam)}\ketbra{\psi_i}{\psi_i}^{\otimes t(\secparam)}} \otimes {\color{blue} \bigotimes_{i=1}^{s(\secparam)}\ketbra{\psi_i}{\psi_i}^{\otimes t(\secparam)}} \right]\right\}.$$
Next, we define $\qclass^{({\sf Haar})}_{n,q,\ell,\secparam}$ as follows
$$\qclass_{n,q,\ell,\secparam}^{({\sf Haar})} = \bigcup_{\substack{s,t,\ell'\\ \text{such that }q = st\\\text{and } \ell = \ell'+st}} \qclass_{n,s,t,\ell',\secparam}^{({\sf Haar})}.$$

\noindent Let $n(\cdot),q(\cdot),\ell(\cdot)$ be polynomials. We also define $\qclass_{{\sf Haar}}$ (implicitly parameterized by $n(\cdot),q(\cdot),\ell(\cdot)$) to be $\qclass_{{\sf Haar}}=\left\{ \qclass_{n,q,\ell,\secparam}^{({\sf Haar})} \right\}_{\secparam \in \mathbb{N}}$.

\paragraph{Selective PRI.} Above, we considered the security of PRI in the setting where the queries came from a specific query set. However, we can consider an alternate definition where the indistinguishability holds against computationally bounded adversaries making a single parallel query to an oracle that is either PRI or Haar. We term such a PRI to be a selectively secure PRI.

\begin{definition}[Selective Pseudorandom Isometry]
\label{def:pri}
$\pri=\left\{F_{\secparam} \right\}_{\secparam \in \mathbb{N}}$ is an $(n,n+m)$-selective pseudorandom isometry if the following holds:
\begin{itemize}
    \item For every $k \in \{0,1\}^{\secparam}$, $F_{\secparam}(k,\cdot)$ is a QPT algorithm such that it is functionally equivalent to ${\haarisometry}_k$, where ${\haarisometry}_k$ is an isometry that maps $n$ qubits to $n+m$ qubits. 
    \item For sufficiently large $\secparam \in \mathbb{N}$, for any $q=\poly(\secparam)$, any QPT distinguisher $\adversary$ making $1$ query to the oracle, the following holds:
    $$\left| \Pr\left[ \adversary^{(F_{\secparam}(k,\cdot))^{\otimes q}} = 1 \right] - \prob \left[ \adversary^{(\haarisometry(\cdot))^{\otimes q}} = 1 \right]  \right| \leq \negl(\secparam),$$
    where:
    \begin{itemize}
        \item $F_{\secparam}(k,\cdot)$ takes as input $\ket{\psi}$ and outputs  $F_{\secparam}(k,\ket{\psi})$
        \item $\haarisometry(\cdot)$ is a  Haar-random isometry that takes as $n$-qubit input $\ket{\psi}$ and outputs an $(n+m)$-qubit output $\haarisometry(\ket{\psi})$. 
    \end{itemize}
\end{itemize}
\end{definition}

\noindent The following claim is immediate. 

\begin{claim}
Let $n(\cdot),m(\cdot)$ be two polynomials. Suppose $\pri$ is an $(n,n+m)$-$\qclass_{n,q,\ell}$-secure pseudorandom isometry for every polynomial $q(\cdot),\ell(\cdot)$, and, $\qclass_{n,q,\ell}=\left\{ \qclass_{n,q,\ell,\secparam} \right\}_{\secp\in\N}$, where $\qclass_{n,q,\ell,\secparam} = {\cal D}(\C^{2^{nq+\ell}})$. Then, $\pri$ is a selective pseudorandom isometry. 
\end{claim} 

\noindent Similarly, the other direction is true as well. 

\begin{claim}
Let $n(\cdot),m(\cdot)$ be two polynomials. Suppose $\pri$ is an $(n,n+m)$-secure pseudorandom isometry. Then $\pri$ is a $(n,n+m)$-$\qclass_{n,q,\ell}$-secure pseudorandom isometry for every polynomial $q(\cdot),\ell(\cdot)$, and, $\qclass_{n,q,\ell}=\left\{ \qclass_{n,q,\ell,\secparam} \right\}_{\secp\in\N}$, where $\qclass_{n,q,\ell,\secparam} = {\cal D}(\C^{2^{nq+\ell}})$.
\end{claim} 

\paragraph{Adapive PRI.} We also define an adaptive version of the pseudorandom isometries below. In this definition, the adversary can make an arbitrary number of queries to the oracle. 

\begin{definition}[Adaptive Pseudorandom Isometry]
\label{def:pri_a}
$\pri=\left\{F_{\secparam} \right\}_{\secparam \in \mathbb{N}}$ is an $(n,n+m)$-adaptive pseudorandom isometry if the following holds:
\begin{itemize}
    \item For every $k \in \{0,1\}^{\secparam}$, $F_{\secparam}(k,\cdot)$ is a QPT algorithm such that it is functionally equivalent to ${\haarisometry}_k$, where ${\haarisometry}_k$ is an isometry that maps $n$ qubits to $n+m$ qubits. 
    \item For sufficiently large $\secparam \in \mathbb{N}$, for any $t=\poly(\secparam)$, any QPT distinguisher $\adversary$ making $t$ queries to the oracle, the following holds:
    $$\left| \Pr\left[ \adversary^{F_{\secparam}(k,\cdot)} = 1 \right] 
    - \Pr \left[ \adversary^{\haarisometry(\cdot)} = 1 \right] \right| 
    \leq \negl(\secparam),$$
    where:
    \begin{itemize}
        \item $F_{\secparam}(k,\cdot)$ takes as input $\ket{\psi}$ and outputs  $F_{\secparam}(k,\ket{\psi})$
        \item $\haarisometry(\cdot)$ is a  Haar-random isometry that takes as $n$-qubit input $\ket{\psi}$ and outputs an $(n+m)$-qubit output $\haarisometry(\ket{\psi})$. 
    \end{itemize}
\end{itemize}
\end{definition}

\paragraph{Observations.} It should be immediate that pseudorandom unitaries, introduced in~\cite{JLS18}, imply adaptive PRI, which in turn implies selectively secure PRI. Whether pseudorandom isometries are separated from pseudorandom unitaries or there is a transformation from the former to the latter is an interesting direction to explore. 
\par If we weaken our definition of pseudorandom isometries further, where we a priori fix the number of queries made by the adversary and allow the description of the pseudorandom isometry to depend on this then this notion is implied by unitary $t$-designs~\cite{AE07CCC,BHH16}.
\par In terms of implications of pseudorandom isometries to other notions of pseudorandomness in the quantum world, we note that pseudorandom isometries imply both PRSGs and PRFSGs (see~\Cref{app:pseudo_defs} for formal definitions and~\Cref{sec:PRIimplyPRG_PRFSG} for the proof.).

\subsection{Invertibility} \label{sec:invertibility}
\paragraph{Invertible Pseudorandom Isometries.} In applications, we need a stronger notion of \emph{invertible} pseudorandom isometries.
\begin{definition}[Invertible $\qclass$-Secure Pseudorandom Isometry] \label{def:invertible}
We say that $\pri=\left\{F_{\secparam} \right\}_{\secparam \in \mathbb{N}}$ is an invertible $(n,n+m)$-$\qclass$-secure pseudorandom isometry if first and foremost, it is a $\qclass$-secure pseudorandom isometry (\Cref{def:Q-pri}) and secondly, there is a QPT algorithm $\inv$ with the following guarantee: for every $\ket{\psi}\in\pqstates{n}$ and $k \in \bit^\secp$,
$$\TD( \ketbra{\psi}{\psi}, \inv\left(k,F_{\secparam}(k,\ket{\psi}) \right) ) = \negl(\secp).$$
\end{definition}

\begin{remark}
Similarly, we can define invertible versions of $\qclass$-secure PRIs and selectively secure PRIs. Also, note that for $\ket{\phi}$ which is orthogonal to the range of $F_{\secparam}(k,\cdot)$, being invertible has no guarantee on $\inv(k,\ket{\phi})$. 
\end{remark}

\paragraph{Inverse of Isometries.}
For a (fixed) isometry $\cI$ maps $n$-qubit states to $(n+m)$-qubit states, the ``inverse'' of $\cI$ is not unique. However, under the view of \emph{Stinespring dilation}, it is possible to naturally define a quantum channel $\cI^{-1}$ such that $\cI^{-1}\circ (\cI \ketbra{\psi}{\psi} \cI^\dagger) = \ketbra{\psi}{\psi}$ for every $\ket{\psi}\in\pqstates{n}$.\footnote{The readers should not confuse $\cI^\dagger$, the conjugate transpose of $\cI$, with the channel $\cI^{-1}$.} Consider an arbitrary unitary $U_\cI$ on $n+m$ qubits such that $U_\cI$ is \emph{consistent} with $\cI$, that is, $U_\cI\ket{\psi}\ket{0^m}_\aux = \cI\ket{\psi}$ for every $\ket{\psi}\in\pqstates{n}$. One can easily verify that $\Tr_\aux \left( U_\cI^\dagger \cI \ketbra{\psi}{\psi} \cI^\dagger U_\cI \right) = \ketbra{\psi}{\psi}$ for every $\ket{\psi}\in\pqstates{n}$. Furthermore, one can even provide a distribution over such unitaries. This yields the following candidate definition: let $\mu_\cI$ be some distribution over unitaries that are consistent with $\cI$, the inverse of $\cI$ can be defined as
\begin{align*}
\cI^{-1}(X) = \Ex_{U_\cI\gets\mu_\cI} \Tr_\aux \left( U_\cI^\dagger X U_\cI \right).
\end{align*}
Since we focus on Haar isometries in this work, we'll choose the distribution $\mu_\cI$ to be Haar random conditioned on being consistent with $\cI$. Formally, we have the following definition.
\begin{definition}[Inverse of Isometries] \label{def:inverse_isometry}
Let $\cI$ be an isometry from $n$ qubits to $n+m$ qubits. The inverse of $\cI$ is a quantum channel from $n+m$ qubits to $n$ qubits defined to be
\begin{align*}
    \cI^{-1}(X) := \Ex_{U\gets\overline{\Haar_{n+m}}\mid_\cI} \Tr_\aux \left( U^\dagger X U \right),
\end{align*}
for any $X\in\cL(\C^{2^{n+m}})$, where register $\aux$ refers to the last $m$ qubits and $\overline{\Haar_{n+m}}\mid_\cI$ denotes the Haar measure over $(n+m)$-qubit unitaries $U$ conditioned on $U\ket{\psi}\ket{0^m}_\aux = \cI\ket{\psi}$ for any $\ket{\psi}\in\pqstates{n}$.
\end{definition}

\noindent From~\Cref{fact:sampling_Haar_unitary}, sampling $U$ according to $\overline{\Haar_{n+m}}\mid_\cI$ is equivalent to the following: Fix $\cI$ and then keep appending columns one-by-one by sampling a uniform unit vector conditioned on being orthogonal to the existing columns until the matrix is square. Therefore, by~\Cref{fact:sampling_Haar_isometry}, the inverse of a Haar isometry satisfies the following:
\begin{fact} \label{fact:inverse_of_Haar}
Let $\haarisometry$ be a Haar isometry from $n$ qubits to $n+m$ qubits. Then the joint distribution of $(\haarisometry,\haarisometry^{-1})$ is identically distributed to the following procedures: (1) Sample $U\gets\Haar_{n+m}$. (2) Define $\haarisometry$ to be the first $2^n$ columns of $U$. That is, $\cI$ satisfies $\haarisometry\ket{\psi} = U\ket{\psi}\ket{0^m}_\aux$ for any $\ket{\psi}\in\pqstates{n}$. (3) Define $\haarisometry^{-1}(X) := \Tr_\aux(U^\dagger XU)$.
\end{fact}

\paragraph{Strong Invertible Adaptive PRI.}
In order to achieve more applications, we define the following stronger security definition in which the adversary is given the inversion oracle.
\begin{definition}[Strong Invertible Adaptive Pseudorandom Isometry] \label{def:strong_inv_PRI}
$\pri=\set{F_{\secp}}_{\secp \in \N}$ is a strong invertible $(n,n+m)$-pseudorandom isometry if it satisfies the following conditions for every $\secp\in\N$:
\begin{itemize}
    \item For every $k \in \bit^\secp$, $F(k,\cdot)$ is a QPT algorithm such that it is functionally equivalent to ${\haarisometry}_k$, where ${\haarisometry}_k$ is an isometry that maps $n$ qubits to $n+m$ qubits.
    \item For every $k \in \bit^\secp$, $\inv(k,\cdot)$ is a QPT algorithm such that it is functionally equivalent to $\haarisometry^{-1}_k$, where $\haarisometry^{-1}_k$ is the inverse of ${\haarisometry}_k$ (\Cref{def:inverse_isometry}) that maps $n+m$ qubits to $n$ qubits.
    \item For any polynomial $t = \poly(\secp)$, any QPT distinguisher $\adversary$ making a total of $t$ queries to the oracles, the following holds:
    $$\left| \Pr_{k\gets\bit^\secp} \left[ \adversary^{F(k,\cdot),\inv(k,\cdot)} = 1 \right] 
    - \Pr_{\haarisometry\gets\overline{\Haar_{n,n+m}}} \left[ \adversary^{\haarisometry(\cdot),\haarisometry^{-1}(\cdot)} = 1 \right]  \right| 
    \leq \negl(\secp).$$
\end{itemize}
\end{definition}

%% file: construction/invar_to_sec.tex
\section{Properties of Haar Unitaries}

\noindent We prove some useful properties of Haar unitaries.

\input{haar_ortho_single}
\subsection{Almost Invariance under $q$-fold Haar Unitary}
\label{sec:haar_invariance}
\par We introduce a notion called \emph{almost invariance} under a $q$-fold Haar unitary and prove some important properties about it. Most importantly, we characterize the condition that a given quantum channel is close to a $q$-fold Haar unitary using almost invariance in \Cref{clm:almostinvariance:closeness}.

\begin{definition}[Almost Invariance]
\label{def:almostinvariance}
Let $n,q,\ell\in\N$. An $(nq+\ell)$-qubit state $\rho$ is $\eps$-almost invariant under $q$-fold Haar unitary if the following holds:
\[
\TD\left(\rho, \E_{U\leftarrow\overline{\Haar_{m+n}}}\left[\left(I_{\ell}\otimes U^{\otimes q}\right)(\rho)\right]\right)\leq \eps.
\]
Moreover, if $\rho$ is $0$-almost invariant, then we say that $\rho$ is \emph{invariant} under $q$-fold Haar unitary.
\end{definition}

\noindent We prove two important facts about almost invariance property. The first fact states the following: if $\rho$ is almost invariant under $q$-fold Haar and moreover, $\sigma$ is close to $\rho$ then $\sigma$ should also be $q$-fold Haar invariant. The second fact states that almost invariance under $q$-fold Haar unitary can be leveraged to show closeness to the action of $q$-fold Haar unitary. 

\begin{claim}
\label{clm:observation:almostinvariance}
    Let $\rho,\sigma$ be two $(nq+\ell)$-qubit states be such that:
    \begin{itemize}
        \item $\TD\left( \rho, \sigma \right) \leq \delta$,
        \item $\rho$ is $\eps$-almost invariant under $q$-fold Haar unitary,
    \end{itemize}
then $\sigma$ is $(\eps+2\delta)$-almost invariant under $q$-fold Haar unitary.
\end{claim}
\begin{proof}
Since $\TD\left(\rho,\sigma \right) \leq \delta$ and $\TD\left(\rho, \E_{U\leftarrow\overline{\Haar_{m+n}}}\left[\left(I_{\ell}\otimes U^{\otimes q}\right)(\rho)\right]\right)\leq \eps$, by triangle inequality, we have
\[
\TD\left(\sigma, \E_{U\leftarrow\overline{\Haar_{m+n}}}\left[\left(I_{\ell}\otimes U^{\otimes q}\right)(\rho)\right]\right)
\leq \eps + \delta.
\]
\noindent Since applying a channel on two states cannot increase the trace distance between them (i.e., monotonicity of trace distance), we have
\[
\TD\left(\E_{U\leftarrow\overline{\Haar_{m+n}}}\left[\left(I_{\ell}\otimes U^{\otimes q}\right)(\rho)\right], \E_{U\leftarrow\overline{\Haar_{m+n}}}\left[\left(I_{\ell}\otimes U^{\otimes q}\right)(\sigma)\right]\right)
\leq \delta.
\]
By triangle inequality, $$\TD\left(\sigma, \E_{U\leftarrow\overline{\Haar_{m+n}}}\left[\left(I_{\ell}\otimes U^{\otimes q}\right)(\sigma)\right]\right)\leq \eps + 2\delta.$$
Hence, $\sigma$ is $(\eps+2\delta)$-almost invariant under $q$-fold Haar unitary.
\end{proof}


\begin{claim}
\label{clm:almostinvariance:closeness}
Let $\mu,q,\ell \in \mathbb{N}$. Suppose $\Phi$ is a quantum channel that is a probabilistic mixture of unitaries on $(\mu q+\ell)$ qubits.\footnote{Such channel is referred to as a \emph{mixed-unitary channel}, see~\cite[page~202]{WatrousBook}.} More precisely, $\Phi(\rho) = \mathbb{E}_{k \leftarrow \mathcal{D}}[ (I_{\ell} \otimes V_k^{\otimes q}) \rho (I_{\ell} \otimes (V_k^{\dagger})^{\otimes q})]$, where $\mathcal{D}$ is a distribution on $\bit^*$ and $V_k:\C^{2^{\mu}} \rightarrow \C^{2^{\mu}}$ is a unitary for every $k\in\bit^*$.
\par Suppose for a $(\mu q+\ell)$-qubit state $\rho$, $\Phi(\rho)$ is $\eps$-almost invariant under $q$-fold Haar unitary, where $\eps$ is a negligible function, then the following holds: 
\[
\TD\left(\Phi(\rho),\  \E_{U\leftarrow\overline{\Haar_{\mu}}}\left[\left(I_{\ell}\otimes U^{\otimes q}\right)(\rho)  \left(I_{\ell}\otimes (U^{\dagger})^{\otimes q}\right) \right]\right) 
\leq \eps.
\]
\end{claim}
\begin{proof}
Since $\Phi(\rho)$ is $\eps$-almost invariant under $q$-fold Haar unitary,
\[
\TD\left(\Phi(\rho),\  \E_{U\leftarrow\overline{\Haar_{\mu}}}\left[\left(I_{\ell}\otimes U^{\otimes q}\right)(\Phi(\rho))  \left(I_{\ell}\otimes (U^{\dagger})^{\otimes q}\right) \right]\right) 
\leq \eps.
\]
From the unitary invariance property of Haar, it follows that:
\[
\E_{U\leftarrow\overline{\Haar_{\mu}}}\left[\left(I_{\ell}\otimes U^{\otimes q}\right)(\Phi(\rho))  \left(I_{\ell}\otimes (U^{\dagger})^{\otimes q}\right) \right] = \E_{U\leftarrow\overline{\Haar_{\mu}}}\left[\left(I_{\ell}\otimes U^{\otimes q}\right)(\rho)  \left(I_{\ell}\otimes (U^{\dagger})^{\otimes q}\right) \right].
\]
The claim follows. 
\end{proof}

\noindent What the above claim says is that if the output of $\Phi$ (on $\rho$) is almost invariant under $q$-fold Haar then the action of $\Phi$ (on $\rho$) is close to $q$-fold Haar. 

\subsubsection{Invariant Subspace of $q$-fold Haar Unitary}
In the last subsection, we introduce the notion of almost invariance under $q$-fold Haar and show that this notion is very closely linked to checking if the action of a channel is close to the action of $q$-fold Haar. In this section, we characterize the space of states that are invariant under the $q$-fold Haar unitary. In particular, we will characterize the  ($qn$-qubit) states $\rho$ that satisfy the following property: 
$$\rho = \E_{U\leftarrow \overline{\Haar_{n}}} \left[U^{\otimes q}\rho (U^{\dagger})^{\otimes q}\right].$$

\noindent Note that any permutation operator commutes with any $q$-fold unitary, i.e. $U^{\otimes q}  P_{\sigma} = P_{\sigma} U^{\otimes q}$ for any $\sigma\in S_q$.\footnote{In fact, Schur-Weyl duality states that the commutant of $q$-fold unitaries is the span of permutation operators associated to $S_q$. See~\cite{HarrowThesis} and~\cite{mele2023introduction} for an exposition in quantum-information perspective.} Hence we get that for any $\sigma\in S_q$, 
\[
\E_{U\leftarrow \overline{\Haar_{n}}} \left[U^{\otimes q}P_{\sigma} (U^{\dagger})^{\otimes q}\right] = \E_{U\leftarrow \overline{\Haar_{n}}} \left[P_{\sigma}U^{\otimes q} (U^{\dagger})^{\otimes q}\right] = P_{\sigma}.
\]
This means that $P_{\sigma}$ is invariant under the $q$-fold Haar unitary for all $\sigma\in S_q$. Hence any linear combination $\rho = \sum_{\sigma\in S_q}\alpha_{\sigma} P_{\sigma}$ of permutation operators is also invariant under $q$-fold Haar unitary. It turns out that this condition is also necessary. That is, if $\rho$ is invariant under $q$-fold Haar unitary, then $\rho = \sum_{\sigma\in S_q}\alpha_{\sigma} P_{\sigma}$ for some values of $\alpha_{\sigma}$. To see this, we need the following theorem regarding the output of applying $q$-fold Haar unitary on a state.

\begin{theorem}[Twirling channel, rephrased from~{\cite[Theorem 10]{mele2023introduction}}]
    Let $\rho\in\mathcal{D}(\C^{2^{nq}})$, then $$\E_{U\leftarrow \overline{\Haar_{n}}} \left[U^{\otimes q}\rho (U^{\dagger})^{\otimes q}\right] = \sum_{\sigma\in S_q} c_{\sigma}(\rho) P_{\sigma},$$ where $c_{\sigma}(\rho) \in \C$.
\end{theorem}

\noindent Thus, if $\rho$ is invariant under $q$-fold Haar unitary, then $\rho = \E_{U\leftarrow \overline{\Haar_{n}}} \left[U^{\otimes q}\rho (U^{\dagger})^{\otimes q}\right] = \sum_{\sigma\in S_q} c_{\sigma}(\rho) P_{\sigma}$. Formally, we have the following corollary.

\begin{corollary}
\label{lem:perm-invariance}
    Let $\rho\in\mathcal{D}(\C^{2^{nq}})$. Then $\rho$ is invariant under $q$-fold Haar unitary if and only if there exists $c_{\sigma}(\rho)\in\C$ for all $\sigma\in S_q$ such that $$\rho = \sum_{\sigma\in S_q} c_{\sigma}(\rho) P_{\sigma}.$$
\end{corollary}

\input{construction/rho_uni}


%% file: haar_ortho_single.tex
\subsection{Haar Unitary on Orthogonal Inputs}
We start by studying the action of an $s$-fold Haar unitary. Recall that a Haar unitary is closely related to a Haar isometry, for the latter can be represented as appending appropriately many $0$s followed by applying a Haar unitary. 
\par One way to understand an $s$-fold Haar unitary $U$ is that it scrambles a collection of $s$ quantum states while respecting the pairwise inner-products. A special case of interest is when all the pairwise inner products are zero, i.e. when the input equals the tensor product of $s$ orthogonal states. By unitary invariance of the Haar measure, we can consider the input in the computational basis without loss of generality. Below in \Cref{lem:haar_perp_to_iid} we formalize this depiction of a Haar unitary by showing that this is statistically close to $s$ i.i.d. Haar-random states, even if given polynomially many copies of each state.

\begin{lemma}
\label{lem:haar_perp_to_iid}
Let $n,s,t\in\N$ and $\vec{x} = (x_1,\cdots,x_s)\in\bit^{ns}$ such that $\vec{x}$ has no repeating coordinates. Let
$$        
\rho 
:= \E_{U\gets\overline{\Haar_{n}}} \left[ \bigotimes_{j=1}^{s} \left( U\ketbra{x_j}{x_j}U^{\dagger}\right)^{\otimes t} \right], 
$$
and
$$
\sigma 
:= \bigotimes_{j=1}^{s} \Ex_{U_j \gets \overline{\Haar_{n}}} \left[ \left( U_j \ketbra{0^n}{0^n} U_j^{\dagger}\right)^{\otimes t} \right],
$$
then $\TD(\rho,\sigma) = O(s^2t/2^{n})$.
\end{lemma}

\begin{proof}
    We prove this using the hybrid method.
    \paragraph{Hybrid~$1$.} Sample $(U_1,\ldots,U_s)$ i.i.d. from $\overline{\Haar_n}$ and output $$\bigotimes_{j=1}^{s} \left( U_j \ketbra{0^n}{0^n} U_j^{\dagger}\right)^{\otimes t}.$$
    \paragraph{Hybrid~$2.i$ for $1\leq i\leq s$.} Sample $U$ from $\overline{\Haar_{n}}$ and $U_i,\cdots,U_{s}$ i.i.d. from $\overline{\Haar_{n}}$. Output $$\bigotimes_{j=1}^{i-1} \left(U\ketbra{x_j}{x_j}U^{\dagger}\right)^{\otimes t} \otimes \bigotimes_{j=i}^{s} \left(U_j\ketbra{0^n}{0^n}U_j^{\dagger}\right)^{\otimes t}.$$
    \paragraph{Hybrid~$3$.} Sample $U$ from $\overline{\Haar_{n}}$ and output $$\bigotimes_{j=1}^{s}\left(U^{\otimes t}\ketbra{x_j}{x_j}^{\otimes t}\left(U^{\dagger}\right)^{\otimes t}\right).$$

    \noindent Note that Hybrid~$1$ and Hybrid~$2.1$ are syntactically equivalent.
    \begin{claim}
        For $1\leq i\leq s-1$, the trace distance between Hybrid~$2.i$ and Hybrid~$2.(i+1)$ is $O(it/2^{n})$.
    \end{claim}

\begin{proof}
For $1\le k\le 2^n$, we define the distribution $\mu_k$ over $\pqstates{n}^{\otimes k}$ via the following procedures:
\begin{itemize}
    \item Let $V_0 := \set{0}$.
    \item For $i = 1,2\dots,k$: Sample $\ket{\vartheta_i} \gets \Haar(V_{i-1}^\perp)$ and let $V_i := \mathsf{span}\set{\ket{\vartheta_1},\ket{\vartheta_2},\dots,\ket{\vartheta_i}}$.
    \item Output $(\ket{\vartheta_1},\ket{\vartheta_2},\dots,\ket{\vartheta_k})$.
\end{itemize}
From~\Cref{fact:sampling_Haar_unitary}, the output of Hybrid~$2.i$ is identical to 
\begin{align*}
\rho_i 
= \Ex_{(\ket{\vartheta_1},\ket{\vartheta_2},\dots,\ket{\vartheta_{i-1}}) \gets \mu_{i-1}}
\left[ 
\bigotimes_{j=1}^{i-1} \ketbra{\vartheta_j}{\vartheta_j}^{\otimes t}
\right] 
\otimes \bigotimes_{j=i}^{s} \Ex_{\ket{\vartheta_j}\gets\Haar_n}\left[ \ketbra{\vartheta_j}{\vartheta_j}^{\otimes t} \right]
\end{align*}

\noindent Similarly, the output of Hybrid~$2.(i+1)$ is identical to 
\begin{align*}
\rho_{i+1}
& = \Ex_{(\ket{\vartheta_1},\ket{\vartheta_2},\dots,\ket{\vartheta_{i}}) \gets \mu_i}
\left[ 
\bigotimes_{j=1}^i \ketbra{\vartheta_j}{\vartheta_j}^{\otimes t}
\right] 
\otimes \bigotimes_{j=i+1}^{s} \Ex_{\ket{\vartheta_j}\gets\Haar_n}\left[ \ketbra{\vartheta_j}{\vartheta_j}^{\otimes t} \right]
\end{align*}

\noindent From the fact that $\TD(X \otimes Z, Y \otimes Z) = \TD(X, Y)$, the trace distance between $\rho_i, \rho_{i+1}$ is equivalent to that between 
\begin{align*}
\tilde{\rho}_i := 
\Ex_{(\ket{\vartheta_1},\ket{\vartheta_2},\dots,\ket{\vartheta_{i-1}}) \gets \mu_{i-1}}
\left[ 
\bigotimes_{j=1}^{i-1} \ketbra{\vartheta_j}{\vartheta_j}^{\otimes t}
\otimes \Ex_{\ket{\vartheta_i}\gets\Haar_n}\left[ \ketbra{\vartheta_i}{\vartheta_i}^{\otimes t} \right]
\right] 
\end{align*}
and
\begin{align*}
\tilde{\rho}_{i+1} &:= 
\Ex_{(\ket{\vartheta_1},\ket{\vartheta_2},\dots,\ket{\vartheta_i}) \gets \mu_i}
\left[ 
\bigotimes_{j=1}^i \ketbra{\vartheta_j}{\vartheta_j}^{\otimes t}
\right] \\
&= \Ex_{(\ket{\vartheta_1},\ket{\vartheta_2},\dots,\ket{\vartheta_{i-1}}) \gets \mu_{i-1}}
\left[\bigotimes_{j=1}^{i-1} \ketbra{\vartheta_j}{\vartheta_j}^{\otimes t}\otimes  \Ex_{\ket{\vartheta_i}\gets\Haar(V_{i-1}^\perp)} \left[
\ketbra{\vartheta_i}{\vartheta_i}^{\otimes t}
\right] \right].
\end{align*}
By strong convexity of trace distance and the fact $\TD(X \otimes Z, Y \otimes Z) = \TD(X, Y)$ again,
\begin{align*}
\TD(\tilde{\rho}_i,\tilde{\rho}_{i+1}) 
\leq \Ex_{(\ket{\vartheta_1},\ket{\vartheta_2},\dots,\ket{\vartheta_{i-1}}) \gets \mu_{i-1}}\left[
\TD \left( 
\Ex_{\ket{\vartheta_i}\gets\Haar_n}\left[ 
\ketbra{\vartheta_i}{\vartheta_i}^{\otimes t} 
\right], 
\Ex_{\ket{\vartheta_i}\gets\Haar(V_{i-1}^\perp)} \left[
\ketbra{\vartheta_i}{\vartheta_i}^{\otimes t}
\right] \right)
\right].
\end{align*}
For any fixed $(\ket{\vartheta_1},\ket{\vartheta_2},\dots,\ket{\vartheta_{i-1}})$ sampled from $\mu_{i-1}$, 
\begin{align*}
& \TD \left( 
\Ex_{\ket{\vartheta_i}\gets\Haar_n}\left[ 
\ketbra{\vartheta_i}{\vartheta_i}^{\otimes t} 
\right], 
\Ex_{\ket{\vartheta_i}\gets\Haar(V_{i-1}^\perp)} \left[
\ketbra{\vartheta_i}{\vartheta_i}^{\otimes t}
\right] \right) \\
= & \TD\left( \frac{\Pi_\sym^{2^n,t}}{\dim(\vee^t \C^{2^n})}, \frac{\Pi_\sym^{V_{i-1}^\perp,t}}{\dim(\vee^t V_{i-1}^\perp)} \right) \\
= & \frac{\dim(\vee^t \C^{2^n}) - \dim(\vee^t V_{i-1}^\perp)}{\dim(\vee^t \C^{2^n})}
= 1 - \frac{\binom{(2^n-i+1)+t-1}{t}}{\binom{2^n+t-1}{t}} \\
= & 1 - \frac{(2^n+t-i)\cdot (2^n+t-i-1) \cdot \dots \cdot (2^n-i+1)}{(2^n+t-1) \cdot (2^n+t-2) \cdot \dots \cdot 2^n} \\
= & 1 - \prod_{j=0}^{t-1} \left(1 - \frac{i-1}{2^n+t-1-j} \right) \\
\leq & 1 - \left(1 - \sum_{j=0}^{t-1}  \frac{i-1}{2^n+t-1-j} \right) 
= O\left( \frac{it}{2^n} \right).
\end{align*}
The first equality follows from~\Cref{fact:avg-haar-random}. The second equality follows from the following reasons. First, $V_{i-1}^\perp$ is a subspace of $\C^{2^n}$, so $\vee^t V_{i-1}^\perp$ is also a subspace of $\vee^t \C^{2^n}$. Therefore, the fully mixed states in $\vee^t V_{i-1}^\perp$ and $\vee^t \C^{2^n}$ can be simultaneously diagonalized. In such a basis, the trace distance between them degrades to the statistical distance between two uniform distributions with support $S_0,S_1$  respectively such that $|S_0| = \dim(\vee^t V_{i-1}^\perp)$, $|S_1| = \dim(\vee^t \C^{2^n})$ and $S_0 \subseteq S_1$. The statistical distance between is $(|S_1|-|S_0|)/|S_1|$ from a direct calculation. The last inequality follows from $1-\sum_i\veps_i \leq \prod_i(1-\veps_i)$ when $\veps_i\in[0,1]$ for every $i$.
\end{proof}

    \begin{claim}
        The trace distance between Hybrid $2.s$ and Hybrid 3 is $O(st/2^{n})$.
    \end{claim}
    \begin{proof}
        Using the same argument as for the above claim, we get $O(st/2^{n})$.
    \end{proof}

\noindent By triangle inequalities, the trace distance between Hybrid~$1$ and Hybrid~$3$ is $\sum_{i=1}^{s} O(it/2^{n})$ $ = O(s^2t/2^{n})$. This completes the proof of~\Cref{lem:haar_perp_to_iid}.
\end{proof}

\noindent Letting $t=1$ in~\Cref{lem:haar_perp_to_iid} yields the following corollary.
\begin{corollary}
\label{lem:haar_perp}
Let $n,q\in\N$ and $\vec{x} = (x_1,\cdots,x_q)\in\bit^{nq}$ such that $\vec{x}$ has no repeating coordinates. Let
\begin{align*}        
\rho 
:= \E_{U\gets\overline{\Haar_{n}}} \left[U^{\otimes q}\ketbra{\vec{x}}{\vec{x}} \left(U^{\dagger}\right)^{\otimes q} \right]   
\quad \text{ and } \quad
\sigma 
:= \Ex\left[ \ketbra{\vec{z}}{\vec{z}}: \vec{z}\ugets\cS_{n,q} \right],
\end{align*}
where $\cS_{n,q} := \set{\vec{z} = (z_1,\dots,z_q)\in\bit^{nq}: \vec{z} \text{ has no repeating coordinates}}$. Then $\TD(\rho,\sigma) = O(q^2/2^n)$.
\end{corollary}
\begin{proof}
From \Cref{lem:haar_perp_to_iid}, we know that $\rho$ is close to the following state $\rho'$ with the trace distance bounded by $O(q^2/2^n)$,
$$\rho' = \E_{U_1,\cdots,U_q\gets\overline{\Haar_{n}}} \left[\otimes_{i=1}^{q}U_i\ketbra{0^n}{0^n}U_i^{\dagger}\right].$$
Then by \Cref{fact:avg-haar-random}, we have
$$\rho' = \E_{a_1,\cdots,a_q\gets\bit^n} \left[\otimes_{i=1}^{q}\ketbra{a_i}{a_i}\right].$$
This can equivalently be written as  
$$\rho' = \E_{\vec{a}\gets\bit^{nq}} \left[\ketbra{\vec{a}}{\vec{a}}\right].$$
Then by a collision bound, we get that $\rho'$ and $\sigma$ are close in trace distance $O(q^2/2^n)$. Hence, the trace distance between $\rho$ and $\sigma$ is at most $O(q^2/2^n)$.
\end{proof}

%% file: construction/rho_uni.tex
\subsubsection{Instantiations of Almost Invariant States}
\label{sec:invariance_rho_uni}
In this subsection, we find a state that is almost invariant under the $q$-fold Haar unitary. This state would mimic the properties of output construction on various classes of inputs (as we will see in~\Cref{sec:closeness_to_rho_uni}).


\noindent We start by defining two special classes of tuples of types. To define these, we consider the symmetric subspace of $(\C^N)^{\otimes t}$ denoted by $\vee^t \mathbb{C}^N$, where the dimension is $N := 2^{n+m}$. We use the notation $\cH_{\sym} := \left( \vee^t \mathbb{C}^N \right)^s$ (the $s$-fold tensor of the symmetric subspace). It holds that 
$$\set{\ket{\type_{T_1}}\otimes \dots \otimes \ket{\type_{T_s}}: \hamming(T_i)=t,\ \forall i\in\set{1,\dots,s}}$$
forms an orthonormal basis of $\cH_{\sym}$. We say that an $s$-tuple of types $(T_1,\dots,T_s)$ is \emph{distinct} if for all $i,j\in [s]$ with $i\neq j$, $\setT(T_i)\cap\setT(T_j)=\emptyset$. Moreover, we say that an $s$-tuple of types $(T_1,\dots,T_s)$ is \emph{unique} if $(T_1,\dots,T_s)$ is distinct and for all $i\in [s]$, $\setT(T_i)$ contains $t$ distinct elements.

\noindent We define $\tdis{n+m}{s}{t}$ to be the set of all distinct $(T_1,\dots,T_s)$ where for all $i\in [s]$, $\hamming(T_i) = t$ and $\tuni{n+m}{s}{t}$ to be the set of all unique $(T_1,\dots,T_s)$ where for all $i\in [s]$, $\hamming(T_i) = t$. Note that $\tuni{n+m}{s}{t}\subseteq \tdis{n+m}{s}{t}$.

\noindent Let 
$$\rho_{\unique_{s,t}} := \E_{(T_1,\cdots,T_s)\leftarrow \tuni{m+n}{s}{t}} \bigotimes_{i=1}^{s} \ketbra{\type_{T_i}}{\type_{T_i}}.$$
We will show that $\rho_{\unique_{s,t}}$ is almost invariant under $q(=st)$-fold Haar unitary.

\begin{lemma}[Almost Invariance of $\rho_{\unique}$]
\label{lem:tuni_invar}
    Let $n,m,s,t\in\poly(\secparam)$, $q=st$, and let $\tuni{m+n}{s}{t}$ be defined as the set containing all $s$ tuples of types $(T_1,\cdots,T_s)$ which are unique. Let 
    $$\rho_{\unique_{s,t}} := \E_{(T_1,\cdots,T_s)\leftarrow \tuni{m+n}{s}{t}} \bigotimes_{i=1}^{s} \ketbra{\type_{T_i}}{\type_{T_i}} $$ 
    then $\rho_{\unique_{s,t}}$ is $O(s^2t^2/2^{m+n})$-almost invariant under $q$-fold Haar unitary.
\end{lemma}

\input{construction_proofs/rho_uni_proof}
\begin{lemma}
\label{lem:tuni_invar_aux}
    Let $\rho_{\unique_{s,t}}$ be as defined above. Define for any $\ell$-qubit state $\sigma$, $\rho^{\sigma}_{\unique_{s,t}}:= \sigma\otimes\rho_{\unique_{s,t}}$. Then $\rho^{\sigma}_{\unique_{s,t}}$ is also $O(s^2t^2/2^{m+n})$-almost invariant under $q$-fold Haar unitary with $I_{\ell}$ being applied on $\sigma$ (or is $O(s^2t^2/2^{m+n})$-almost invariant under $I_{\ell}\otimes U^{\otimes q}$ where $U$ is sampled from the Haar measure).
\end{lemma}


%% file: construction_proofs/rho_uni_proof.tex
\begin{proof}
We prove this by showing that $\rho_{\unique_{s,t}}$ is close to $t$ copies of $s$ i.i.d. sampled Haar states. Next we show that $t$ copies of $s$ i.i.d. sampled Haar states can be written as a mixture of permutation operators and hence is invariant under $q (=st)$-fold Haar unitary. Then by~\Cref{clm:observation:almostinvariance}, we would get that $\rho_{\unique_{s,t}}$ is almost invariant under $q$-fold Haar unitary. We start by showing the following lemma: 
    \begin{lemma}
        \label{lem:tuni_haar_dis}
        Let $n,m,s,t\in\poly(\secparam)$, and let $\tuni{n}{s}{t}$ be defined as the set containing all $s$ tuples of types $(T_1,\cdots,T_s)$ which are unique. Let 
        $$\rho_{\unique_{s,t}} := \E_{(T_1,\cdots,T_s)\leftarrow \tuni{m+n}{s}{t}} \bigotimes_{i=1}^{s} \ketbra{\type_{T_i}}{\type_{T_i}} $$ and let $$\hat{\rho} := \E_{U_1,\cdots,U_s\leftarrow \overline{\Haar_{m+n}}} \bigotimes_{i=1}^{s}\left(U_i\ketbra{0^n}{0^n}U_i^{\dagger}\right)^{\otimes t},$$ then $$\TD(\rho_{\unique_{s,t}},\hat{\rho}) = O(s^2t^2/2^{m+n}).$$
    \end{lemma}
    \begin{proof}
        We prove this using the hybrid method.
        \paragraph{Hybrid $1$.} Sample $(T_1,\cdots,T_s)$ from $\tuni{n+m}{s}{t}$ and output $$\bigotimes_{i=1}^{s} \ketbra{\type_{T_i}}{\type_{T_i}}.$$
        \paragraph{Hybrid $2.i$., for $1\leq i\leq s$} Sample $(T_i,\cdots,T_{s})$ from $\tuni{n+m}{s-i+1}{t}$, for $1\leq j<i$, sample $T_j$ from $\tuni{n+m}{1}{t}$ and output $$\bigotimes_{j=1}^{s} \ketbra{\type_{T_j}}{\type_{T_j}}.$$
        \paragraph{Hybrid $3$.} Sample $U_1,\cdots,U_s$ i.i.d. from $\overline{\Haar_{m+n}}$, and output 
        $$\bigotimes_{i=1}^{s}\left(U_i\ketbra{0^n}{0^n}U_i^{\dagger}\right)^{\otimes t}.$$
    
        \begin{claim}
            Hybrid $1$. and Hybrid $2.1$ are identical. 
        \end{claim}
        \begin{proof}
            This is true since the sampling procedures used in Hybrid $1$ and Hybrid $2.1$ are the same.
        \end{proof}
        \begin{lemma}
            For $1\leq i\leq s-1$, the trace distance between Hybrid $2.i$ and Hybrid $2.(i+1)$ is $O((s-i+1)t^2/2^{n+m})$. 
        \end{lemma}
        \begin{proof}
            Notice that for $j<i$, $T_j$ is identically distributed. Hence, we need to find the distance between $\bigotimes_{j=i}^{s} \ketbra{\type_{T_j}}{\type_{T_j}}$ for $(T_i,\cdots,T_s)$ sampled from $\tuni{n+m}{s-i}{t}\times\tuni{n+m}{1}{t}$ versus $\tuni{n+m}{s-i+1}{t}$.
            
            \noindent Notice that, sampling from $\tuni{n+m}{s-i+1}{t}$ is equivalent to choosing $(s-i+1)t$ distinct elements from $[2^{n+m}]$. Similarly, sampling from $\tuni{n+m}{s-i}{t}\times\tuni{n+m}{1}{t}$ is equivalent to choosing $(s-i)t$ distinct elements from $[2^{n+m}]$ and then choosing $t$ distinct elements from $[2^{n+m}]$. In this case, the probability of having a collision between these two sets is $O((s-i+1)t^2/2^{n+m})$. Thus, the statistical distance between the uniform distribution on $\tuni{n+m}{s-i+1}{t}\times\tuni{n+m}{1}{t}$ and the uniform distribution on $\tuni{n+m}{s-i}{t}$ is $O((s-i+1)t^2/2^{n+m})$. This in turn implies that the trace distance between Hybrid $2.i$ and Hybrid $2.(i+1)$ is $O((s-i+1)t^2/2^{n+m})$.
        \end{proof}
        \begin{lemma}
            The trace distance between Hybrid $2.s$ and Hybrid $3$ is $O(st^2/2^{n+m})$.
        \end{lemma}
        \begin{proof}
            Since, in Hybrid $3$, all the $U_j$'s are sampled independently, the output of Hybrid $3$ can be equivalently written as $$ \bigotimes_{i=1}^{s}\E_{\ket{\vartheta_i}\leftarrow \Haar_{n+m}}\left(\ketbra{\vartheta_i}{\vartheta_i}\right)^{\otimes t}.$$ 
            Next by \Cref{fact:avg-haar-random}, we know that that this is equivalent to $$\bigotimes_{i=1}^{s}\E_{\substack{T_i\leftarrow [t+1]^{2^{n+m}}\\ \hamming(T_i)=t}} \ketbra{\type_{T_i}}{\type_{T_i}}.$$
            Note that if instead of sampling $T_i$ uniformly from the set of vectors from $[t+1]^{2^{n+m}}$ with $\hamming(T_i) = t$, we sample $T_i$ from $\tuni{n+m}{1}{t}$, we get the output of Hybrid $2.s$. In particular, we know that the output of Hybrid $2.s$ can be written as $$\bigotimes_{i=1}^{s}\E_{T_i\leftarrow \tuni{n+m}{1}{t}} \ketbra{\type_{T_i}}{\type_{T_i}}.$$ 
            Since the probability of having a collision when choosing $t$ elements from $[2^{n+m}]$ is $O(t^2/2^{n+m})$, the statistical distance between the distributions $T_i$ chosen uniformly from the vectors in $[t+1]^{2^{n+m}}$ with $\hamming(T_i)=t$ versus $T_i$ sampled from $\tuni{n+m}{1}{t}$ is $O(t^2/2^{n+m})$ for each $i$. Hence, the trace distance between Hybrid $2.s$ and Hybrid $3$ is $O(st^2/2^{n+m})$.
        \end{proof}
        \noindent Combining the above, we get the trace distance between Hybrid $1$ and Hybrid $3$ is $O(s^2t^2/2^{m+n})$. This completes the proof of~\Cref{lem:tuni_haar_dis}.
    \end{proof}
    \noindent Next we show that $$\hat{\rho} = \E_{U_1,\cdots,U_s\leftarrow \overline{\Haar_{m+n}}} \bigotimes_{i=1}^{s}\left(U_i\ketbra{0^n}{0^n}U_i^{\dagger}\right)^{\otimes t},$$ is invariant under $q$-fold Haar unitary. To do this we show that $\hat{\rho}$ can be written as a mixture of permutation operators. Notice that by \Cref{fact:avg-haar-random}, $$\E_{U_1,\cdots,U_s\leftarrow \overline{\Haar_{m+n}}} \bigotimes_{i=1}^{s}\left(U_i\ketbra{0^n}{0^n}U_i^{\dagger}\right)^{\otimes t} = \E_{\sigma_1,\cdots,\sigma_s\leftarrow S_t} \bigotimes_{i=1}^{s} P_{\sigma_i}.$$ Here, note that for any $\sigma_1,\cdots,\sigma_s\in S_t$, $\bigotimes_{i=1}^{s} P_{\sigma_i}$ can be written as $P_{\sigma_{1,\ldots,s}}$ for some $\sigma_{1,\ldots,s}\in S_{st}$. Hence, 
$$\E_{U_1,\cdots,U_s\leftarrow \overline{\Haar_{m+n}}} \bigotimes_{i=1}^{s}\left(U_i\ketbra{0^n}{0^n}U_i^{\dagger}\right)^{\otimes t} = \E_{\substack{\sigma\leftarrow S_{st}\\ \sigma_{1,\ldots,s}\text{ is $t$-internal}}} P_{\sigma},$$
where we say $\sigma$ is $t$-internal if $P_{\sigma}$ can be written as $\bigotimes_{i=1}^{s} P_{\sigma_i}$ for some $\sigma_1,\cdots,\sigma_s\in S_t$. Hence, from~\Cref{lem:perm-invariance}, we have that $\hat{\rho}$ is invariant under q-fold Haar unitary. Hence, by~\Cref{lem:tuni_haar_dis}, we get that $\rho_{\unique_{s,t}}$ is negligibly close to some mixture of permutation operators and by~\Cref{clm:observation:almostinvariance} $\rho_{\unique_{s,t}}$ is almost invariant under $q$-fold Haar unitary.

\end{proof}

%% file: construction_root_of_unity.tex
\section{Construction}
\label{sec:construction}
\noindent Let $m(\cdot),n(\cdot)$ be polynomials. Let $p=p(\secparam)$ be a $\secparam$-bit integer. Let $\secparam = 2\secparam_1$. We use the following tools in the construction of PRI. 
\begin{itemize}
    \item $f:\{0,1\}^{\secparam_1}\times\{0,1\}^{n(\secp_1) + m(\secp_1)} \to \mathbb{Z}_{p}$ is a quantum-query secure pseudorandom function ($\qprf$,~\Cref{def:qqprfs}). For a key $k\in\bit^{\secp_1}$, we denote $O_{f_k}$ to be a unitary which maps the state $\ket{x}$ to $\omega_{p}^{f(k,x)}\ket{x}$ for every $x\in\bit^{n(\secp_1)+m(\secp_1)}$ where $\omega_p$ is the $p$-th root of unity.
    \item $g:\{0,1\}^{\secparam_1}\times\{0,1\}^{n(\secp_1) + m(\secp_1)} \to\{0,1\}^{n(\secp_1) + m(\secp_1)}$ is a quantum-query secure pseudorandom permutation ($\qprp$,~\Cref{def:qqprps}). For a key $k\in\bit^{\secp_1}$, we denote $O_{g_k}$ to be a unitary which maps the state $\ket{x}$ to $\ket{g(k,x)}$ for every $x\in\bit^{n(\secp_1)+m(\secp_1)}$.\footnote{The instantiation of the unitary $O_{g_k}$ requires one query to $g(k,\cdot)$ and one query to $g^{-1}(k,\cdot)$~\cite{JLS18}.}
\end{itemize}
\noindent We present the construction of psuedorandom isometry $\{\prfsi_{\secparam}\}_{\secparam \in \mathbb{N}}$ in~\Cref{fig:prp_prs}. Note that the construction presented is functionally equivalent to an isometry even though it performs a partial trace.
Note that after appending $0$'s in the second step, our construction is a mixture of unitaries parametrized by the key $k$, so that it satisfies the condition of~\Cref{clm:almostinvariance:closeness}. Moreover, our construction is invertible (\Cref{def:invertible}). The inversion is done by reversing all the unitary operations in $F_\secparam$ and discarding (tracing out) the $m$-qubit register.\footnote{The application of $f(k_1, \cdot)$ can be inverted by manipulating the phase oracle to apply a negative phase, whereas $g(k_2, \cdot)$ can be inverted using the oracle access to $g^{-1}(k_2, \cdot)$.}

\begin{figure}[H]
   \begin{tabular}{|p{\textwidth}|}
   \hline 
   \ \\
\noindent On input a key $k \in \{0,1\}^{\secparam}$ and an $n$-qubit register $\bfX$. We define the operation of $\prfsi_{\secparam}(k,\cdot)$ as follows.
\begin{itemize}
    \item Parse the key $k$ as $k_1 || k_2$, where $k_1\in\{0,1\}^{\secparam_1}$ is a $\qprf$ key and $k_2\in\{0,1\}^{\secparam_1}$ is a $\qprp$ key.
    \item Append an $m$-qubit register $\bfZ$ initalized with $\ket{0^m}_\bfZ$ to register $\bfX$.
    \item Apply $H^{\otimes m}$ to register $\bfZ$.
    \item Apply $O_{f_{k_1}}$ to registers $\bfX$ and $\bfZ$.
    \item Apply $O_{g_{k_2}}$ to registers $\bfX$ and $\bfZ$.
\end{itemize}
Explicitly, $\prfsi_{\secparam}(k,\cdot)$ maps the basis vector $\ket{x}_\bfX$ to 
\[
\frac{1}{\sqrt{2^m}} \sum_{z\in \bit^m} \omega_{p}^{f(k_1,x||z)} \ket{g(k_2,x||z)}_{\bfX\bfZ}.
\]
\ \\
\hline
\end{tabular}
\caption{Description of $\prfsi_{\secparam}$.}
\label{fig:prp_prs}
\end{figure}

\input{sec_analysis}

%% file: sec_analysis.tex
\input{construction/comp_to_info}
\input{construction/closeness_to_rho_uni}

%% file: construction/comp_to_info.tex
\subsection{Invoking Cryptographic Assumptions}
\label{sec:comp_security}

We start by defining the information-theoretic version of~\Cref{fig:prp_prs}, i.e., the same construction but with $\qprp$ replaced by a random permutation $\pi\in S_{2^{n+m}}$ and $\qprf$ replaced by a random function $f\in \mathcal{F}_{2^{n+m},p}$. This construction, denoted by $G_{(f,\pi)}$, is given in~\Cref{fig:info_prp_prs}. 

We show that the construction in~\Cref{fig:info_prp_prs} is computationally indistinguishable from the one in~\Cref{fig:prp_prs}. 

\input{construction_proofs/info_fig}

\begin{theorem} \label{lem:comp_to_info}
Let $n,m=\poly(\secparam)$. Let $C_{\prfsi_{k}}$ be the quantum channel defined in~\Cref{fig:prp_prs} and let $G_{(f,\pi)}$ be as given in~\Cref{fig:info_prp_prs}. Then, assuming the security of $\qprf$ and $\qprp$, for any QPT adversary $\adversary$, the following holds: 
$$\left| \prob\left[1 = \adversary^{C_{\prfsi_k}}(1^{\secparam})\ : k \xleftarrow{\$} \{0,1\}^{\secparam} \right] - \prob\left[1 = \adversary^{G_{(f,\pi)}}\left( 1^{\secparam} \right)\ :\ \substack{f \xleftarrow{\$} \mathcal{F}_{2^{n+m},p}\\ \ \\ \pi \xleftarrow{\$}  S_{2^{n+m}}}  \right] \right| \leq \negl(\secparam),$$
for some negligible function $\negl(\cdot)$.
\end{theorem}

\begin{proof}[Proof of~\Cref{lem:comp_to_info}] 
We prove this by a standard hybrid argument. Consider the following hybrids: 

\begin{itemize}
\item Hybrid $\hybrid_0$: The oracle is $\prfsi_{\secparam}$ defined in~\Cref{fig:prp_prs}.
\item Hybrid $\hybrid_1$: The oracle is the same as $\prfsi_{\secparam}$ except that $\qprf$ is replaced by a random function.
\item Hybrid $\hybrid_2$: The oracle is $G_{(f,\pi)}$ defined in~\Cref{fig:info_prp_prs}.
\end{itemize}

\newcommand{\oracle}{{\cal O}}
\begin{claim}
Assuming the quantum-query security of $\qprf$, the output distributions of the hybrids $\hybrid_0$ and $\hybrid_1$ are computationally indistinguishable. 
\end{claim}
\begin{proof}
Suppose there exists some QPT algorithm $\adversary$ that distinguishes Hybrid~0 from Hybrid~1 with a non-negligible advantage $\nu$. We'll construct a reduction $\distinguisher$ that given oracle access to $\oracle$ distinguishes whether $\oracle$ is either the $\qprf$ oracle or a random function with the same advantage $\nu$ by using $\adversary$. Upon receiving a query $\ket{\psi}$ from $\adversary$, the reduction $\distinguisher$ responds by first applying  $H^{\otimes m}\otimes I_{n}$ on $\ket{0^m}\ket{\psi}$, querying $\oracle$ and finally, computing $g(k_2,\cdot)$, where $k_2$ is sampled uniformly at random from $\{0,1\}^{\secparam_1}$. Since $\distinguisher$ perfectly simulates the distributions of oracles in hybrids $\hybrid_0$ and $\hybrid_1$, it has the same distinguishing advantage as that of $\adversary$. However, this contradicts the post-quantum security of the underlying $\qprf$.
\end{proof}

\begin{claim}
Assuming the quantum-query security of $\qprp$, the output distributions of the hybrids $\hybrid_1$ and $\hybrid_2$ are computationally indistinguishable. 
\end{claim}
\begin{proof}
Suppose there exists some QPT algorithm $\adversary$ that distinguishes hybrids $\hybrid_1$ from $\hybrid_2$ with a non-negligible advantage $\nu$. We'll construct a reduction $\distinguisher$ that given access to an oracle $\oracle$ distinguishes where $\oracle$ implements $\qprp$ or a random permutation with the same advantage $\nu'$ by using $\adversary$. Suppose the number of queries made by $\adversary$ is $q = \poly(\secp)$. Since each query to the oracle needs to invoke the random function once, the number of queries to the random function is also $q$. Upon receiving a query $\ket{\psi}$ from $\adversary$, the reduction $\distinguisher$ responds by first applying  $H^{\otimes m}\otimes I_{n}$ on $\ket{0^m}\ket{\psi}$, applying a $2q$-wise independent hash function and finally, querying $\oracle$. From~\Cref{thm:zha12}, it follows that 
a $2q$-wise independent hash function perfectly simulates a random function. Thus, $\distinguisher$ perfectly simulates the distributions of the oracles in the hybrids $\hybrid_1$ and $\hybrid_2$. So $\distinguisher$ has the same distinguishing advantage as that of $\adversary$. However, this contradicts the post-quantum security of the underlying $\qprp$.
\end{proof}
\noindent Combining the above claims completes the proof of~\Cref{lem:comp_to_info}.
\end{proof}

\input{construction/pathway}



%% file: construction_proofs/info_fig.tex
\begin{figure}[H]
   \begin{tabular}{|p{\textwidth}|}
   \hline 
   \ \\
\noindent Let $f\in\mathcal{F}_{2^{n+m},p}$ and $\pi\in S_{2^{n+m}}$. On input an $n$-qubit register $\bfX$. We define the operation of $G_{(f,\pi)}(\cdot)$ as follows.
\begin{itemize}
    \item Append an $m$-qubit register $\bfZ$ initalized with $\ket{0^m}_\bfZ$ to register $\bfX$.
    \item Apply $H^{\otimes m}$ to register $\bfZ$.
    \item Apply $O_f$ to registers $\bfX$ and $\bfZ$.
    \item Apply $O_\pi$ to registers $\bfX$ and $\bfZ$.
\end{itemize}
Explicitly, $G_{(f,\pi)}(\cdot)$ maps the basis vector $\ket{x}_\bfX$ to 
\[
\frac{1}{\sqrt{2^m}} \sum_{z\in \bit^m} \omega_{p}^{f(x||z)} \ket{\pi(x||z)}_{\bfX\bfZ}.
\]
\ \\
\hline
\end{tabular}
\caption{Description of $G_{(f,\pi)}$.}
\label{fig:info_prp_prs}
\end{figure}

%% file: construction/pathway.tex
\subsection{A Pathway to Security via Almost Invariance} 
\noindent The next step would be to show that $q$-fold $G_{(f,\pi)}$ on a state from the query set is close (in trace distance) to $q$-fold Haar unitary, where $q$ is the number of adversarial queries, on the same state. To prove this, we rely on the notion of almost invariance defined in~\Cref{sec:haar_invariance}.\\

\noindent In particular, we identify interesting classes of $\qclass$ and show the closeness of $q$-fold $G_{(f,\pi)}$ on a state from one of the interesting classes is close to an almost invariant state (specifically the one given in~\Cref{sec:invariance_rho_uni}). Combining this with~\Cref{clm:observation:almostinvariance}, we would be showing that the state obtained after applying $q$-fold $G_{(f,\pi)}$ on $\rho$, where $\rho$ comes from one of these query classes, is almost invariant under $q$-fold Haar unitary. This when combined with~\Cref{clm:almostinvariance:closeness} would then show that $q$-fold $G_{(f,\pi)}$ on $\rho$ is close to $q$-fold Haar unitary on $\rho$. We formally show this in~\Cref{sec:summary}. 



%% file: construction/closeness_to_rho_uni.tex
\subsection{Closeness to Almost Invariant States}
\label{sec:closeness_to_rho_uni}

\noindent We identify different classes of $\qclass$ and show that applying $q$-fold $G_{(f,\pi)}$ on a state $\rho$ from one of these classes will be close to $\rho_{\unique_{s,t}}$.\\

\noindent Note that if the input state $\rho$ (which is $nq+\ell$ qubits) can be written as a product state where $\rho = \rho_1\otimes \rho_2$ where $\rho_1$ is an $\ell$ qubit state. Then we only need to show that $\rho_2$ is close to $\rho_{\unique_{s,t}}$ for some $s,t$ such that $st = q$. Hence, in the proofs, we ignore $\rho_1$.


\subsubsection{Distinct Type Queries} \label{sec:distinct_type}
We define a class of states $$\qclass^{(\distinct)}_{{n,t,s,\ell,\secparam}}:= {\color{red}{\cal D}(\C^{2^{\ell'(\secparam)}})} \otimes\set{{\color{blue}\bigotimes_{i=1}^{s}\ketbra{\type_{T_i}}{\type_{T_i}}}:(T_1,\cdots,T_s)\in\tdis{n}{s}{t}}.$$
Next, we define the following class: 
$$\qclass^{(\distinct)}_{{n,q,\ell,\secparam}}:= \bigcup_{\substack{s,t\\\text{ such that } q=st}}\qclass^{(\distinct)}_{{n,t,s,\ell,\secparam}}.$$
In this section, we prove the security of the construction on $\qclass_{(\distinct)} := \set{\qclass^{(\distinct)}_{{n,q,\ell,\secparam}}}_{\secparam\in\N}$. In particular, we prove the following:

\begin{theorem}
\label{lem:tdis_proof}
    Let $n,m,q,\ell = \poly(\secparam)$ and $\qclass_{\distinct} $ be as defined above, then, assuming the existence of post-quantum one-way functions, the construction of $\pri$ given in~\Cref{fig:prp_prs} is $\qclass_{\distinct}$-secure.
\end{theorem}

A straightforward corollary of the above theorem is that our construction is secure against computational basis states. Recall the definition of $\qclass_{\comp}$-security in~\Cref{sec:defs}. Suppose there are $t$ elements in the query set $\set{x_1,\dots,x_q}$ equal to some $x\in\set{0,1}^n$. Observe that $\ket{x}^{\otimes t}$ is a valid type state (\Cref{def:type_states}). In this manner, we can represent $\bigotimes_{i=1}^q\ketbra{x_i}{x_i}$ (up to re-ordering registers) as a tensor product of distinct type vectors. This results in the following corollary:

\begin{corollary}
\label{lem:comp_proof}
    Let $n,m,q,\ell = \poly(\secparam)$ and $\qclass_{\comp}$ be defined as in~\Cref{sec:defs}. Assuming the existence of post-quantum one-way functions, the construction of $\pri$ given in~\Cref{fig:prp_prs} is $\qclass_{\comp}$-secure.
\end{corollary}

To prove this, we show that the output of $G_{(f,\pi)}$ on any $\bigotimes_{i=1}^{s}\ketbra{\type_{T_i}}{\type_{T_i}}$ for $(T_1,\cdots,T_s)\in\tdis{n}{s}{t}$ is negligibly close to $\rho_{\unique_{s,t}}$.
\begin{lemma}
    \label{lem:tdis_info}
    Let $n,m,t,s = \poly(\secparam)$. Let $(T_1,\cdots,T_s)\in\tdis{n}{s}{t}$. 
    Let 
    $$\rho := \E_{(f,\pi)\leftarrow(\mathcal{F}_{2^{n+m},p},S_{2^{n+m}})}\left[\bigotimes_{i=1}^{s}G_{(f,\pi)}^{\otimes t}\ketbra{\type_{T_i}}{\type_{T_i}}(G_{(f,\pi)}^{\dagger})^{\otimes t}\right],$$
    and 
    $$\rho_{\unique_{s,t}} := \E_{(\overline{T}_1,\cdots,\overline{T}_s)\leftarrow \tuni{m+n}{s}{t}} \left[\bigotimes_{i=1}^{s} \ketbra{\type_{\overline{T}_i}}{\type_{\overline{T}_i}}\right].$$
    Then $\TD(\rho,\rho_{\unique_{s,t}}) = O(st^2/2^m)$.
\end{lemma}

\input{construction_proofs/type_distinct_proof}

Combining \Cref{lem:tdis_info}, \Cref{lem:tuni_invar} and \Cref{clm:observation:almostinvariance}, we get the desired result.




\subsubsection{Multiple Copies of the Same Input}
\label{sec:security_multi_same}
In this section, we prove security against multiple copies of the same input. Formally, we prove the following,
 
\begin{theorem}
\label{lem:single_proof}
    Let $n,m,q = \poly(\secparam)$ and $\qclass^{(\single)}_{n,q,\ell,\secparam}$ be as defined in~\Cref{sec:defs} then, assuming the existence of post-quantum one-way functions, the construction of $\pri$ given in \Cref{fig:prp_prs} is $\qclass_{\mathsf{single}_{q}}$-secure.
\end{theorem}

Note that since $\ketbra{\phi}{\phi}^{\otimes q}$ is in the symmetric subspace, and type states form an orthogonal basis of the symmetric subspace, we can write $\ketbra{\phi}{\phi}^{\otimes q} = \sum_{T,T'}\alpha_{T,T'}\ketbra{\type_T}{\type_{T'}}$. Notice that, $\ketbra{\type_T}{\type_{T}}\in\qclass_{\distinct}$ for any type $T$. So, to find the action of $q$-fold $G_{(f,\pi)}$ on $\ketbra{\phi}{\phi}^{\otimes q}$, we just need to look at its action on $\ketbra{\type_T}{\type_{T'}}$ for $T\neq T'$. To do this, we analyze the output of the construction on $\ketbra{\vec{x}}{\vec{x'}}$ for $\type(\vec{x})\neq\type(\vec{x'})$.

\begin{lemma}
    \label{lem:outer_comp_query}
    Let $n,m\in\N$, $q\in\poly(\secparam)$, $\vec{x} = (x_1,\cdots,x_q)\in\{0,1\}^{nq}$, and $\vec{x'} = (x'_1,\cdots,x'_q)\in\{0,1\}^{nq}$. Let $\type(\vec{x})\neq\type(\vec{x'})$. Let $G_{(f,\pi)}$ be as defined in \Cref{fig:info_prp_prs}, then $$\E_{(f,\pi)\leftarrow(\mathcal{F}_{2^{n+m},p},S_{2^{n+m}})} \left[\left(G_{(f,\pi)}^{\otimes q}\right)\left(\bigotimes_{i=1}^{q}\ketbra{x_i}{x'_i}\right)\left(G_{(f,\pi)}^{\dagger}\right)^{\otimes q}\right]=0.$$
\end{lemma}
\ifllncs
\noindent The proof of~\Cref{lem:outer_comp_query} can be found in~\Cref{app:closeness_to_rho_uni}.
\else
\input{construction_proofs/comp_outer_query_proof}
\fi

\noindent A corollary of the above lemma is as follows,
\begin{corollary}
    \label{lem:type_outer}
    Let $n,m\in\N$, $q\in\poly(\secparam)$, $T,T'\in[q+1]^{2^{n}}$ with $T\neq T'$ and $\hamming(T) = \hamming(T') = q$. Let $G_{(f,\pi)}$ be as defined in \Cref{fig:info_prp_prs}, then $$\E_{(f,\pi)\leftarrow(\mathcal{F}_{2^{n+m},p},S_{2^{n+m}})} \left[\left(G_{(f,\pi)}^{\otimes q}\right)\ketbra{\type_{T}}{\type_{T'}}\left(G_{(f,\pi)}^{\dagger}\right)^{\otimes q}\right]=0.$$
\end{corollary}

\noindent Using the above, we get that the output of $q$-fold $G_{(f,\pi)}$ on  $\ketbra{\phi}{\phi}^{\otimes q}$, is close to $\rho_{\unique_{s,t}}$. Formally, we prove the following,
\begin{lemma}
    \label{lem:multicopy_info}
    Let $n,m,q=\poly(\secparam)$, $\ket{\phi}$ be any $n$ qubit pure state. Let $G_{(f,\pi)}$ be as defined in \Cref{fig:info_prp_prs}, let $$\rho = \E_{(f,\pi)\leftarrow(\mathcal{F}_{2^{n+m},p},S_{2^{n+m}})} \left[\left(G_{(f,\pi)}^{\otimes q}\right)\ketbra{\phi}{\phi}^{\otimes q}\left(G_{(f,\pi)}^{\dagger}\right)^{\otimes q}\right],$$
    and 
    $$\rho_{\unique_{1,q}} = \E_{(\overline{T}_1,\cdots,\overline{T}_q)\leftarrow \tuni{m+n}{1}{q}} \left[\ketbra{\type_{\overline{T}_i}}{\type_{\overline{T}_i}}\right].$$
    Then $\TD(\rho,\rho_{\unique_{1,q}}) = O(q^2/2^m)$.
\end{lemma}
\ifllncs
\noindent The proof of~\Cref{lem:multicopy_info} can be found in~\Cref{app:closeness_to_rho_uni}.
\else
\input{construction_proofs/multicopy_proof}
\fi
\noindent Combining \Cref{lem:multicopy_info}, \Cref{lem:tuni_invar} and \Cref{clm:observation:almostinvariance}, we get the desired result. 



\subsubsection{Security against Haar Inputs}
In this section, we prove security against queries which are sampled i.i.d. from the Haar measure. In particular, we show the following theorem:
\begin{theorem}
\label{lem:haar_proof}
    Let $n,m,s,t = \poly(\secparam)$ and $\qclass_{n,q,\ell,\secparam}^{({\sf Haar})}$ be as defined in~\Cref{sec:defs} then, assuming the existence of post-quantum one-way functions, the construction of $\pri$ given in \Cref{fig:prp_prs} is $\qclass_{\mathsf{Haar}_{s,t}}$-secure. 
\end{theorem}

\noindent To prove this we first note that $$\E_{\ket{\vartheta_1},\cdots,\ket{\vartheta_s}\leftarrow\Haar_{n}}\left[\bigotimes_{i=1}^{s} \ketbra{\vartheta_i}{\vartheta_i}^{\otimes 2t}\right] = \E_{T_1,\cdots,T_s}\left[\bigotimes_{i=1}^{s} \ketbra{\type_{T_i}}{\type_{T_i}}\right],$$
where each $T_i$ is an i.i.d. sampled type containing $2t$ elements. We know that with overwhelming probability this lies in $\tuni{n}{s}{2t}$. Note that our construction is only acting on half of each of the type states and not the complete states. Hence, if we prove security against queries from $\tuni{n}{s}{2t}$ with the construction being applied to only one-half of each of the type states, we would get security for i.i.d. sampled Haar queries as required. In particular, we prove the following:
\begin{lemma}
\label{lem:tuni_info}
    Let $n,m,s,t = \poly(\secparam)$ and let $(T_1,\cdots,T_s)\in\tuni{n}{s}{2t}$. Let $G_{(f,\pi)}$ is as defined in \Cref{fig:info_prp_prs}. Let 
    $$\rho:= \E_{(f,\pi)\leftarrow(\mathcal{F}_{n+m},S_{n+m})} \bigotimes_{i=1}^{s} \left(\left(I_{nt}\otimes G_{(f,\pi)}^{\otimes t} \right)\ketbra{\type_{T_i}}{\type_{T_i}}\left(I_{nt}\otimes G_{(f,\pi)}^{\otimes t}\right)^{\dagger}\right),$$
    $$\sigma:= \E_{\substack{(\overline{T}_1,\cdots,\overline{T}_s)\subset (T_1,\cdots,T_s)\\ \forall i\in [s], \hamming(\overline{T}_i)=t}}\bigotimes_{i=1}^{s}\left(\ketbra{\type_{\overline{T}_i}}{\type_{\overline{T}_i}}\right),$$
    and $$\rho^{\sigma}_{\unique_{s,t}}:= \E_{\substack{(\overline{T}_1,\cdots,\overline{T}_s)\subset (T_1,\cdots,T_s)\\ \forall i\in [s], \hamming(\overline{T}_i)=t\\
    (\hat{T_1},\cdots,\hat{T_s})\leftarrow\tuni{n+m}{s}{t}}}\bigotimes_{i=1}^{s}\left(\ketbra{\type_{\overline{T}_i}}{\type_{\overline{T}_i}}\otimes\ketbra{\type_{\hat{T}_i}}{\type_{\hat{T}_i}}\right).$$
    Then $\TD(\rho,\rho^{\sigma}_{\unique_{s,t}})=O(st^2/2^{m})$.
\end{lemma}
\ifllncs
\noindent The proof of~\Cref{lem:tuni_info} can be found in~\Cref{app:closeness_to_rho_uni}. 
\else
\input{construction_proofs/tuni_info_proof}
\fi

\begin{remark}
\Cref{lem:tuni_info} implies that our construction is secure against specific adversaries with \emph{side-information}. That is, the adversary's registers and the registers acted on by the PRI are entangled in $\ket{\type_{T_i}}$.
\end{remark}

\begin{remark}
Re-ordering the registers of $\rho^{\sigma}_{\unique_{s,t}}$ in~\Cref{lem:tuni_info}, it can be written as $\sigma \otimes \rho_{\unique_{s,t}}$. That is, the state is \emph{unentangled}. Furthermore, the second register is fully scrambled to a uniform mixture of unique types and $\sigma$ is exactly the partial trace of $\rho$ (that you would get after tracing out the second register). We note that it is an analog of \emph{quantum one-time pad}~\cite{MTdW00}. In particular, for any $\rho_{AB}\in \cD(\C^2\otimes\C^2)$, it holds that $\Ex_{a,b}\left[ (I_{A} \otimes X^aZ^b)\rho_{AB}(I_{A} \otimes X^aZ^b)^\dagger\right] = \Tr_B(\rho_{AB})\otimes I_B/2$.
\end{remark}

\noindent Using~\Cref{lem:tuni_info}, we prove the following for Haar states.
\begin{lemma}
    \label{lem:haar_info}
    Let $n,m,s,t = \poly(\secparam)$. Let $G_{(f,\pi)}$ be as defined in \Cref{fig:info_prp_prs}. Let 
    $$\rho:= \E_{(f,\pi)\leftarrow(\mathcal{F}_{n+m},S_{n+m})} \bigotimes_{i=1}^{s} \left(\E_{\ket{\vartheta_i}\leftarrow \Haar_{n}} \ketbra{\vartheta_i}{\vartheta_i}^{\otimes t} \otimes \left(G_{(f,\pi)}\ketbra{\vartheta_i}{\vartheta_i} G_{(f,\pi)}^{\dagger}\right)^{\otimes t}\right),$$
    let 
    $$\sigma := \E_{\substack{(\overline{T}_1,\cdots,\overline{T}_s)\leftarrow \tuni{n}{s}{t}}} \bigotimes_{i=1}^{s} \left(\ketbra{\type_{\overline{T}_i}}{\type_{\overline{T}_i}}\right),$$ and let 
    $$\rho^\sigma_{\unique_{s,t}} := \E_{\substack{(T_1,\cdots,T_s)\leftarrow \tuni{n+m}{s}{t}\\ (\overline{T}_1,\cdots,\overline{T}_s)\leftarrow \tuni{n}{s}{t}}} \bigotimes_{i=1}^{s} \left( \left(\ketbra{\type_{\overline{T}_i}}{\type_{\overline{T}_i}}\right)\otimes \left(\ketbra{\type_{T_i}}{\type_{T_i}}\right) \right). $$ 
    Then $$\TD(\rho,\rho^\sigma_{\unique_{s,t}}) = O(s^2t^2/2^n+st^2/2^m).$$
\end{lemma}
\ifllncs
\noindent The proof of~\Cref{lem:haar_info} can be found in~\Cref{app:closeness_to_rho_uni}.
\else
\input{construction_proofs/haar_info_proof}
\fi

\noindent Hence, combining \Cref{lem:haar_info}, \Cref{lem:tuni_invar_aux} and \Cref{clm:observation:almostinvariance}, we get the desired result. 

\subsection{Main Results}
\label{sec:summary}

Combining \Cref{lem:tdis_proof,lem:single_proof,lem:haar_proof}, our construction is secure against inputs of the form: (1) distinct type state, (2) multiple copies of the same pure state, (3) i.i.d. Haar states. We state the formal theorem below:




\begin{theorem}[Main Theorem] \label{thm:main}
    Let $n,m,s,t,\ell,q = \poly(\secparam)$. Let $\qclass^{(\distinct)}_{{n,t,s,\ell,\secparam}}$, $\qclass_{n,q,\ell,\secparam}^{({\sf Haar})}$ and $\qclass^{(\single)}_{n,q,\ell,\secparam}$ be as defined in~\Cref{sec:distinct_type,sec:defs}. then, assuming the existence of post-quantum one-way functions, the construction of $\pri$ given in \Cref{fig:prp_prs} is $\qclass$-secure for $\qclass \in \biggl\{\qclass^{(\distinct)}_{{n,t,s,\ell,\secparam}}, \allowbreak \qclass_{n,q,\ell,\secparam}^{({\sf Haar})}, \allowbreak \qclass^{(\single)}_{n,q,\ell,\secparam}\biggl\}$.
\end{theorem}

\noindent Although we are not able to prove stronger security of our construction, we observe that our construction naturally mimics the steps of sampling a Haar isometry by truncating columns of a Haar unitary. We have the following conjecture.
\begin{conjecture}
Assuming the existence of post-quantum one-way functions, the construction of $\pri$ given in~\Cref{fig:prp_prs} is a strong invertible adaptive PRI (\Cref{def:strong_inv_PRI})
\end{conjecture}

%% file: construction_proofs/type_distinct_proof.tex
\begin{proof}
    Observe that $\ketbra{\type_{T_1}}{\type_{T_1}}\otimes\cdots\otimes\ketbra{\type_{T_s}}{\type_{T_s}}$ can be seen as a convex sum over $\ketbra{\vec{x_1},\cdots,\vec{x_s}}{\vec{x'_1},\cdots,\vec{x'_s}}$ where $\vec{x_i},\vec{x'_i}\in T_i$. By~\Cref{lem:type_struc}, this can equivalently be written as a sum over $\vec{x_i}\in T_i$, $\sigma_i\in S_t$, $\ketbra{\vec{x_1},\cdots,\vec{x_s}}{\sigma_1(\vec{x_1}),\cdots,\sigma_s(\vec{x'_s})}$. Applying $G_{(f,\pi)}$ can be seen as three operators, $C_{\ket{+^m}}$ which appends $\ket{+^m}$ to each entry, $O_f$ which maps $\ket{\vec{x_i}}$ to $\omega_p^{f(\vec{x_i})}\ket{\vec{x_i}}$ and $O_{\pi}$ which maps $\ket{\vec{x_i}}$ to $\ket{\vec{x_i}_{\pi}}$. Let us look at these operations one at a time. 
    \begin{enumerate}
        \item $C_{\ket{+^m}}$: We first apply to every register of the convex sum over $\ketbra{\vec{x_1},\cdots,\vec{x_s}}{\vec{x'_1},\cdots,\vec{x'_s}}$, this is equivalent to mapping each $\ketbra{\vec{x_1},\cdots,\vec{x_s}}{\sigma_1(\vec{x_1}),\cdots,\sigma_s(\vec{x_s})}$ to a sum over
        \[
        \ketbra{\vec{x_1}||\vec{a_1},\cdots,\vec{x_s}||\vec{a_s}}{\sigma_1(\vec{x_1})||\vec{a'_1},\cdots,\sigma_s(\vec{x_s})||\vec{a'_s}}
        \]
        for all $\vec{a_i},\vec{a'_i}\in\bit^{mt}$.
        \item $O_f$: When we apply $O_f$, we get a leading coefficient of $\omega_p^{f(\vec{x_i}||\vec{a_i})-f(\sigma_i(\vec{x_i})||\vec{a'_i})}$ on each term. Since $p>t$, taking expectation over $f$ would map this to zero unless $\vec{a'_i} = \sigma_i(\vec{a_i})$. Hence, the only terms left after this step are $\ketbra{\vec{x_1}||\vec{a_1},\cdots,\vec{x_s}||\vec{a_s}}{\sigma_1(\vec{x_1}||\vec{a_1}),\cdots,\sigma_s(\vec{x_s}||\vec{a_s})}$. 
        \item $O_{\pi}$: Notice that the state is a sum over all $(\vec{a_1},\cdots,\vec{a_s})$. With very high probability all elements of $(\vec{a_1},\cdots,\vec{a_s})$ are distinct. In this case, $\ketbra{\vec{x_1}||\vec{a_1},\cdots,\vec{x_s}||\vec{a_s}}{\sigma_1(\vec{x_1}||\vec{a_1}),\cdots,\sigma_s(\vec{x_s}||\vec{a_s})}$ has only distinct elements and applying $O_{\pi}$ maps it to random vectors with distinct elements. Taking sum over all $\sigma_i$, we get that this is equivalent to sampling from $\mathcal{T}_{\unique}$.
    \end{enumerate}

\noindent We now provide the formal details. We know that 
    $$\rho = \E_{f}\E_{\pi}\left[\bigotimes_{i=1}^{s}G_{(f,\pi)}^{\otimes t}\ketbra{\type_{T_i}}{\type_{T_i}}(G_{(f,\pi)}^{\dagger})^{\otimes t}\right].$$
Using \Cref{lem:type_struc}, we get
$$\rho = \frac{1}{(t!)^s}\sum_{\substack{\sigma_1,\cdots,\sigma_s\in S_t\\ (\vec{x_1},\cdots,\vec{x_t})\in(T_1,\cdots,T_s)}}\E_{f}\E_{\pi}\left[\bigotimes_{i=1}^{s}G_{(f,\pi)}^{\otimes t}\ketbra{\vec{x_i}}{\sigma_i(\vec{x_i})}(G_{(f,\pi)}^{\dagger})^{\otimes t}\right].$$
Using the fact that every $t$-fold tensor operator commutes with the permutation operator $P_{\sigma_i}$, we can simplify this to:
$$\rho = \frac{1}{(t!)^s}\sum_{\substack{\sigma_1,\cdots,\sigma_s\in S_t\\ (\vec{x_1},\cdots,\vec{x_t})\in(T_1,\cdots,T_s)}}\E_{f}\E_{\pi}\left[\bigotimes_{i=1}^{s}\left(G_{(f,\pi)}^{\otimes t}\ketbra{\vec{x_i}}{\vec{x_i}}(G_{(f,\pi)}^{\dagger})^{\otimes t}P_{\sigma_{i}}\right)\right].$$
Writing $G_{(f,\pi)} = O_{\pi} O_{f} (I\otimes H^{\otimes m}) C_{\ket{0^m}}$ (where $O_{\pi}$ refers to the unitary applying the permutation $\pi$, $O_{f}$ refers to the unitary applying the function $f$ and $C_{\ket{0^m}}$ refers to appending $\ket{0^m}$), 
\begin{multline*}
    \rho = \frac{1}{(t!)^s}\sum_{\substack{\sigma_1,\cdots,\sigma_s\in S_t\\ (\vec{x_1},\cdots,\vec{x_t})\in(T_1,\cdots,T_s)}}\E_{f}\E_{\pi}\left[\bigotimes_{i=1}^{s}\left(O_{\pi}^{\otimes t}O_{f}^{\otimes t}\left(\frac{1}{2^{mt}} \sum_{\substack{\vec{a_i}\in\set{0,1}^{mt}\\ \vec{a'_i}\in\set{0,1}^{mt}}}\ketbra{\vec{x_i}||\vec{a_i}}{\vec{x_i}||\vec{a'_i}}\right)(O_{f}^{\dagger})^{\otimes t}(O_{\pi}^{\dagger})^{\otimes t}P_{\sigma_i}\right)\right].
\end{multline*}
Applying $O_f$, 
\begin{multline*}
    \rho = \frac{1}{(t!)^s}\sum_{\substack{\sigma_1,\cdots,\sigma_s\in S_t\\ (\vec{x_1},\cdots,\vec{x_t})\in(T_1,\cdots,T_s)}}\E_{f}\E_{\pi}\left[\bigotimes_{i=1}^{s}\left(O_{\pi}^{\otimes t}\left(\frac{1}{2^{mt}}\omega_p^{f(\vec{x_i}||\vec{a_i})-f(\vec{x_i}||\vec{a'_i})}\times\right.\right.\right.\\
    \left.\left.\left.\sum_{\substack{\vec{a_i}\in\set{0,1}^{mt}\\ \vec{a'_i}\in\set{0,1}^{mt}}}\ketbra{\vec{x_i}||\vec{a_i}}{\vec{x_i}||\vec{a'_i}}\right)(O_{\pi}^{\dagger})^{\otimes t}P_{\sigma_i}\right)\right].
\end{multline*}
By linearity, we get
\begin{multline*}
    \rho = \frac{1}{2^{mst}(t!)^s}\sum_{\substack{\sigma_1,\cdots,\sigma_s\in S_t\\ (\vec{x_1},\cdots,\vec{x_t})\in(T_1,\cdots,T_s)\\ (\vec{a_1},\cdots,\vec{a_s})\in\set{0,1}^{mts}\\ (\vec{a'_1},\cdots,\vec{a'_s})\in\set{0,1}^{mts}}}\E_{f}\left[\omega_p^{\sum_{i=1}^s(f(\vec{x_i}||\vec{a_i})-f(\vec{x_i}||\vec{a'_i}))}\right]\times\\
    \E_{\pi}\left[\bigotimes_{i=1}^{s}\left(O_{\pi}^{\otimes t}\left(\ketbra{\vec{x_i}||\vec{a_i}}{\vec{x_i}||\vec{a'_i}}\right)(O_{\pi}^{\dagger})^{\otimes t}P_{\sigma_i}\right)\right].
\end{multline*}
Since sum of powers of a root of unity is zero, we have $$\Ex_{f}\left[\omega_p^{\sum_{i=1}^s(f(\vec{x_i}||\vec{a_i})-f(\vec{x_i}||\vec{a'_i}))}\right] = 0$$ except when $$\type((\vec{x_1}||\vec{a_1},\cdots,\vec{x_s}||\vec{a_s}))=\type((\vec{x_1}||\vec{a'_1},\cdots,\vec{x_s}||\vec{a'_s})) \bmod{p}.$$
Note that since $\type((\vec{x_1}||\vec{a_1},\cdots,\vec{x_s}||\vec{a_s})),\type((\vec{x_1}||\vec{a'_1},\cdots,\vec{x_s}||\vec{a'_s}))\in\mathbb{Z}^{2^{n+m}}_{st}$ and $st<p$, $$\type((\vec{x_1}||\vec{a_1},\cdots,\vec{x_s}||\vec{a_s}))=\type((\vec{x_1}||\vec{a'_1},\cdots,\vec{x_s}||\vec{a'_s})) \bmod{p}$$ iff $$\type((\vec{x_1}||\vec{a_1},\cdots,\vec{x_s}||\vec{a_s}))=\type((\vec{x_1}||\vec{a'_1},\cdots,\vec{x_s}||\vec{a'_s})) .$$

\noindent Also, note that since, $(T_1,\cdots,T_s)$ are distinct, $\vec{x_i}$ and $\vec{x_j}$ has distinct elements for $i\neq j$. Hence, no element of $\vec{x_i}||\vec{a_i}$ can be equal to $\vec{x_j}||\vec{a'_j}$ for any $\vec{a_i},\vec{a'_j}$ and $i\neq j$. Hence, we need $\type(\vec{x_i}||\vec{a_i}) = \type(\vec{x_i}||\vec{a'_i})$ for all $1\leq i\leq s$. Also, note that whenever this condition is true, we get $\Ex_{f}\left[\omega_p^{\sum_{i=1}^s(f(\vec{x_i}||\vec{a_i})-f(\vec{x_i}||\vec{a'_i}))}\right] = 1$. 
Hence, we get 
\begin{multline*}
    \rho = \frac{1}{2^{mst}(t!)^s}\sum_{\substack{\sigma_1,\cdots,\sigma_s\in S_t\\ (\vec{x_1},\cdots,\vec{x_t})\in(T_1,\cdots,T_s)\\ (\vec{a_1},\cdots,\vec{a_s})\in\set{0,1}^{mts}}}\sum_{\substack{(\vec{a'_1},\cdots,\vec{a'_s})\in\set{0,1}^{mts}\\ \forall i\in [s], \type(\vec{x_i}||\vec{a_i}) = \type(\vec{x_i}||\vec{a'_i})}}
    \E_{\pi}\left[\bigotimes_{i=1}^{s}\left(O_{\pi}^{\otimes t}\left(\ketbra{\vec{x_i}||\vec{a_i}}{\vec{x_i}||\vec{a'_i}}\right)(O_{\pi}^{\dagger})^{\otimes t}P_{\sigma_i}\right)\right].
\end{multline*}
Next we for each fixed $(\vec{x_1},\cdots,\vec{x_t})\in(T_1,\cdots,T_s)$, we define 
\begin{multline*}
    \rho_{(\vec{x_1},\cdots,\vec{x_s})} = \sum_{\substack{\sigma_1,\cdots,\sigma_s\in S_t\\ (\vec{a_1},\cdots,\vec{a_s})\in\set{0,1}^{mts}}}\sum_{\substack{(\vec{a'_1},\cdots,\vec{a'_s})\in\set{0,1}^{mts}\\ \forall i\in [s], \type(\vec{x_i}||\vec{a_i}) = \type(\vec{x_i}||\vec{a'_i})}}
    \E_{\pi}\left[\bigotimes_{i=1}^{s}\left(O_{\pi}^{\otimes t}\left(\ketbra{\vec{x_i}||\vec{a_i}}{\vec{x_i}||\vec{a'_i}}\right)(O_{\pi}^{\dagger})^{\otimes t}P_{\sigma_i}\right)\right].
\end{multline*}
Then, we get 
$$\rho = \frac{1}{2^{mst}(t!)^s}\sum_{(\vec{x_1},\cdots,\vec{x_t})\in(T_1,\cdots,T_s)}\rho_{(\vec{x_1},\cdots,\vec{x_s})}.$$
We will show that for each of these $\rho_{(\vec{x_1},\cdots,\vec{x_s})}$ can be shown to be close to some constant times $\rho_{\unique_{s,t}}$. 
We start by defining $\xi_{(\vec{x_1},\cdots,\vec{x_s})}$ as 
\begin{multline*}
    \xi_{(\vec{x_1},\cdots,\vec{x_s})} = \sum_{\substack{\sigma_1,\cdots,\sigma_s\in S_t\\ (\vec{a_1},\cdots,\vec{a_s})\in\set{0,1}^{mts}\\ (\vec{x_1}||\vec{a_1},\cdots,\vec{x_q}||\vec{a_s})\text{ has distinct elements}}}\sum_{\substack{(\vec{a'_1},\cdots,\vec{a'_s})\in\set{0,1}^{mts}\\ \forall i\in [s], \type(\vec{x_i}||\vec{a_i}) = \type(\vec{x_i}||\vec{a'_i})}}\\
    \E_{\pi}\left[\bigotimes_{i=1}^{s}\left(O_{\pi}^{\otimes t}\left(\ketbra{\vec{x_i}||\vec{a_i}}{\vec{x_i}||\vec{a'_i}}\right)(O_{\pi}^{\dagger})^{\otimes t}P_{\sigma_i}\right)\right].
\end{multline*}
Notice that whenever $(\vec{x_1}||\vec{a_1},\cdots,\vec{x_s}||\vec{a_s})$ has distinct elements, $(\vec{x_1}||\vec{a'_1},\cdots,\vec{x_s}||\vec{a'_s})$ also has distinct elements, because $\forall i\in [s], \type(\vec{x_i}||\vec{a_i}) = \type(\vec{x_i}||\vec{a'_i})$. Also notice that whenever $(\vec{x_1}||\vec{a_1},\cdots,\vec{x_s}||\vec{a_s})$ has a collision, $(\vec{x_1}||\vec{a'_1},\cdots,\vec{x_s}||\vec{a'_s})$ also has a collision. Note that applying $O_{\pi}^{\otimes st}$ and permuting $(\vec{x_1}||\vec{a'_1},\cdots,\vec{x_s}||\vec{a'_s})$ still perserves this property. Hence, $\xi_{(\vec{x_1},\cdots,\vec{x_s})}$ and $\eta_{(\vec{x_1},\cdots,\vec{x_s})} = \rho_{(\vec{x_1},\cdots,\vec{x_s})}-\xi_{(\vec{x_1},\cdots,\vec{x_s})}$ belong to orthogonal subspaces. 
We now simplify $\xi_{(\vec{x_1},\cdots,\vec{x_s})}$. We know that 
\begin{multline*}
    \xi_{(\vec{x_1},\cdots,\vec{x_s})} = \sum_{\substack{\sigma_1,\cdots,\sigma_s\in S_t\\ (\vec{a_1},\cdots,\vec{a_q})\in\set{0,1}^{mts}\\ (\vec{x_1}||\vec{a_1},\cdots,\vec{x_s}||\vec{a_s})\text{ has distinct elements}}}\sum_{\substack{(\vec{a'_1},\cdots,\vec{a'_s})\in\set{0,1}^{mts}\\ \forall i\in [s], \type(\vec{x_i}||\vec{a_i}) = \type(\vec{x_i}||\vec{a'_i})}}\\
    \E_{\pi}\left[\bigotimes_{i=1}^{s}\left(O_{\pi}^{\otimes t}\left(\ketbra{\vec{x_i}||\vec{a_i}}{\vec{x_i}||\vec{a'_i}}\right)(O_{\pi}^{\dagger})^{\otimes t}P_{\sigma_i}\right)\right].
\end{multline*}

Note that, whenever $\type(\vec{x_i}||\vec{a_i}) = \type(\vec{x_i}||\vec{a'_i})$, we can write $\vec{x_i}||\vec{a'_i} = \tau_i(\vec{x_i}||\vec{a_i})$ for some $\tau_i\in S_t$. Let the set of values of $\tau_i\in S_t$, such that $\vec{x_i}||\vec{a'_i} = \tau_i(\vec{x_i}||\vec{a_i})$ for some $\vec{a_i},\vec{a'_i}$ be denoted by $A_i$. Then each of the elements in $\tau_i\in A_i$ just need to map $\vec{x_i}$ to $\vec{x_i}$, hence, the size of $A_i$ is $(\prod_{w_i\in T_i} w_i!)$. Also, notice that $A_i$ doesn't depend on $\vec{a_i}$ or $\vec{a'_i}$. Also, notice that for each $\vec{a_i}$ with distinct elements and $\tau_i\in A_i$, there's a distinct $\vec{a'_i}$. Hence, 
\begin{multline*}
    \xi_{(\vec{x_1},\cdots,\vec{x_s})} = \sum_{\substack{\sigma_1,\cdots,\sigma_s\in S_t\\ (\vec{a_1},\cdots,\vec{a_q})\in\set{0,1}^{mts}\\ (\vec{x_1}||\vec{a_1},\cdots,\vec{x_q}||\vec{a_s})\text{ has distinct elements}}}\sum_{(\tau_1,\cdots,\tau_s)\in (A_1,\cdots,A_s)}\\
    \E_{\pi}\left[\bigotimes_{i=1}^{s}\left(O_{\pi}^{\otimes t}\left(\ketbra{\vec{x_i}||\vec{a_i}}{\vec{x_i}||\vec{a_i}}P_{\tau_i}\right)(O_{\pi}^{\dagger})^{\otimes t}P_{\sigma_i}\right)\right].
\end{multline*}
Again, since $t$-fold unitaries commute with $P_{\tau_i}$, and $P_{\tau_i}P_{\sigma_i} = P_{\sigma_i\tau_i}$, we get 
\begin{multline*}
    \xi_{(\vec{x_1},\cdots,\vec{x_s})} = \sum_{\substack{\sigma_1,\cdots,\sigma_s\in S_t\\ (\vec{a_1},\cdots,\vec{a_s})\in\set{0,1}^{mts}\\ (\vec{x_1}||\vec{a_1},\cdots,\vec{x_s}||\vec{a_s})\text{ has distinct elements}}}\sum_{(\tau_1,\cdots,\tau_s)\in (A_1,\cdots,A_s)}\\
    \E_{\pi}\left[\bigotimes_{i=1}^{s}\left(O_{\pi}^{\otimes t}\left(\ketbra{\vec{x_i}||\vec{a_i}}{\vec{x_i}||\vec{a_i}}\right)(O_{\pi}^{\dagger})^{\otimes t}P_{\sigma_i\tau_i}\right)\right].
\end{multline*}
Notice that since $\sigma_i$ is summing over all of $S_t$, $P_{\sigma_i\tau_i}$ is distributed the same as $P_{\sigma_i}$. Define $\gamma = \left(\prod_{i=1}^{s}\left(\prod_{w_i\in\supp(T_i)} \freq{T_i}{w_i}!\right)\right)$. Hence, using the size of $A_i$ is $(\prod_{w_i\in\supp(T_i)} \freq{T_i}{w_i}!)$, we get, 
\begin{multline*}
    \xi_{(\vec{x_1},\cdots,\vec{x_s})} = \gamma \sum_{\substack{\sigma_1,\cdots,\sigma_s\in S_t\\ (\vec{a_1},\cdots,\vec{a_s})\in\set{0,1}^{mts}\\ (\vec{x_1}||\vec{a_1},\cdots,\vec{x_s}||\vec{a_s})\text{ has distinct elements}}}\E_{\pi}\left[\bigotimes_{i=1}^{s}\left(O_{\pi}^{\otimes t}\left(\ketbra{\vec{x_i}||\vec{a_i}}{\vec{x_i}||\vec{a_i}}\right)(O_{\pi}^{\dagger})^{\otimes t}P_{\sigma_i}\right)\right].
\end{multline*}
Applying $O_{\pi}$ and taking expectation, we get 
\begin{multline*}
    \xi_{(\vec{x_1},\cdots,\vec{x_s})} = \gamma \sum_{\substack{\sigma_1,\cdots,\sigma_s\in S_t\\ (\vec{a_1},\cdots,\vec{a_s})\in\set{0,1}^{mts}\\ (\vec{x_1}||\vec{a_1},\cdots,\vec{x_s}||\vec{a_s})\text{ has distinct elements}}}
    \E_{\substack{(\vec{z_1},\cdots,\vec{z_s})\in\set{0,1}^{(n+m)st}\\ (\vec{z_1},\cdots,\vec{z_s})\text{ has distinct elements}}}\left[\bigotimes_{i=1}^{s}\left(\ketbra{\vec{z_i}}{\vec{z_i}}P_{\sigma_i}\right)\right].
\end{multline*}
Using~\Cref{lem:type_struc}, we get, 
$$\xi_{(\vec{x_1},\cdots,\vec{x_s})} = \gamma\sum_{\substack{(\vec{a_1},\cdots,\vec{a_s})\in\set{0,1}^{mts}\\ (\vec{x_1}||\vec{a_1},\cdots,\vec{x_s}||\vec{a_s})\text{ has distinct elements}}}\left(\rho_{\unique_{s,t}}\right).$$
Let $\vec{x_i}$ have $v$ distinct elements with $t_1,\cdots,t_v$ copies. Then $t_1+\cdots+t_v = t$. Then the number of values of $\vec{a_i}$ such that $\vec{x_i}||\vec{a_i}$ has distinct elements is $\prod_{i=1}^v(2^m\cdots (2^{m}-i+1)) = 2^{mt}(1-O(t^2/2^m))$. Hence, we have 
$$\xi_{(\vec{x_1},\cdots,\vec{x_s})} = \gamma 2^{mts} (1-O(t^2/2^m))^s\left(\rho_{\unique_{s,t}}\right).$$
Substituting, we get 
\begin{multline*}
    \rho = \frac{1}{2^{mst}(t!)^s}\sum_{(\vec{x_1},\cdots,\vec{x_t})\in(T_1,\cdots,T_s)}\gamma 2^{mts} (1-O(t^2/2^m))^s\left(\rho_{\unique_{s,t}}\right)\\
    +\frac{1}{2^{mst}(t!)^s}\sum_{(\vec{x_1},\cdots,\vec{x_t})\in(T_1,\cdots,T_s)}\eta_{(\vec{x_1},\cdots,\vec{x_s})}.
\end{multline*}
Simplifying, we get 
\begin{multline*}
    \rho = \frac{\gamma}{(t!)^s}\sum_{(\vec{x_1},\cdots,\vec{x_t})\in(T_1,\cdots,T_s)} (1-O(st^2/2^m))\left(\rho_{\unique_{s,t}}\right)\\
    +\frac{1}{2^{mst}(t!)^s}\sum_{(\vec{x_1},\cdots,\vec{x_t})\in(T_1,\cdots,T_s)}\eta_{(\vec{x_1},\cdots,\vec{x_s})}.
\end{multline*}
Notice that the number of values for each $\vec{x_i}$ is $\frac{t!}{\prod_{w_i\in \supp(T_i)} \freq{T_i}{w_i}!}$.\\ Hence, using $\gamma = \left(\prod_{i=1}^{s}\left(\prod_{w_i\in\supp(T_i)} \freq{T_i}{w_i}!\right)\right)$, 
$$\rho = (1-O(st^2/2^m))\rho_{\unique_{s,t}}+\frac{1}{2^{mst}(t!)^s}\sum_{(\vec{x_1},\cdots,\vec{x_t})\in(T_1,\cdots,T_s)}\eta_{(\vec{x_1},\cdots,\vec{x_s})}.$$
Notice that $\sum_{(\vec{x_1},\cdots,\vec{x_t})\in(T_1,\cdots,T_s)}\eta_{(\vec{x_1},\cdots,\vec{x_s})}$ is orthogonal to $\rho_{\unique_{s,t}}$. 
Hence, we get that the trace distance between $\rho$ and $\rho_{\unique_{s,t}}$ is $O(st^2/2^{m})$.
\end{proof}

%% file: construction_proofs/comp_outer_query_proof.tex
\begin{proof}
    We again see $G_{(f,\pi)}$ as $O_{\pi}O_{f}C_{\ket{+^m}}$. Acting on $\bigotimes_{i=1}^{q}\ketbra{x_i}{x'_i}$, applying $C_{\ket{+^m}}$ results in $\bigotimes_{i=1}^{q}\ketbra{x_i||a_i}{x'_i||a'_i}$ for each $a_i,a'_i\in\bit^m$. Applying $O_f$ gives us a leading coefficient of $\omega_p^{\sum_i (f(x_i||a_i)-f(x'_i||a'_i))}$ which always goes to zero when we take expectation over $f$ because $\sum_i (f(x_i||a_i)-f(x'_i||a'_i))\neq 0$ for all values of $a_i,a'_i$'s. Hence, we get that the resulting matrix is also zero.

\noindent We now provide the formal details. Let $$\rho = \E_{(f,\pi)\leftarrow(\mathcal{F}_{2^{n+m},p},S_{2^{n+m}})} \left[\left(G_{(f,\pi)}^{\otimes q}\right)\left(\bigotimes_{i=1}^{q}\ketbra{x_i}{x'_i}\right)\left(G_{(f,\pi)}^{\dagger}\right)^{\otimes q}\right]$$
    Writing $G_{(f,\pi)}$ as $O_{\pi}O_{f}(I\otimes H^{m})C_{\ket{0^m}}$.
    $$\rho = \E_{\pi} \E_{f} \left[O^{\otimes q}_{\pi}O^{\otimes q}_{f}\left(\frac{1}{2^{mq}}\sum_{\substack{\vec{a}\in\set{0,1}^{mq}\\ \vec{a'}\in\set{0,1}^{mq}}} \ketbra{\vec{x}||\vec{a}}{\vec{x'}||\vec{a'}}\right)(O^{\otimes q}_{f})^{\dagger}(O^{\otimes q}_{\pi})^{\dagger}\right].$$
    Applying $O_f$,
    $$\rho = \E_{\pi} \E_{f} \left[O^{\otimes q}_{\pi}\left(\frac{1}{2^{mq}}\sum_{\substack{\vec{a}\in\set{0,1}^{mq}\\ \vec{a'}\in\set{0,1}^{mq}}} \omega_p^{f(\vec{x}||\vec{a})-f(\vec{x'}||\vec{a'})}\ketbra{\vec{x}||\vec{a}}{\vec{x'}||\vec{a'}}\right)(O^{\otimes q}_{\pi})^{\dagger}\right].$$
    Then by linearity, we get, 
    $$\rho = \sum_{\substack{\vec{a}\in\set{0,1}^{mq}\\ \vec{a'}\in\set{0,1}^{mq}}}\E_{\pi} \E_{f}\left[\omega_p^{f(\vec{x}||\vec{a})-f(\vec{x'}||\vec{a'})}\right] \left[O^{\otimes q}_{\pi}\left(\frac{1}{2^{mq}} \ketbra{\vec{x}||\vec{a}}{\vec{x'}||\vec{a'}}\right)(O^{\otimes q}_{f})^{\dagger}\right].$$
    Note that since $\type(\vec{x})\neq\type(\vec{x'})$, for any $\vec{a},\vec{a'}\in\set{0,1}^{mq}$, $\type(\vec{x}||\vec{a})\neq\type(\vec{x'}||\vec{a'})$. Also, since the sum of powers of a root of unity is $0$, then for any $\vec{a},\vec{a'}\in\set{0,1}^{mq}$, $\E_{f}\left[\omega_p^{f(\vec{x}||\vec{a})-f(\vec{x'}||\vec{a'})}\right] = 0$. Hence, we get $\rho = 0$, as required.
\end{proof}

%% file: construction_proofs/multicopy_proof.tex
\begin{proof}
    Since $\ketbra{\phi}{\phi}^{\otimes q}$ is in the symmetric subspace, we can write $\ketbra{\phi}{\phi}^{\otimes q} = \sum_{T,T'}\alpha_{T,T'}\ketbra{\type_T}{\type_{T'}}$. By \Cref{lem:type_outer}, the only terms remaining in this sum are when $T=T'$. The security on this is just implied by \Cref{lem:tdis_info}.

    \noindent Formally, since $\ketbra{\phi}{\phi}^{\otimes q}$ is in the symmetric subspace, we can write $\ketbra{\phi}{\phi}^{\otimes q} = \sum_{T,T'}\alpha_{T,T'}\ketbra{\type_T}{\type_{T'}}$. Hence, we get 
    $$\rho = \E_{(f,\pi)\leftarrow(\mathcal{F}_{2^{n+m},p},S_{2^{n+m}})} \left[\left(G_{(f,\pi)}^{\otimes q}\right)\sum_{T,T'}\alpha_{T,T'}\ketbra{\type_T}{\type_{T'}}\left(G_{(f,\pi)}^{\dagger}\right)^{\otimes q}\right].$$
    By linearity and \Cref{lem:type_outer}, we get
    $$\rho = \sum_{T}\alpha_{T,T}\E_{(f,\pi)\leftarrow(\mathcal{F}_{2^{n+m},p},S_{2^{n+m}})} \left[\left(G_{(f,\pi)}^{\otimes q}\right)\ketbra{\type_T}{\type_{T}}\left(G_{(f,\pi)}^{\dagger}\right)^{\otimes q}\right].$$
    Using \Cref{lem:tdis_info}, the following state is $O((\sum_{T}\alpha_{T,T})q^2/2^m)$ away from $\rho$, 
    $$\rho' = \sum_{T}\alpha_{T,T}\rho_{\unique_{1,q}}.$$
    Note that, since $\sum_{T} \alpha_{T,T} = 1$, we get that the distance between $\rho$ and $\rho_{\unique_{1,q}}$ is $O(q^2/2^m)$.
\end{proof}

%% file: construction_proofs/tuni_info_proof.tex
\begin{proof}
    We know from \Cref{lem:type_struc}, $\bigotimes_{i=1}^{s} \ketbra{\type_{T_i}}{\type_{T_i}}$ can be written as a sum over $\vec{x_i}\in T_i$, $\sigma_i\in S_{2t}$, $\bigotimes_{i=1}^{s} \ketbra{\vec{x_i}}{\sigma_i(\vec{x_i})}$. 
    Let each $\vec{x_i}$ (containing $2t$ elements) is a concatenation of $\vec{c_i}$ and $\vec{d_i}$ (where each contains $t$ elements). Note that, all elements of $\vec{x_i}$ are distinct, hence $\vec{c_i}$ and $\vec{d_i}$ also contain distinct elements. Notice that if any of the $\sigma_i$'s maps any of the first $t$ elements to the last $t$ elements, then by \Cref{lem:outer_comp_query}, we get $0$. Hence, we get that each $\sigma_i$ can be written as a combination $\sigma^1_i\in S_t$ and $\sigma^2_i\in S_t$ where $\sigma^1_i$ is applied to the first $t$ elements of $\vec{x_i}$ and $\sigma^2_i$ is applied to the last $t$ elements of $\vec{x_i}$. Hence, we get that the input is just of the form $\bigotimes_{i=1}^{s} \ketbra{\vec{c_i}}{\sigma^1_i(\vec{c_i})}\otimes\ketbra{\vec{d_i}}{\sigma^2_i(\vec{d_i})}$, with the construction being applied to the second half of each state. Notice that summing over all $\sigma^1_i,\sigma^2_i$, there are type states too. We know that the effect of the construction on the second half is very close to $\rho_{\unique}$ by \Cref{lem:tdis_info}. Hence, it gets unentangled from the first half and we get the desired result.

    \noindent Formally, we start by analysing $\rho$, then 
    $$\rho= \E_{(f,\pi)\leftarrow(\mathcal{F}_{n+m},S_{n+m})} \bigotimes_{i=1}^{s} \left(\left(I_{nt}\otimes G_{(f,\pi)^{\otimes t}}\right)\ketbra{\type_{T_i}}{\type_{T_i}}\left(I_{nt}\otimes G_{(f,\pi)^{\otimes t}}\right)^{\dagger}\right).$$
    Using \Cref{lem:type_struc}, 
    \begin{multline*}
        \rho = \E_{\substack{(f,\pi)\leftarrow(\mathcal{F}_{n+m},S_{n+m})\\ (\vec{x_1},\cdots,\vec{x_s})\in (T_1,\cdots,T_s)}}\sum_{\sigma_1,\cdots,\sigma_s\in S_{2t}} \left[\bigotimes_{i=1}^{s} \left(\left(I_{nt} \otimes \left(G_{(f,\pi)}\right)^{\otimes t}\right)\right.\right.\\
        \left.\left.\ketbra{\vec{x_i}}{\sigma_i(\vec{x_i})} \left(I_{nt}\otimes \left(G_{(f,\pi)}^{\dagger}\right)^{\otimes t}\right)\right)\right].
    \end{multline*}
    Writing each $\vec{x_i}$ (containing $2t$ elements) as a concatenation of $\vec{c_i}$ and $\vec{d_i}$ (where each contains $t$ elements). Note that, all elements of $\vec{x_i}$ are distinct, hence $\vec{c_i}$ and $\vec{d_i}$ also contain distinct elements. Notice that if any of the $\sigma_i$'s maps any of the first $t$ elements to the last $t$ elements, then by \Cref{lem:outer_comp_query}, we get $0$. Hence, we get that each $\sigma_i$ can be written as a combination $\sigma^1_i\in S_t$ and $\sigma^2_i\in S_t$ where $\sigma^1_i$ is applied to the first $t$ elements of $\vec{x_i}$ and $\sigma^2_i$ is applied to the last $t$ elements of $\vec{x_i}$. Hence, we get 
    \begin{multline*}
        \rho = \E_{\substack{(f,\pi)\leftarrow(\mathcal{F}_{n+m},S_{n+m})\\ ((\vec{c_1},\vec{d_1}),\cdots,(\vec{c_s},\vec{d_s}))\in (T_1,\cdots,T_s)}}
        \sum_{\substack{\sigma^1_1,\cdots,\sigma^1_s\in S_{t}\\ \sigma^2_1,\cdots,\sigma^2_s\in S_{t}}}
        \left[\bigotimes_{i=1}^{s} \left(\ketbra{\vec{c_i}}{\sigma^1_i(\vec{c_i})}\otimes \left(G_{(f,\pi)}\right)^{\otimes t}\ketbra{\vec{d_i}}{\sigma^2_i(\vec{d_i})} \left(G_{(f,\pi)}^{\dagger}\right)^{\otimes t}\right)\right].
    \end{multline*}
    Note that using \Cref{lem:type_struc}, we get 
    \begin{multline*}
        \rho = \E_{\substack{(f,\pi)\leftarrow(\mathcal{F}_{n+m},S_{n+m})\\ (\overline{T_1},\cdots,\overline{T_s})\subset (T_1,\cdots,T_s)\\ \forall i\in [s], \hamming(\overline{T}_i) = t\\ \forall i\in [s], \hat{T_i} = T_i\setminus \overline{T_i} }}\left[\bigotimes_{i=1}^{s} \left(\ketbra{\type_{\overline{T_i}}}{\type_{\overline{T_i}}}\bigotimes\left(G_{(f,\pi)}\right)^{\otimes t}\ketbra{\type_{\hat{T_i}}}{\type_{\hat{T_i}}} \left(G_{(f,\pi)}^{\dagger}\right)^{\otimes t}\right)\right],
    \end{multline*}
    where $\hat{T_i} = T_i\setminus \overline{T_i}$ denotes the type vector $\hat{T_i}$ such that $\setT(\hat{T_i})\cup\setT(\overline{T_i}) = \setT(T_i)$ and $\hamming(\hat{T_i})+\hamming(\overline{T_i})=\hamming(T_i)$. Notice that $(\hat{T_1},\cdots,\hat{T_s})\in\tdis{n}{s}{t}$, hence by \Cref{lem:tdis_info}, we get that $\rho$ is at a distance of $O(st^2/2^{m})$ from the following state, 
    $$\rho^{\sigma}_{\unique_{s,t}}= \E_{\substack{(\overline{T}_1,\cdots,\overline{T}_s)\subset (T_1,\cdots,T_s)\\ \forall i\in [s], \hamming(\overline{T}_i)=t\\
    (\hat{T_1},\cdots,\hat{T_1})\leftarrow\tuni{n+m}{s}{t}}}\left(\ketbra{\type_{\overline{T}_i}}{\type_{\overline{T}_i}}\otimes\ketbra{\type_{\hat{T}_i}}{\type_{\hat{T}_i}}\right).$$
    Hence, we get the desired result.
\end{proof}

%% file: construction_proofs/haar_info_proof.tex
\begin{proof}
    Note that we have $2t$ copies of $s$ i.i.d. sampled Haar states. Then by \Cref{fact:avg-haar-random}, this can be seen as i.i.d. sampling $T_1,...,T_s$ each over $2t$, elements. With very high probability, these $2st$ elements do not have any collisions, hence with very high probability $(T_1,...,T_s)$ is in $\mathcal{T}_{\unique}$. The security on each type in $\mathcal{T}_{\unique}$ is shown by \Cref{lem:tuni_info}.
    
    \noindent Formally, we prove this using the hybrid method. 
    \paragraph{Hybrid $1$.} Sample a random function $f$ from $\mathcal{F}_{n+m}$ and a random permutation $\pi$ from $S_{n+m}$. Sample $2t$ copies of $s$ Haar random $n$-qubit states, $\ket{\vartheta_1},\cdots,\ket{\vartheta_s}$. Output $$\bigotimes_{i=1}^{s} \left(\ketbra{\vartheta_i}{\vartheta_i}^{\otimes t} \otimes \left(G_{(f,\pi)}\ketbra{\vartheta_i}{\vartheta_i} G_{(f,\pi)}^{\dagger}\right)^{\otimes t}\right).$$
    \paragraph{Hybrid $2$.} Sample a random function $f$ from $\mathcal{F}_{n+m}$ and a random permutation $\pi$ from $S_{n+m}$. Sample $(T_1,\ldots,T_s)$ uniformly at random from $\tuni{n}{s}{2t}$. Output $$\bigotimes_{i=1}^{q} \left(\left(I_{nt} \otimes \left(G_{(f,\pi)}\right)^{\otimes t}\right)\ketbra{\type_{T_i}}{\type_{T_i}} \left(I_{nt}\otimes \left(G_{(f,\pi)}^{\dagger}\right)^{\otimes t}\right)\right).$$
    \paragraph{Hybrid $3$.} Sample $(T_1,\ldots,T_s)$ uniformly at random from $\tuni{n}{s}{t}$ and sample $(\overline{T}_1,\ldots,\overline{T}_s)$ uniformly at random from $\tuni{n+m}{s}{t}$. Output 
    $$\bigotimes_{i=1}^{s} \left(\ketbra{\type_{T_i}}{\type_{T_i}}\otimes\ketbra{\type_{\overline{T}_i}}{\type_{\overline{T}_i}}\right).$$
    \begin{lemma}
    \label{lem:haar_input_1}
        The trace distance between Hybrid $1$ and Hybrid $2$ is $O(s^2t^2/2^n)$.
    \end{lemma}
    \begin{proof}
        This just follows from \Cref{lem:tuni_haar_dis}.
    \end{proof}
    
    \begin{lemma}
        \label{lem:haar_input_2}
        The trace distance between the outputs of Hybrid $2$ and Hybrid $3$ is $O(st^2/2^{m})$.
    \end{lemma}
    \begin{proof}
        Let $\rho$ be the output of Hybrid $2$. Hence, 
        $$\rho = \E_{\substack{(f,\pi)\leftarrow(\mathcal{F}_{n+m},S_{n+m})\\ (T_1,\ldots,T_s)\leftarrow\tuni{n}{s}{2t}}}\bigotimes_{i=1}^{s} \left(\left(I_{nt} \otimes \left(G_{(f,\pi)}\right)^{\otimes t}\right)\ketbra{\type_{T_i}}{\type_{T_i}} \left(I_{nt}\otimes \left(G_{(f,\pi)}^{\dagger}\right)^{\otimes t}\right)\right).$$
        Using \Cref{lem:tuni_info}, we get that the following state is at trace distance $O(st^2/2^m)$, 
        $$\sigma = \E_{(T_1,\cdots,T_s)\leftarrow\tuni{n}{s}{2t}}\E_{\substack{(\overline{T}_1,\cdots,\overline{T}_s)\subset (T_1,\cdots,T_s)\\ \forall i\in [s], \hamming(\overline{T}_i)=t\\
        (\hat{T_1},\cdots,\hat{T_1})\leftarrow\tuni{n+m}{s}{t}}}\left(\ketbra{\type_{\overline{T}_i}}{\type_{\overline{T}_i}}\otimes\ketbra{\type_{\hat{T}_i}}{\type_{\hat{T}_i}}\right).$$
        Notice that in the state, we could equivalently pick $(\overline{T}_1,\cdots,\overline{T}_s)$ directly from $\tuni{n}{s}{t}$. Hence, 
        $$\sigma = \E_{\substack{(\overline{T}_1,\cdots,\overline{T}_s)\leftarrow\tuni{n}{s}{t}\\
        (\hat{T_1},\cdots,\hat{T_1})\leftarrow\tuni{n+m}{s}{t}}}\left(\ketbra{\type_{\overline{T}_i}}{\type_{\overline{T}_i}}\otimes\ketbra{\type_{\hat{T}_i}}{\type_{\hat{T}_i}}\right).$$
        The above is exactly the output of Hybrid 3. Hence, the trace distance between the outputs of Hybrid $2$ and Hybrid $3$ is $O(s^2t^2/2^{m})$.
    \end{proof}
    \noindent Hence, combining, we get $\TD(\rho,\rho^\sigma_{\unique_{s,t}}) = O(s^2t^2/2^n+st^2/2^m).$
\end{proof}

%% file: applications.tex
\section{Applications}
\label{sec:applications}
We explore the cryptographic applications of pseudorandom isometries. Notably, some applications in this section only require invertible $\qclass$-secure (\Cref{def:invertible}), for classes of ${\cal Q}$ which can be initiated by post-quantum one-way functions, as we showed in~\Cref{sec:construction}.

\ifllncs
    In~\Cref{sec:qmacs}, we present quantum message authentication codes. In~\Cref{sec:prs_extension}, we present length extension theorems. Multi-copy secure encryption schemes and succinct quantum commitments are deferred to~\Cref{sec:multicopy_enc,sec:commitments}.
\else
    In~\Cref{sec:PRIimplyPRG_PRFSG}, we show that PRIs imply other quantum pseudorandom primitives. In~\Cref{sec:multicopy_enc}, we present multi-copy secure encryption schemes. In~\Cref{sec:commitments}, we present succinct quantum commitments. In~\Cref{sec:qmacs}, we present message authentication codes for quantum data. In~\Cref{sec:prs_extension}, we present length extension transformations for pseudorandom state generators.
\fi

\subsection{PRI implies PRSG and PRFSG}
\label{sec:PRIimplyPRG_PRFSG}

\begin{theorem}[PRI implies PRSG and PRFSG] \label{thm:PRIimplyPRFS}
Assuming $(n,n+m)$-$\qclass_{\cmp}$-pseudorandom isometries exist, there exist an $(n+m)$-$\mathrm{PRSG}$ and a selectively-secure $(n,n+m)$-$\mathrm{PRFSG}$.
\end{theorem}
\begin{proof}
Let $\pri$ be an $(n,n+m)$-$\qclass_{\cmp}$-PRI. The state generation algorithm of $\mathrm{PRSG}$ on input $k\in\bit^\secp$ is defined as $\pri_k\ket{0^n}$. The pseudorandomness of $\mathrm{PRSG}$ follows from invoking the security of $\pri$. The construction of $\mathrm{PRFSG}$ $F$ is the following: on input $k\in\bit^\secp$ and $x\in\bit^n$, append $\ket{0^n}$ and apply CNOT on $\ket{x}\ket{0^n}$ to get $\ket{x}\ket{x}$, and then output $\ket{x}\otimes\pri_k\ket{x}$. We prove the selective security of $F$ via a reduction. Suppose there exists a QPT adversary $\adversary$, polynomials $q(\cdot),t(\cdot)$ and a set of indices $\set{x_1,x_2,\dots,x_q}$ where $x_i\in\bit^n$ and $q(\secp) = \poly(\secp)$ such that
\begin{multline*}
\Bigg|
\Pr_{k\gets\bit^\secp}[\adversary_\secp(x_1,\dots,x_q,F(k,x_1)^{\otimes t},\dots,F(k,x_q)^{\otimes t}) = 1] - \\
\Pr_{\ket{\vartheta_1},\dots,\ket{\vartheta_q}\gets\Haar_{n+m}}[\adversary_\secp(x_1,\dots,x_q,\ket{\vartheta_1}^{\otimes t},\dots,\ket{\vartheta_q}^{\otimes t}) = 1]
\Bigg| 
\ge \nu(\secp),
\end{multline*}
where $\nu(\secp)$ is non-negligible. We construct a distinguisher $\distinguisher$ that uses $\adversary$ to break the security of the underlying PRI. Upon receiving queries $\set{x_1,x_2,\dots,x_q}$ from $\adversary$, the distinguisher $\distinguisher$ first uses CNOTs to generate $\bigotimes_{i=1}^q \ket{x_i}^{t+1}$. Then $\distinguisher$ uses its oracle access to $\mathcal{O}$, which is either $\pri_k$ or a Haar isometry, to reply $\bigotimes_{i=1}^q \ket{x_i} \otimes (\mathcal{O}\ket{x_i})^{\otimes t}$. Then $\distinguisher$ outputs whatever $\adversary$ outputs. Hence, the distinguishing advantage of $\distinguisher$ is
\begin{multline*}
\Bigg| 
\Pr_{k\gets\bit^\secp}[\adversary_\secp(x_1,\dots,x_q,F(k,x_1)^{\otimes t},\dots,F(k,x_q)^{\otimes t}) = 1] - \\
\Pr_{\haarisometry\gets\overline{\Haar_{n,n+m}}}[\adversary_\secp(x_1,\dots,x_q,(\haarisometry\ket{x_1})^{\otimes t},\dots,(\haarisometry\ket{x_q})^{\otimes t}) = 1]
\Bigg|.
\end{multline*}
By~\Cref{lem:haar_perp_to_iid} and viewing $\haarisometry\ket{x_i}$ as applying an $(n+m)$-qubit Haar unitary $U$ on $\ket{x_i}\ket{0^m}$,
$$\TD\left(
\Ex_{\ket{\vartheta_1},\dots,\ket{\vartheta_q}\gets\Haar_{n+m}} \left[\bigotimes_{i=1}^q \ketbra{\vartheta_i}{\vartheta_i}^{\otimes t} \right],
\Ex_{\haarisometry\gets\overline{\Haar_{n,n+m}}} \left[\bigotimes_{i=1}^q (\haarisometry\ketbra{x_i}{x_i}\haarisometry^\dagger)^{\otimes t} \right] \right)
= O(q^2t/2^{n+m}).$$
So the advantage of $\distinguisher$ is at least $\nu(\secp) - O(q^2t/2^{n+m})$, which is non-negligible. But it contradicts the security of the underlying PRI.
\end{proof}

\ifllncs
\else
\input{MultiEnc}
\input{SuccinctQSC}
\fi

\subsection{Quantum Message Authentication Codes}
\label{sec:qmacs}
The scheme of authenticating \emph{quantum messages} was first studied by Barnum et~al. \cite{BCG+02FOCS} in which they considered \emph{one-time private-key} authentication schemes. The definition in \cite{BCG+02FOCS} is generalized in the following works~\cite{DNS12,GYZ17}. In particular, Garg, Yuen, and Zhandry \cite{GYZ17} defined the notion of \emph{total authentication}, which is tailored for one-time (information-theoretic) security. They showed that total authentication implies unforgeability (in certain settings\footnote{In more detail, they show total authentication implies unforgeability for MACs for classical messages with security against a single superposition message query.}) and \emph{key reusability} --- conditioned on successful verification of an authentication scheme that satisfies total authentication, the key can be reused by the honest parties. Moreover, they constructed a total-authenticating scheme from unitary $8$-designs. Later, the works of \cite{Portmann17,AM17} independently improved the construction by using only unitary $2$-designs to achieve total authentication. 

In the fully classical setting, many-time security of an authentication scheme is defined via \emph{unforgeability} --- no efficient adversary can forge an un-queried message-tag pair. A message authentication code (MAC) is a common primitive that satisfies the desired properties. However, consider MACs for classical messages: when the adversary is allowed to query the signing oracle in superposition~\cite{BZ13,AMRS20}, defining the \emph{freshness} of the forgery is already nontrivial. For quantum message authentication schemes, it is well-known that authentication implies encryption~\cite{BCG+02FOCS}. 
Furthermore, due to the quantum nature of no-cloning and entanglement, it is challenging to define a general many-time security notion~\cite{AGM18,AGM21quantum}. Nevertheless, we consider a strict version of MACs for quantum messages in this subsection. We'll focus on several weak yet nontrivial notions of unforgeability and show how to achieve them using PRIs.

\paragraph{\bf{Syntax.}}
A message authentication codes (MAC) scheme for quantum messages of length $n(\secp)$ is a triple of algorithms $(\setup,\sign,\ver)$. 
\begin{itemize}
    \item $\setup(1^{\secparam})$: on input the security parameter $\secp$, output a key $k\gets\bit^\secp$.
    \item $\sign(k,\ket{\psi})$: on input $k\in\bit^\secp$ and a quantum message $\ket{\psi}\in \pqstates{n}$, output a quantum tag\footnote{We emphasize that here we explicitly require the tag to be a pure state. We can relax this condition to allow for the signature algorithm to output a state that is close to a pure state without changing the notion much. } $\ket{\phi}\in\pqstates{s}$ where $s(\secp) = \poly(\secp)$ is the tag length.
    \item $\ver(k,\ket{\phi})$: on input $k\in\bit^\secp$ and a quantum tag $\ket{\phi}\in\pqstates{s}$, output a mixed quantum state $\rho\in\cD(\C^{2^n})$.
\end{itemize}

\begin{definition}[Correctness]
There exists a negligible function $\veps(\cdot)$ such that for every $\secp\in\N$, $k\in\bit^\secp$, and quantum message $\ket{\psi}\in \pqstates{n}$,
\begin{align*}
    \TD(\ver(k,\sign(k,\ket{\psi})), \ketbra{\psi}{\psi} ) \le \veps(\secp).
\end{align*}

\end{definition}

\paragraph{Security Definitions.} Defining security for MACs for quantum states is quite challenging, as discussed in prior works, notably in~\cite{AGM18}. Nonetheless, our goal is to present some reasonable, although restrictive, definitions of MACs for quantum states whose feasibility can be established based on the existence of pseudorandom isometries. We believe that our results shed light on the interesting connection between pseudorandom isometries and MACs for quantum states and we leave the exploration of presenting the most general definition of MACs for quantum states (which in our eyes is an interesting research direction by itself!) for future works. 

\par When the adversary is only asked to output a single copy of the (quantum) forgery, it is unclear how to achieve negligible security error. For example, if the verification is done by simply applying a SWAP test\footnote{The SWAP test is an efficient quantum circuit that takes as input two density matrices $\rho,\sigma$ of the same dimension and output $1$ with probability $\frac{1+\Tr(\rho\sigma)}{2}$.}, then the success probability of the forger is at least $1/2$. In the following, we introduce several notions capturing unforgeability. First, in order to boost security, a straightforward way is to simply ask the adversary to send $t = \poly(\secp)$ copies of the forgery message and tag. 

\begin{definition}[Many-Copies-Unforgeability] 
Let $t = \poly(\secp)$. For every polynomial $q(\cdot)$ and every non-uniform QPT adversary, there exists a function $\veps(\cdot)$ such that for sufficiently large $\secp\in\N$, the adversary wins with probability at most $\veps(\secp)$ in the following security game:
\begin{enumerate}
    \item Challenger samples $k\gets \bit^\secp$.
    \item The adversary sends $\ket{\psi_1},\dots,\ket{\psi_q}\in \pqstates{n}$ and receives $\sign(k,\ket{\psi_i})$ for $i = 1,\dots,q$.
    \item The adversary outputs $(\ket{\psi^*}\otimes\ket{\phi^*})^{\otimes t}$ where $\ket{\psi^*}\in \pqstates{n}$ is orthogonal to $\ket{\psi_i}$ for $i = 1,\dots,q$.
    \item Challenger runs $\swaptest(\ketbra{\psi^*}{\psi^*},\ver(k,\ket{\phi^*}))$ $t$ times in parallel. The adversary wins if and only if every SWAP test outputs $1$.
\end{enumerate}
\end{definition}

\begin{remark}
We note that, in general, the forgery message and the tag could be entangled. Here, we focus on a restricted case in which the message and tag are required to be a product state. We leave the exploration of stronger security notions for future works.
\end{remark}

In some cases, it is unsatisfactory to ask the adversary to output multiple copies of the forgery tag due to the no-cloning theorem and in this case, we can consider the following definition in which the adversary needs to output multiple copies of the forgery message but only a single copy of the forgery tag. The winning condition of the adversary is defined by passing the generalized SWAP test --- called the \emph{permutation test}~\cite{BBA+97,KNY08,GGH+15,BS20}. 

\begin{lemma}[Permutation Test]
The permutation test is an efficient quantum circuit $\permtest$ that takes as input $\rho\in\cD((\C^d)^{\otimes t})$, outputs $1$ with probability $p:=\Tr(\Pi_\sym^{d,t}\rho)$, and outputs $0$ with probability $1-p$.
\end{lemma}

\begin{definition}[$(\permtest,t,\veps)$-unforgeability] 
For every polynomial $q(\cdot)$ and every non-uniform QPT adversary, there exists a function $\veps(\cdot)$ such that for sufficiently large $\secp\in\N$, the adversary wins with probability at most $\veps(\secp)$ in the following security game:
\begin{enumerate}
    \item Challenger samples $k\gets \bit^\secp$.
    \item The adversary sends $\ket{\psi_1},\dots,\ket{\psi_q}\in\pqstates{n}$ and receives $\sign(k,\ket{\psi_i})$ for $i = 1,\dots,q$.
    \item The adversary outputs $\ket{\psi^*}^{\otimes t}\otimes\ket{\phi^*}$ where $\ket{\psi^*}\in \pqstates{n}$ and is orthogonal to $\ket{\psi_i}$ for $i = 1,\dots,q$.
    \item The adversary wins if $\permtest(\ketbra{\psi^*}{\psi^*}^{\otimes t} \otimes \ver(k,\ket{\phi^*})) = 1$.
\end{enumerate}
\end{definition}

\noindent Finally, suppose $\sign(k,\cdot)$ is an isometry for every $k\in\bit^\secp$. We consider another definition in which we ask the adversary to send the classical description of the quantum circuit that generates the forgery message and only one copy of the corresponding tag.
\begin{definition}[Uncompute-Unforgeability]
For every polynomial $q(\cdot)$ and every non-uniform QPT adversary, there exists a negligible function $\veps(\cdot)$ such that for every $\secp\in\N$, the adversary wins with probability at most $\veps(\secp)$ in the following security game:
\begin{enumerate}
    \item Challenger samples $k\gets \bit^\secp$.
    \item The adversary sends $\ket{\psi_1},\dots,\ket{\psi_q}\in\pqstates{n}$ and receives $\sign(k,\ket{\psi_i})$ for $i = 1,\dots,q$.
    \item The adversary outputs a pair $(C,\ket{\phi^*})$ where $C$ is the classical description of a quantum circuit containing no measurements such that $C\ket{0^n}$ is orthogonal to $\ket{\psi_i}$ for $i = 1,\dots,q$.
    \item Challenger applies $C^{\dagger}\ver(k,\cdot)$ on $\ket{\phi^*}$ and performs a measurement on all qubits in the computational basis. The adversary wins if and only if the measurement outcome is $0^n$.
\end{enumerate}
\end{definition}

\noindent Let $\pri = \set{F_\secp}_{\secp\in\N}$ be a strong invertible adaptive $(n,n+m)$-PRI (\Cref{def:strong_inv_PRI}) where $n(\cdot),m(\cdot)$ are polynomials. We construct a MAC for quantum messages from $\pri$.
\begin{construction}[MAC for quantum messages] \hfill
\label{construction:tt-QMAC}
\begin{enumerate}
    \item $\sign(k,\ket{\psi}):$ on input $k\in\bit^\secp$ and a message $\ket{\psi}\in \pqstates{n}$, output $F_\secp(k,\ket{\psi}) \in \pqstates{m+n}$.
    \item $\ver(k,\ket{\phi}):$ on input $k\in\bit^\secp$ and a tag $\ket{\phi} \in \pqstates{m+n}$, output $\inv(k,\ket{\phi})$.
\end{enumerate}
\end{construction}

\noindent The correctness of~\Cref{construction:tt-QMAC} follows from the invertibility of $\pri$.

\begin{lemma}[Operator Norm after Partial Trace, Eq.(23) in~\cite{Ras12}]
\label{lem:operator_norm}
Let $H_A,H_B$ be finite-dimensional Hilbert spaces and $Q \in \cL(H_A\otimes H_B)$. Then $\norm{\Tr_B(Q)}_\infty \leq \dim(H_B) \cdot\norm{Q}_\infty$.
\end{lemma}

\begin{lemma} \label{lem:tt-QMAC:security}
\Cref{construction:tt-QMAC} satisfies many-copies-unforgeability.
\end{lemma}
\ifllncs
    The proof relies on an explicit geometric construction of Haar isometries (unitaries)~\cite[page. 19]{Meckes19}. The details can be found in~\Cref{app:proof_qmac}.
\else
    \input{application_proofs/tt-mac}
\fi

\begin{theorem} \label{thm:t_1_1/t}
For every $t\in\N$, \Cref{construction:tt-QMAC} satisfies $(\permtest,t,O(1/t))$-unforgeability. 
\end{theorem}
\ifllncs
    The proof of~\Cref{thm:t_1_1/t} can be found in~\Cref{app:proof_qmac}.
\else
\input{application_proofs/t1-mac}
\fi

\begin{theorem} \label{thm:uncom_unforge}
\Cref{construction:tt-QMAC} satisfies uncompute-unforgeability. 
\end{theorem}
\ifllncs
    The proof of~\Cref{thm:uncom_unforge} can be found in~\Cref{app:proof_qmac}.
\else
    \input{application_proofs/ucf-mac}
\fi

\subsection{Length Extension of Pseudorandom States}

\label{sec:prs_extension}
We introduce methods to increase the \emph{length} of pseudorandom quantum states while preserving the \emph{number of copies}. In the classical setting, the length extension of pseudorandom strings can be accomplished by repeatedly applying PRGs. On the other hand, since pseudorandom random states are necessarily (highly) pure and entangled~\cite{JLS18,AQY21}, no such method was known that would not decrease the number of copies.

\begin{theorem}[Length Extension Theorem]
Assuming $\qclass_{{\sf Haar}}$-secure pseudorandom isometry, mapping $n$ qubits to $n+m$ qubits, and an $n$-qubit $\mathrm{PRSG}$, there exists an $(n+m)$-$\mathrm{PRSG}$. Similarly, assuming $\qclass_{{\sf Haar}}$-secure pseudorandom isometry, mapping $n$ qubits to $n+m$ qubits, and classical-accessible selectively-secure $(\ell,n)$-$\mathrm{PRFSG}$, there exists an classical-accessible selectively-secure $(\ell,n+m)$-$\mathrm{PRFSG}$.
\end{theorem}
\begin{proof}
The constructions are straightforward. We first construct an $(n+m)$-$\mathrm{PRSG}$ as follows: Let $G$ be an $n$-qubit $\mathrm{PRSG}$ and $\pri$ be a $\qclass_{{\sf Haar}}$-secure $(n,n+m)$-pseudorandom isometry. On input $k = (k_1,k_2)$ where $k_1,k_2\in\bit^\secp$, output $\pri(k_2,G(k_1))$.  Let $t$ be an arbitrary polynomial. Consider the following hybrids:
\begin{itemize}
    \item \textbf{Hybrid~1:} $k_1,k_2\gets\bit^\secp$, output $\pri(k_2,G(k_1))^{\otimes t}$
    \item \textbf{Hybrid~2:} $\ket{\theta}\gets\Haar_n$, $k_2 \gets\bit^\secp$, output $\pri(k_2,\ket{\theta})^{\otimes t}$
    \item \textbf{Hybrid~3:} $\ket{\gamma}\gets\Haar_{n+m}$, output $\ket{\gamma}^{\otimes t}$
\end{itemize}
Hybrids~1 and 2 are computationally indistinguishable from the security of $\mathrm{PRSG}$. Hybrids~2 and 3 are computationally indistinguishable from the security of $\mathrm{PRI}$. We then construct an $(\ell,n+m)$-qubit $\mathrm{PRFSG}$ as follows: Let $F$ be an $(\ell,n)$-qubit $\mathrm{PRFSG}$. On input $k = (k_1,k_2)$ where $k_1,k_2\in\bit^\secp$ and $x\in\bit^\ell$, run $F(k_1,x) = \ket{x}\ket{\theta_x}$ and output $\ket{x}\otimes\pri(k_2,\ket{\theta_x})$. The security follows similarly.
\end{proof}

Next, we introduce another length extension approach that offers an incomparable trade-off compared to the first one. Consider the following scenario: given $t(\secp) = o(\secp)$ copies\footnote{Due to technical issues, we are only able to prove the theorem when $t$ is sublinear in $\secp$.} of a $2n$-qubit Haar state, what is the minimum required randomness in order to generate $t$ copies of a $(2n+m)$-qubit pseudorandom state (where $n(\secp),m(\secp)$ are polynomials)? First, we can ignore the original Haar state and output a truly random state from scratch by employing $t$-designs at the cost of $\poly(t,n+m) = \poly(\secp)$ bits of randomness. Suppose we assume the existence of $(n,n+m)$-PRIs. Trivially, applying the $t$-fold PRI on a fixed initial state can generate $(n+m)$-bit pseudorandom state at the cost of $\secp$ bits of randomness (which serve as the key of the PRI). In the following, we show that the output obtained by applying the $t$-fold PRI on the last $n$ qubits of every Haar state is computationally indistinguishable from $t$-copies of a $(2n+m)$-qubit Haar state.

\begin{theorem}[Another Length Extension Theorem]
\label{thm:length_extension}
Let $\set{F_\secp}_{\secp\in\N}$ be an $(n,n+m)$-$\mathrm{PRI}$, $t=t(\secp)$,
\[
\rho := 
\Ex_{\ket{\theta}\gets\Haar_{2n},k\in\bit^\secp} \left[ (I_n \otimes F_k)^{\otimes t} \ketbra{\theta}{\theta}^{\otimes t} (I_n \otimes F_k^\dagger)^{\otimes t} \right],
\]
where $F_k$ means $F_\secp(k,\cdot)$ and $I_n$ is the identity operator on $n$ qubits, and 
\[
\sigma := \Ex_{\ket{\gamma}\gets\Haar_{2n+m}} \left[ \ketbra{\gamma}{\gamma}^{\otimes t} \right].
\]
Then any non-uniform QPT adversary has at most $O(t!t^2/2^{n+m} + t^2/2^n)$ advantage in distinguishing $\rho$ from $\sigma$.
\end{theorem}
\ifllncs
    The proof of~\Cref{thm:length_extension} can be found in~\Cref{app:proof_extension}.
\else
    \input{application_proofs/length_ext}
\fi

%% file: MultiEnc.tex
\ifllncs
    \section{Multi-Copy Security of Private-Key and Public-Key Encryption Schemes}
\else
    \subsection{Multi-Copy Security of Encryption Schemes}
\fi
\label{sec:multicopy_enc}

\noindent It is well known that quantum states can be generically encrypted using the hybrid encryption technique. However, there is a stronger property referred to as {\em multi-copy} security that states the following: the indistinguishability should still hold even when given multiple copies of the ciphertext. In~\cite{LQSYZ23}, the authors considered multi-copy security only for one-time encryption schemes. We further consider private-key and public-key settings and formalize them below. 

\begin{definition}[Multi-Copy Security of Public-Key Encryption]
We say that $(\setup,\enc,\dec)$ is a public-key encryption scheme for quantum states if it satisfies the following security property: for any two states $\ket{\psi_0},\ket{\psi_1} \in \pqstates{n}$, where $n=n(\secparam)$ is a polynomial, for any non-uniform QPT distinguisher $\distinguisher$, for any polynomial $t=t(\secparam)$, 
\begin{multline*}
\Bigg| \Pr\left[ \distinguisher(\pk, \rho^{\otimes t}) = 1: \substack{(\pk,\sk) \gets \setup(1^{\secparam})\\ \ \\ \rho \gets \enc(\pk,\ket{\psi_0})}  \right] - \Pr\left[ \distinguisher(\pk, \rho^{\otimes t}) = 1: \substack{(\pk,\sk) \gets \setup(1^{\secparam})\\ \ \\ \rho \gets \enc(\pk,\ket{\psi_1})}  \right] \Bigg| 
\leq \veps(\secparam),
\end{multline*}
for some negligible function $\veps(\cdot)$. 
\end{definition}

\begin{definition}[Multi-Copy Security of Private-Key Encryption]
We say that $(\setup,\enc,\dec)$ is a private-key encryption scheme for quantum states if it satisfies the following security property: for any $q=\poly(\secparam)$, for any tuples of states  $\ket{\psi_1^{(0)}},\ldots,\ket{\psi_{q(\secparam)}^{(0)}} \in \pqstates{n}$ and $\ket{\psi_1^{(1)}},\ldots\ket{\psi_{q(\secparam)}^{(1)}}  \in \pqstates{n}$, where $n=n(\secparam)$ is a polynomial, for any non-uniform QPT distinguisher $\distinguisher$, for any polynomial $t=t(\secparam)$,
\begin{multline*}
\Bigg| \Pr\left[ \distinguisher\left(1^{\secparam}, \bigotimes_{i=1}^q \rho_i^{\otimes t} \right) = 1: \substack{\sk \leftarrow \setup(1^{\secparam})\\ \ \\ \forall i \in [q], \\ \rho_i \gets \enc(\sk,\ket{\psi_i^{(0)}})}  \right] - \Pr\left[ \distinguisher\left(1^{\secparam}, \bigotimes_{i=1}^q \rho_i^{\otimes t} \right) = 1: \substack{\sk \gets \setup(1^{\secparam})\\ \ \\ \forall i \in [q], \\ \rho_i \leftarrow \enc(\sk,\ket{\psi_i^{(1)}})}  \right] \Bigg| 
\leq \veps(\secparam),
\end{multline*}
for some negligible function $\veps(\cdot)$.
\end{definition}

\begin{remark}
We can similarly define multi-copy security in the adaptive setting where the adversary can request for $(i+1)^{th}$ encryption after obtaining encryptions on $i$ messages. We can further generalize the above definition to consider encryption for mixed states instead of just pure states. We leave the exploration of both these generalizations to future works. 
\end{remark}

\paragraph{Construction.} We discuss the construction of the multi-copy secure public-key encryption scheme $(\setup,\enc,\dec)$; the construction and security of multi-copy private-key encryption can be similarly derived. We will start with a post-quantum public-key encryption scheme $(\smsetup,\smenc,\smdec)$. We will also use an invertible $\qclass_{{\sf Single}}$-secure pseudorandom isometry $\pri=\left\{F_{\secparam}\right\}_{\secparam \in \mathbb{N}}$ (\Cref{sec:security_multi_same}), where $\inv$ is the inversion function. 
\begin{itemize}
    \item $\setup(1^{\secparam})$: on input the security parameter $\secparam$, compute $(pk,sk) \leftarrow \smsetup(1^{\secparam})$. Output $pk$ as the public key $\pk$ and output $sk$ as the secret key $\sk$. 
    
    \item $\enc_\secp(\pk,\sigma)$: on input a public key $\pk=pk$, state $\sigma$, first sample a PRI key $k \xleftarrow{\$} \{0,1\}^{\secparam}$ and then compute $ct \leftarrow \smenc(pk,k)$. Also, compute $\rho \leftarrow F_{\secparam}(k,\sigma)$. Output the ciphertext state $\ct=\left(ct,\rho\right)$.   
    
    \item $\dec_\secp(\sk,\ct)$: on input the decryption key $\sk=sk$, ciphertext state $\ct=\left(ct,\rho\right)$, first compute $k \leftarrow \smdec(sk,ct)$. Compute $\inv(k,\rho)$ to obtain $\sigma$. Output $\sigma$.  
    
\end{itemize}

\paragraph{Correctness.} Follows from the correctness of the post-quantum encryption scheme $(\smsetup,\smenc,\smdec)$ and from the guarantees of the inversion algorithm. 

\paragraph{Multi-Copy Security.} The multi-copy security follows from the following hybrid argument. Let the challenge messages be $(\ket{\psi_0},\ket{\psi_1}) \in \pqstates{n} \otimes \pqstates{n}$. Let $t(\secparam)$ be a polynomial in $\secparam$. \\

\noindent \textbf{Hybrid~1.} Output $\left(\enc(\pk,\ket{\psi_0})\right)^{\otimes t(\secparam)}$. \\

\noindent \textbf{Hybrid~2.} Output $\left(\hybrid.\enc(\pk,\ket{\psi_0})\right)^{\otimes t(\secparam)}$, where $\hybrid.\enc$ performs just like $\enc$ except that it computes $\smenc(pk,0)$ instead of $\smenc(pk,k)$.
\par The computational indistinguishability of Hybrid~1 and Hybrid~2 follows from the security of the post-quantum public-key encryption scheme. \\

\noindent \textbf{Hybrid~3.} Output $\left(\hybrid.\enc(\pk,\ket{\psi_1})\right)^{\otimes t(\secparam)}$. 
\par The computational indistinguishability of Hybrid~2 and Hybrid~3 follows from the $\qclass_{\mathsf{pure}_{t}}$-security of PRI. More specifically, we can consider an intermediate hybrid, where we switch the output of PRI on $\ket{\psi_0}$ to the output of a Haar isometry on $\ket{\psi_0}$. Note that this is identical to the output of a Haar isometry on $\ket{\psi_1}$. Finally, invoking the security of PRI, we can switch this to the output of PRI on $\ket{\psi_1}$. \\

\noindent \textbf{Hybrid~4.} Output $\left(\enc(\pk,\ket{\psi_1})\right)^{\otimes t(\secparam)}$. 
\par The computational indistinguishability of Hybrid~3 and Hybrid~4 follows from the security of the post-quantum public-key encryption scheme. 

\begin{remark}
In the above scheme, if we instantiate $(\smsetup,\smenc,\smdec)$ using a post-quantum secure private-key encryption scheme then we obtain a multi-copy secure private-key encryption for quantum states scheme.   
\end{remark}

%% file: SuccinctQSC.tex
\ifllncs
    \section{Succinct Quantum Commitments}
\else
    \subsection{Succinct Quantum Commitments}
\fi
\label{sec:commitments}
This subsection closely follows~\cite[Appendix~C]{GJMZ23} in which they showed a generic transformation from $t$-time secure $d$-dimensional PRUs to one-time secure symmetric encryption schemes for $\binom{d+t-1}{t}$-dimensional quantum messages. The main approach of~\cite{GJMZ23} relies on the \emph{Schur transform}~\cite{HarrowThesis}. In short, the Schur transform is a basis transform between the computational basis and the Schur basis. We observe that PRIs are already sufficient for such a transformation. Recall that a Haar random isometry is distributed identically to first appending $\ket{0^m}$ followed by applying a Haar random unitary. Hence, our construction needs to perform a Schur transform and an inverse Schur transform with \emph{different} dimensions. An immediate corollary is that PRIs imply succinct quantum state commitment (QSC) schemes. This follows from Theorem~5.3 in~\cite{GJMZ23} which states that one-time secure quantum encryption schemes imply succinct QSC schemes. 

\begin{construction}[One-time quantum encryption scheme from PRIs]
\label{construction:OT-PRI}
Let $\pri = \set{F_\secp}_{\secp\in\N}$ be a secure $(n,n+m)$-PRI family and $t(\secp) = \poly(\secp)$. We construct a one-time quantum encryption scheme $\set{\expand(F_\secp,t)}_{\secp\in\N}$ as follows. On input a $d$-dimensional quantum message $\ket{\psi}$, where $d = d(n,m,t) := \binom{2^n+t-1}{t}/2^{mt}$, do the following:
\begin{itemize}
    \item Initialize the state $\ket{\Psi} := \ket{\Lambda = 0} \ket{p_{\Lambda} = 0} \ket{\psi}$.\footnote{We follow the notation in~\cite{GJMZ23}.}
    \item Apply $U_{\mathsf{Sch},d'} (F_\secp(k,\cdot))^{\otimes t} U_{\mathsf{Sch},d}^\dagger$ on $\ket{\Psi}$, where $d' := \binom{2^n+t-1}{t}$ is the dimension of the (quantum) ciphertext.
    \item Trace out the first two registers and output the last register as the ciphertext.
\end{itemize}
\end{construction}

\begin{theorem}[PRI Expansion] \label{thm:PRI_expansion}
If $\set{F_\secp}_{\secp\in\N}$ is an $(n(\secp),m(\secp))$-PRI family, then \Cref{construction:OT-PRI} is a secure quantum one-time encryption scheme with message of dimension $\binom{2^n+t-1}{t}/2^{mt}$.
\end{theorem}
\begin{proof}[Proof sketch.]
By security of PRI, we can replace $F_\secp$ with a Haar random isometry for the rest of the proof. Since the subspace labeled by $\Lambda = 0$ corresponds to the symmetric subspace $\vee^t\C^d$, it holds that $U_{\mathsf{Sch},d}^\dagger\ket{\Psi} \in \vee^t\C^d$. By~\Cref{lem:multi_span_sym}, there exists some finite set $\cS$ of vectors in $\C^d$ such that $U_{\mathsf{Sch},d}^\dagger\ket{\Psi}$ can be written as a linear combination of $\ket{v}^{\otimes t}$ with $\ket{v}\in\cS$. After appending $\ket{0^m}^{\otimes t}$ to $U_{\mathsf{Sch},d}^\dagger\ket{\Psi}$, the state is now a linear combination of $(\ket{v}\otimes\ket{0^m})^{\otimes t}$, which implies that $U_{\mathsf{Sch},d}^\dagger\ket{\Psi}\ket{0^m}^{\otimes t} \in \vee^t\C^{d'}$. Let $\rho := U_{\mathsf{Sch},d}^\dagger\ket{\Psi}\ket{0^m}^{\otimes t} \bra{\Psi}U_{\mathsf{Sch},d} \bra{0^m}^{\otimes t}$. Then applying $t$-fold Haar unitary on $\rho$ results in the fully mixed state of $\vee^t\C^{d'}$ by Schur's lemma. Finally, applying the second Schur transform followed by tracing out the first two registers generates a $d'$-dimensional fully mixed state.
\end{proof}

%% file: application_proofs/tt-mac.tex
\begin{proof} 
By the security of $\pri$, we replace it with a Haar isometry $\haarisometry$ in the construction. Fix $\lambda$ and queries $\ket{\psi_1},\dots,\ket{\psi_q}$. Let $V_{in} := \mathsf{span}\set{\ket{\psi_1}\otimes\ket{0^m}_\aux,\dots,\ket{\psi_q}\otimes\ket{0^m}_\aux} \subseteq \C^{2^{n+m}}$ and $d:=\dim(V_{in})\le q$. Choose an arbitrary orthonormal basis of $V_{in}$ denoted by $\set{\ket{e_1},\dots,\ket{e_d}}$. 

Given $\ket{e_1},\dots,\ket{e_d}$, the Haar random unitary can be viewed as being partially defined by sampling $\ket{v_1},\dots,\ket{v_d}$ according to the procedures in~\Cref{fact:sampling_Haar_unitary}. Let $V_{out} := \mathsf{span}\set{ \ket{v_1},\dots,\ket{v_d} } \subseteq \C^{2^{n+m}}$. Note that all quantum tags $\ket{\phi_i}$ are defined since each of them is in $V_{out}$. Now, fix the forgery $(\ket{\psi^*},\ket{\phi^*})$. Suppose $\ket{\phi^*} = \ket{v_{out}} + \ket{v^\perp_{out}}$ where $\ket{v_{out}}\in V_{out}$ and $\ket{v^\perp_{out}}\in V_{out}^\perp$ are sub-normalized states. Using~\Cref{fact:inverse_of_Haar}, we consider the average fidelity between $\haarisometry^{-1}(\ket{\phi^*})$ and $\ket{\psi^*}$:
\begin{align} \label{eq:average_fidelity}
\Ex_{\haarisometry\mid_V} \left[ \expbraket{\psi^*}{\haarisometry^{-1}(\ket{\phi^*})}{\psi^*}  \right]
= \expbraket{\psi^*}{\Ex_{U\gets\overline{\Haar_{n+m}}\mid_V}\left[\Tr_\aux(U^\dagger\ketbra{\phi^*}{\phi^*}U)\right]}{\psi^*},
\end{align}
where $\haarisometry\mid_V$ means $\haarisometry$ is a Haar isometry conditioned on $\haarisometry\ket{e_i} = \ket{v_i}$ for $i=1,2,\dots,d$; $\overline{\Haar_{n+m}}\mid_V$ is defined similarly. Expanding $\ket{\phi^*}$, this yields
\begin{align*}
& \Cref{eq:average_fidelity} 
= \expbraket{\psi^*}{\Tr_\aux \left( \Ex_{U\gets\overline{\Haar_{n+m}}\mid_V}\left[U^\dagger\ketbra{\phi^*}{\phi^*}U\right] \right)}{\psi^*} \\
& = \expbraket{\psi^*}{\Tr_\aux \left( \Ex_{U\gets\overline{\Haar_{n+m}}\mid_V}\left[U^\dagger (\ketbra{v^\perp_{out}}{v^\perp_{out}} + \ketbra{v^\perp_{out}}{v_{out}} + \ketbra{v_{out}}{v^\perp_{out}} + \ketbra{v_{out}}{v_{out}}) U \right] \right)}{\psi^*}.
\end{align*}
First note that 
\[
\Ex_{U\gets\overline{\Haar_{n+m}}\mid_V}\left[U^\dagger\ketbra{v^\perp_{out}}{v^\perp_{out}}U\right] 
= \norm{\ket{v^\perp_{out}}}^2 \cdot \Ex_{\ket{\vartheta}\gets\Haar(V_{out}^\perp)}\left[ \ketbra{\vartheta}{\vartheta} \right] 
= \frac{\norm{\ket{v^\perp_{out}}}^2 \cdot I_{V^\perp_{out}}}{\dim(V^\perp_{out})}
\]
from~\Cref{fact:avg-haar-random}. Next, 
\[
\Ex_{U\gets\overline{\Haar_{n+m}}\mid_V}\left[U^\dagger\ketbra{v_{out}}{v^\perp_{out}}U\right] = \ket{v_{in}}\Ex_{U\gets\overline{\Haar_{n+m}}\mid_V}\left[\bra{v^\perp_{out}}U\right] 
= \ket{v_{in}}\Ex_{\ket{\vartheta}\gets\Haar(V_{out}^\perp)}\left[\bra{\vartheta} \right] 
= 0,
\]
since the average of a uniformly random vector on a sphere is $0$, where $\ket{v_{in}}\in V_{in}$ is the state such that $U\ket{v_{in}} = \ket{v_{out}}$ for every $U$ sampled from $\overline{\Haar}_{n+m}\mid_V$. Similarly, we have
\[
\Ex_{U\gets\overline{\Haar_{n+m}}\mid_V}\left[U^\dagger\ketbra{v^\perp_{out}}{v_{out}}U\right] = 0.
\]
Moreover, $\Ex_{U\gets\overline{\Haar_{n+m}}\mid_V}\left[U^\dagger\ketbra{v_{out}}{v_{out}}U\right] = \ketbra{v_{in}}{v_{in}}$ is supported by $V_{in}$. We then obtain
\begin{align*}
\Cref{eq:average_fidelity} 
= \expbraket{\psi^*}{\Tr_\aux \left( \frac{\norm{\ket{v^\perp_{out}}}^2 \cdot I_{V^\perp_{out}}}{\dim(V^\perp_{out})} + \ketbra{v_{in}}{v_{in}} \right)}{\psi^*}.
\end{align*}
However, the last $m$ qubits of $\ket{v_{in}}$ on register $\aux$ must be $\ket{0^m}_\aux$ by the definition of $V_{in}$. So after partially tracing out $\aux$, the reduced (sub-normalized) density matrix $\Tr_\aux \left(\ketbra{v_{in}}{v_{in}} \right)$ is supported by $\mathsf{span}\set{\ket{\psi_1},\dots,\ket{\psi_q}}$. But recall that the forgery message $\ket{\psi^*}$ must be orthogonal to the previous queries $\ket{\psi_1},\dots,\ket{\psi_q}$, thus $\expbraket{\psi^*}{\Tr_\aux \left( \ketbra{v_{in}}{v_{in}} \right)}{\psi^*} = 0$. Finally, the average fidelity can be simplified and bounded as follows:
\begin{align*}
\Cref{eq:average_fidelity}
& = \frac{\norm{\ket{v^\perp_{out}}}^2}{\dim(V^\perp_{out})} \cdot \expbraket{\psi^*}{\Tr_\aux \left( I_{V^\perp_{out}} \right) }{\psi^*} \\
& \leq \frac{\norm{\ket{v^\perp_{out}}}^2}{\dim(V^\perp_{out})} \cdot \norm{ \Tr_\aux \left( I_{V^\perp_{out}} \right)  }_\infty  \\
& \leq \frac{\norm{\ket{v^\perp_{out}}}^2}{\dim(V^\perp_{out})} \cdot  \dim(H_\aux) \cdot \norm{ I_{V^\perp_{out}} }_\infty \\
& \leq \frac{1}{2^{m+n}-q} \cdot 2^m
= \negl(\secp),
\end{align*}
where the first inequality follows from the definition of operator norm and the second inequality follows from~\Cref{lem:operator_norm}. By Markov's inequality, we have
\begin{align*}
\Pr_{\haarisometry\mid_V} \left[ \expbraket{\psi^*}{\haarisometry^{-1}(\ket{\phi^*})}{\psi^*} \ge 0.1 \right] = \negl(\secp).
\end{align*}
Hence, the probability of all the $t$ swap tests outputting $1$ satisfies
\begin{align*}
& \Ex_{\haarisometry\mid_V} \left[ \Pr \left[ \swaptest_i( \ket{\psi^*}, \haarisometry^{-1} (\ket{\phi^*}) ) = 1,\ i = 1,\dots,t \right] \right] \\
= & \Ex_{\haarisometry\mid_V} \left[ \left( \frac{1}{2} + \frac{1}{2} \expbraket{\psi^*}{\haarisometry^{-1}(\ket{\phi^*})}{\psi^*} \right)^t \right] \\
\leq & \left( \frac{1}{2} + 0.1 \right)^t + \negl(\secp) 
= 2^{-\Omega(t)} + \negl(\secp)
\end{align*}
from Hoeffding bounds.
This finishes the proof of~\Cref{lem:tt-QMAC:security}.
\end{proof}

%% file: application_proofs/t1-mac.tex
\begin{proof}\ifllncs[Proof of~\Cref{thm:t_1_1/t}]\else\fi 
By the security of $\pri$, we replace it with a Haar isometry $\haarisometry$ in the construction. Fix $\secp$ and queries $\ket{\psi_1},\dots,\ket{\psi_q}$. Similar to~\Cref{lem:tt-QMAC:security}, we can view the Haar unitary to be partially sampled. Let $V_{in},V_{out}$ be defined as in~\Cref{lem:tt-QMAC:security}. The winning probability of the forger is 
\begin{align} \label{eq:winning_prob_perm}
& \Ex_{\haarisometry\mid_V} \left[ \Pr\left[ \permtest\left( \ketbra{\psi^*}{\psi^*}^{\otimes t} \otimes \haarisometry^{-1}(\ket{\phi^*}) \right) = 1 \right] \right] \nn \\
= & \Ex_{U\gets\overline{\Haar}_{n+m}\mid_V}\left[ \Tr\left( \Pi_\sym^{2^n,t+1} \left( \ketbra{\psi^*}{\psi^*}^{\otimes t} \otimes \Tr_\aux(U^\dagger\ketbra{\phi^*}{\phi^*}U) \right) \right) \right] \nn \\
= & \Tr\left( \frac{\sum_{\sigma \in S_{t+1}} P_\sigma}{(t+1)!} \left( \ketbra{\psi^*}{\psi^*}^{\otimes t} \otimes \Ex_{U\gets\overline{\Haar}_{n+m}\mid_V}\left[ \Tr_\aux \left(U^\dagger\ketbra{\phi^*}{\phi^*}U \right) \right] \right) \right).
\end{align}
Consider the two cases classified by whether $t+1$ is a fixed point of $\sigma$: first, if $\sigma(t+1) = t+1$, then
\begin{align*}
& \Tr\left( P_\sigma \left( \ketbra{\psi^*}{\psi^*}^{\otimes t} \otimes \Ex_{U\gets\overline{\Haar}_{n+m}\mid_V}\left[ \Tr_\aux \left(U^\dagger\ketbra{\phi^*}{\phi^*}U \right) \right] \right) \right) \\
= & \Tr\left( \Ex_{U\gets\overline{\Haar}_{n+m}\mid_V}\left[ \Tr_\aux \left(U^\dagger\ketbra{\phi^*}{\phi^*}U \right) \right] \right) = 1.
\end{align*}
Otherwise, we can decompose $\ket{\phi^*} = \ket{v_{out}} + \ket{v_{out}^\perp}$ as in~\Cref{lem:tt-QMAC:security} and use the same argument to get
\begin{align*}
& \Tr\left( P_\sigma \left( \ketbra{\psi^*}{\psi^*}^{\otimes t} \otimes \Ex_{U\gets\overline{\Haar}_{n+m}\mid_V}\left[ \Tr_\aux \left(U^\dagger\ketbra{\phi^*}{\phi^*}U \right) \right] \right) \right) \\
& \leq \Tr\left( P_\sigma \left( \ketbra{\psi^*}{\psi^*}^{\otimes t} \otimes \frac{\Tr_\aux \left( I_{V^\perp_{out}} \right)}{\dim(V^\perp_{out})} \right) \right).
\end{align*}

\noindent As there is a $\frac{1}{t+1}$ fraction of $\sigma$'s that belong to the first case, we have
\begin{align*}
\Cref{eq:winning_prob_perm} 
\leq \frac{1}{t+1} + \frac{1}{(t+1)!} \cdot \sum_{\substack{\sigma\in S_{t+1}:\\ \sigma(t+1)\neq t+1}} \Tr\left( P_\sigma \left( \ketbra{\psi^*}{\psi^*}^{\otimes t} \otimes \frac{\Tr_\aux \left( I_{V^\perp_{out}} \right)}{\dim(V^\perp_{out})} \right) \right).
\end{align*}
Now, let $\sum_i \lambda_i \ketbra{\lambda_i}{\lambda_i}$ be the spectral decomposition of $\frac{\Tr_\aux \left( I_{V^\perp_{out}} \right)}{\dim(V^\perp_{out})}$. We finally have
\begin{align*}
\Cref{eq:winning_prob_perm} 
\leq & \frac{1}{t+1} + \frac{1}{(t+1)!} \cdot \sum_i \lambda_i \cdot \sum_{\substack{\sigma\in S_{t+1}:\\ \sigma(t+1)\neq t+1}} \Tr\left( P_\sigma \left( \ketbra{\psi^*}{\psi^*}^{\otimes t} \otimes \ketbra{\lambda_i}{\lambda_i} \right) \right) \\
= & \frac{1}{t+1} + \frac{t}{t+1} \cdot \sum_i \lambda_i |\inner{\psi^*}{\lambda_i}|^2 \\
= & \frac{1}{t+1} + \frac{t}{t+1} \cdot \expbraket{\psi^*}{\frac{\Tr_\aux \left( I_{V^\perp_{out}} \right)}{\dim(V^\perp_{out})}}{\psi^*} \\
\leq & \frac{1}{t+1} + \frac{t}{t+1} \cdot \norm{\frac{\Tr_\aux \left( I_{V^\perp_{out}} \right)}{\dim(V^\perp_{out})}}_\infty \\
= & \frac{1}{t+1} + \negl(\secp),
\end{align*}

\noindent where the last inequality follows from the calculation of~\Cref{eq:average_fidelity}. This finishes the proof of~\Cref{thm:t_1_1/t}.
\end{proof}

%% file: application_proofs/ucf-mac.tex
\begin{proof}\ifllncs[Proof of~\Cref{thm:uncom_unforge}]\else\fi
By the security of $\pri$, we replace it with a Haar isometry $\haarisometry$ in the construction. Fix $\lambda$ and queries $\ket{\psi_1},\dots,\ket{\psi_q}$. Suppose $\ket{\psi^*} := C\ket{0^n}$ is orthogonal to all previous queries. Similar to~\Cref{lem:tt-QMAC:security}, we consider the Haar unitary to be partially sampled. From~\Cref{fact:average_inner_product}, the success probability of the forger is
\begin{align*}
\Ex_{\haarisometry\mid_V}[\expbraket{\psi^*}{\haarisometry^{-1}(\ket{\phi^*})}{\psi^*}]
= \negl(\secp)
\end{align*}
from the calculation of~\Cref{eq:average_fidelity}.
\end{proof}

%% file: application_proofs/length_ext.tex
\begin{proof}\ifllncs[Proof of~\Cref{thm:length_extension}]\else\fi
By security of the PRI, we will consider
\begin{align*}
\rho' := \Ex_{\ket{\theta}\gets\Haar_{2n}, \haarisometry} \left[ (I_n \otimes \haarisometry)^{\otimes t} \ketbra{\theta}{\theta}^{\otimes t} (I_n \otimes \haarisometry^\dagger)^{\otimes t} \right].
\end{align*}
It's sufficient to prove that $\TD(\rho',\sigma) = O(t!t^2/2^{n+m}+t^2/2^n)$. Expanding $t$-copies of a Haar state in the type basis (\Cref{fact:avg-haar-random}), we can write $\rho'$ as
\[
\rho' = \Ex_{T \gets [t+1]^N\mid_{\hamming(T)=t}, \haarisometry} \left[ (I_n \otimes \haarisometry)^{\otimes t} \ketbra{\type_T}{\type_T} (I_n \otimes \haarisometry^\dagger)^{\otimes t} \right],
\]
where $N := 2^{2n}$.

Given a type $T \in [t+1]^N$ such that $\setT(T) = \Vec{x}||\Vec{y} = \set{x_1||y_1,\dots,x_t||y_t}$, where $x_i,y_i\in\bit^n$. We say $T$ is \emph{good} if and only if (1) all $x_i$'s are pairwise distinct, and (2) all $y_i$'s are pairwise distinct. The following observation regarding good types is the crux of the proof. Intuitively, the Haar isometry scrambles the last $n$ bits of every element in $\setT(T)$, i.e., $\Vec{y}$, to a random vector with no repeating coordinates.

\begin{lemma} \label{lemma:type_suffix_mix}
For every good type $T$ with $\setT(T) = \Vec{x}||\Vec{y} = \set{x_1||y_1,\dots,x_t||y_t}$, where $x_i,y_j\in\bit^n$ and w.l.o.g. $x_1 < x_2 < \dots < x_t$, let
\begin{align*}
\rho_{\sf left} := \Ex_{\haarisometry} \left[ (I_n \otimes \haarisometry)^{\otimes t} \ketbra{\type_T}{\type_T} (I_n \otimes \haarisometry^\dagger)^{\otimes t} \right]
\end{align*}
and
\begin{align*}
\rho_{\sf right} := \Ex\left[ \ketbra{\type_{T'}}{\type_{T'}}: 
\substack{(z_1,z_2,\dots,z_t) \ugets \cS_{n+m,t}, \\ T' := \type(\Vec{x}||\Vec{z})} \right],
\end{align*}
where $\cS_{n+m,t} := \set{ \Vec{z} = (z_1,z_2,\dots,z_t) \in \bit^{(n+m)t}: \vec{z} \text{ has no repeating coordinates} }$. Then $\TD \left( \rho_{\sf left}, \rho_{\sf right} \right) \leq O\left( \frac{t!t^2}{2^{n+m}} \right)$.
\end{lemma}

\begin{proof}[Proof of~\Cref{lemma:type_suffix_mix}]
By the definition of type vectors (\Cref{def:type_states}) and the premise that $T$ is good, we have
\begin{align*}
\ketbra{\type_T}{\type_T}
= \frac{1}{t!} \sum_{\sigma,\pi\in S_t} \ketbra{\sigma(\Vec{x}||\Vec{y})}{\pi(\Vec{x}||\Vec{y})}
= \frac{1}{t!} \sum_{\sigma,\pi\in S_t} \ketbra{\sigma(\Vec{x})}{\pi(\Vec{x})} \otimes \ketbra{\sigma(\Vec{y})}{\pi(\Vec{y})}.
\end{align*}

\noindent Thus, it holds that
\begin{align*}
\rho_{\sf left}
= & \Ex_{\haarisometry} \left[ (I \otimes \haarisometry)^{\otimes t} \ketbra{\type_T}{\type_T} (I  \otimes \haarisometry^\dagger)^{\otimes t} \right] \\
= & \frac{1}{t!} \sum_{\sigma,\pi\in S_t} \ketbra{\sigma(\Vec{x})}{\pi(\Vec{x})} \otimes \Ex_{\haarisometry}[ \haarisometry^{\otimes t} \ketbra{\sigma(\Vec{y})}{\pi(\Vec{y})} (\haarisometry^\dagger)^{\otimes t}] \\
= & \frac{1}{t!} \sum_{\sigma,\pi\in S_t} \ketbra{\sigma(\Vec{x})}{\pi(\Vec{x})} \otimes \Ex_{U \gets \overline{\Haar_{n+m}}}[ U^{\otimes t} \ketbra{\sigma(\Vec{y} \odot 0^m)}{\pi(\Vec{y} \odot 0^m)} (U^\dagger)^{\otimes t}] \\
= & \frac{1}{t!} \sum_{\sigma,\pi\in S_t} \ketbra{\sigma(\Vec{x})}{\pi(\Vec{x})} \otimes P_\sigma \Ex_{U \gets \overline{\Haar_{n+m}}} \left[ U^{\otimes t} \ketbra{\Vec{y} \odot 0^m}{\Vec{y} \odot 0^m} (U^\dagger)^{\otimes t} \right] P^\dagger_\pi,
\end{align*}
where $\Vec{y} \odot 0^m$ denotes $(y_1||0^m,\dots,y_t||0^m)$. Note that $\Vec{y} \odot 0^m$ also has no repeating coordinates. From unitary invariance of trace distance and \Cref{lem:haar_perp}, for every $\sigma,\pi\in S_t$,
\begin{align*}
\TD \Bigg( 
P_\sigma \Ex_{U \gets \overline{\Haar_{n+m}}} \left[ U^{\otimes t} \ketbra{\Vec{y} \odot 0^m}{\Vec{y} \odot 0^m} (U^\dagger)^{\otimes t}\right] & P^\dagger_\pi, \\
 P_\sigma \Ex&\left[ \ketbra{\Vec{z}}{\Vec{z}}: \Vec{z} \ugets \cS_{n+m,t} \right] P^\dagger_\pi 
\Bigg) \\
=  \TD \Bigg( 
\Ex_{U \gets \overline{\Haar_{n+m}}} \left[ U^{\otimes t} \ketbra{\Vec{y} \odot 0^m}{\Vec{y} \odot 0^m} (U^\dagger)^{\otimes t} \right]&, 
\Ex\left[ \ketbra{\Vec{z}}{\Vec{z}}: \Vec{z} \ugets \cS_{n+m,t} \right]
\Bigg) \\
\leq & O(t^2/2^{n+m}).
\end{align*}

\noindent By triangle inequalities over all $\sigma,\pi\in S_t$, the density matrix $\rho_{\sf left}$ is $O(t!t^2/2^{n+m})$-close to
\begin{align*}
& \frac{1}{t!} \sum_{\sigma,\pi\in S_t} \ketbra{\sigma(\Vec{x})}{\pi(\Vec{x})} \otimes P_\sigma \Ex\left[ \ketbra{\Vec{z}}{\Vec{z}}: \Vec{z} \ugets \cS_{n+m,t} \right] P^\dagger_\pi \\
= & \frac{1}{t!} \sum_{\sigma,\pi\in S_t} \ketbra{\sigma(\Vec{x})}{\pi(\Vec{x})} \otimes \Ex\left[ \ketbra{\sigma(\Vec{z})}{\pi(\Vec{z})}: \Vec{z} \ugets \cS_{n+m,t} \right] \\
= & \Ex_{\Vec{z} \ugets \cS_{n+m,t}} \left[ \frac{1}{t!} \sum_{\sigma,\pi\in S_t} \ketbra{\sigma(\Vec{x})}{\pi(\Vec{x})} \otimes  \ketbra{\sigma(\Vec{z})}{\pi(\Vec{z})} \right] \\
= & \Ex\left[ \ketbra{\type_{T'}}{\type_{T'}}:
\substack{(z_1,z_2,\dots,z_t) \ugets \cS_{n+m,t}, \\ T' := \type(\Vec{x}||\Vec{z})} \right] \\
= & \rho_{\sf right}.
\end{align*}
This finishes the proof of~\Cref{lemma:type_suffix_mix}.
\end{proof}

\noindent Now, we continue proving~\Cref{thm:length_extension}. In density matrix $\rho'$, the probability of a $t$-size type $T$ sampled uniformly from $[t+1]^N$ being good is at least $1 - O(t^2/2^n)$ from~\Cref{fact:random_type_halves_collision}. Hence, $\TD(\rho',\rho'_{\sf good}) = O(t^2/2^n)$, where
\begin{align*}
\rho'_{\sf good} := \Ex_{ T \gets [t+1]^N\mid_{\hamming(T)=t\land T \text{ is good}} } \Ex_{\haarisometry} \left[ (I_n \otimes \haarisometry)^{\otimes t} \ketbra{\type_T}{\type_T} (I_n \otimes \haarisometry^\dagger)^{\otimes t} \right].
\end{align*}
Then applying~\Cref{lemma:type_suffix_mix} to every (good) $T$ in $\rho'_{\sf good}$, we have $\TD(\rho'_{\sf good},\rho'') = O(t!t^2/2^{n+m})$ where 
\begin{align*}
\rho'' 
& := \Ex \left[ \ketbra{\type_{T''}}{\type_{T''}}: 
    \substack{ 
        T \gets [t+1]^N\mid_{\hamming(T)=t \land T \text{ is good}}, \\
        \setT(T) = \set{x_1||y_1,\dots,x_t||y_t}\ s.t.\ x_1< \dots <x_t, \\ 
        (z_1,z_2,\dots,z_t) \gets \cS_{n+m,t}, \\
        T'' := \type(x_1||z_1,\dots,x_t||z_t)
    } 
\right] \\
& = \Ex \left[ \ketbra{\type_{T''}}{\type_{T''}}: 
    \substack{ 
        (x_1,x_2,\dots,x_t) \gets \cS_{n,t}, \\ 
        (z_1,z_2,\dots,z_t) \gets \cS_{n+m,t}, \\
        T'' := \type(x_1||z_1,\dots,x_t||z_t)
    } 
\right].
\end{align*}
Again, we expand $\sigma$ in the type basis (\Cref{fact:avg-haar-random}),
\[
\sigma = \Ex_{T \gets [t+1]^M\mid_{\hamming(T)=t} } \left[ \ketbra{\type_T}{\type_T} \right],
\]
where $M := 2^{2n+m}$. To upper bound $\TD(\rho'',\sigma)$, it's sufficient to bound the statistical distance between $T''$ defined in $\rho''$ and a uniformly random $t$-size $T$ in $[t+1]^M$. This is at most $O(t^2/2^n) + O(t^2/2^{n+m})$ from~\Cref{fact:random_type_halves_collision}. Combining the bounds completes the proof of~\Cref{thm:length_extension}.
\end{proof}

%% file: Appommittedproof.tex
\section{Omitted Proofs in~\Cref{sec:prelim}} \label{app:proof_prelim}
\pnote{I'd suggest bringing this entire section into prelims} \fatih{done}
\begin{fact}[Projection onto symmetric subspace stabilizes type states] \label{fact:sym_stab_type}
For all $N,t \in \N$ and $T\in [t+1]^N$ such that $\hamming(T) = t$,
\begin{align*}
    \Pi^{N,t}_\sym \ket{\type_T} = \ket{\type_T}.
\end{align*}
\end{fact}

\begin{fact}[Average inner product with Haar states] \label{fact:average_inner_product}
For any $N\in\N$ and fixed $\ket{\psi}\in\C^N$, $\Ex_{\ket{\vartheta}\gets\Haar(\C^N)}[|\inner{\psi}{\vartheta}|^2] = 1/N$.
\end{fact}

\begin{lemma}[Theorem~7.5 in~\cite{WatrousBook}, restated] \label{lem:multi_span_sym}
For all $N,t \in \N$, there exists a finite set $\cS$ of vectors in $\C^N$ such that $\vee^t\C^N = \mathsf{span}\set{\ket{v}^{\otimes t}: \ket{v}\in\cS}$.
\end{lemma}

\begin{fact} \label{fact:random_type_halves_collision}
Let $T$ be sampled uniformly from $[t+1]^{2^{\ell+k}}$ conditioned on $\hamming(T) = t$, where $\setT(T) = \set{x_1||y_1,x_2||y_2,\dots,x_t||y_t}$ and $x_i\in\bit^\ell$, $y_j\in\bit^k$. Then
$\Pr[\exists i\neq j\ s.t.\ x_i = x_j \lor y_i = y_j ] = O(t^2/2^\ell) + O(t^2/2^k)$.
\end{fact}

\begin{proof}[Proof of~\Cref{lem:type_struc}]
Notice that by \Cref{def:type_states}, 
$$\ket{\type_T} = \sqrt{\frac{\prod_{i\in\supp(T)} \freq{T}{i}!}{t!}}\sum_{\substack{\vec{v}\in[N]^t\\ \type(\vec{v}) = T}}\ket{\vec{v}}.$$
Hence, 
$$\ketbra{\type_T}{\type_T} = \frac{\prod_{i\in\supp(T)} \freq{T}{i}!}{t!}\sum_{\substack{\vec{v}\in[N]^t\\ \type(\vec{v}) = T\\ \vec{v'}\in[N]^t\\ \type(\vec{v'}) = T}}\ketbra{\vec{v}}{\vec{v'}}.$$
Note that since $\type(\vec{v'}) = \type(\vec{v})$, $\vec{v'} = \sigma(\vec{v})$ for some $\sigma\in S_{t}$. Notice that the number of $\vec{v'}$ with $\type(\vec{v'}) = T$ is $\frac{t!}{\prod_{i\in T} i!}$. Summing over all permutations $\sigma$, each $\vec{v'}$ is repeated exactly $\prod_{i\in T} i!$ times. 
Hence, 
$$\ketbra{\type_T}{\type_T} = \frac{1}{t!}\sum_{\sigma\in S_t}\sum_{\substack{\vec{v}\in[N]^t\\ \type(\vec{v}) = T}}\ketbra{\vec{v}}{\sigma(\vec{v})}.$$
\end{proof}

\section{Omitted Proofs in~\Cref{sec:defs}} \label{app:proof_defs}
\begin{proof}[Proof of~\Cref{thm:PRIimplyPRFS}]
Let $\pri$ be an $(n,n+m)$-PRI. The construction of the PRFSG $F$ is the following: on input $k\in\bit^\secp$ and $x\in\bit^n$, apply CNOT on $\ket{x}\ket{0^n}$ to get $\ket{x}\ket{x}$, and then output $\ket{x}\otimes\pri_k\ket{x}$. 

We prove the classical-accessible selective security of $F$ via reduction. Suppose there exists a QPT adversary $\adversary$, polynomials $q(\cdot),t(\cdot)$ and a set of indices $\set{x_1,x_2,\dots,x_q}$ where $x_i\in\bit^n$ and $q(\secp) = \poly(\secp)$ such that, w.l.o.g.,
\begin{multline*}
\Pr_{k\gets\bit^\secp}[\adversary_\secp(1^\secp,x_1,\dots,x_q,F(k,x_1)^{\otimes t},\dots,F(k,x_q)^{\otimes t}) = 1] \ge \\
\Pr_{\ket{\theta_1},\dots,\ket{\theta_q}\gets\Haar_{n+m}}[\adversary_\secp(1^\secp,x_1,\dots,x_q,\ket{\theta_1}^{\otimes t},\dots,\ket{\theta_q}^{\otimes t}) = 1]
+ \nu(\secp),
\end{multline*}
where $\nu(\secp)$ is non-negligible. We construct a distinguisher $\distinguisher$ that uses $\adversary$ to break the security of the underlying PRI. Upon receiving queries $\set{x_1,x_2,\dots,x_q}$ from $\adversary$, the distinguisher $\distinguisher$ simply use its oracle access to $\mathcal{O}$ to reply $\bigotimes_{i=1}^q \ket{x_i} \otimes (\mathcal{O}\ket{x_i})^{\otimes t}$. Then $\distinguisher$ outputs whatever $\adversary$ outputs. Hence, the distinguishing advantage of $\distinguisher$ is
\begin{multline*}
\Bigg| 
\Pr_{k\gets\bit^\secp}[\adversary_\secp(1^\secp,x_1,\dots,x_q,F(k,x_1)^{\otimes t},\dots,F(k,x_q)^{\otimes t}) = 1] - \\
\Pr_{\haarisometry}[\adversary_\secp(1^\secp,x_1,\dots,x_q,(\haarisometry\ket{x_1})^{\otimes t},\dots,(\haarisometry\ket{x_q})^{\otimes t}) = 1]
\Bigg|.
\end{multline*}
From~\Cref{lem:haar_perp_to_iid}, it is at least $\nu(\secp) - O(q^3t^2/2^{n+m})$, which is non-negligible. But it contradicts the security of the underlying PRI.
\end{proof}

\ifllncs
\section{Omitted Proofs in~\Cref{sec:construction}}
\label{app:proof_construction}

\input{construction_proofs/info_fig}
\subsection{Omitted Proofs in~\Cref{sec:comp_security}}
\label{app:comp_security}
\input{construction_proofs/comp_to_info_proof}
\subsection{Omitted Proofs in~\Cref{sec:haar_invariance}}
\label{app:haar_invariance}
\input{construction_proofs/invar_to_sec_proof}
\subsection{Omitted Proofs in~\Cref{sec:invariance_rho_uni}}
\label{app:invariance_rho_uni}
\input{construction_proofs/rho_uni_proof}
\subsection{Omitted Proofs in~\Cref{sec:closeness_to_rho_uni}}
\label{app:closeness_to_rho_uni}

\input{construction_proofs/type_distinct_proof}
\input{construction_proofs/single_input_type_proof}
\input{construction_proofs/haar_type_proof}
\else
\fi

\ifllncs
\section{Omitted Proofs in~\Cref{sec:applications}} \label{app:proof_applications}
 
\subsection{Omitted Proofs in~\Cref{sec:qmacs}} 
\label{app:proof_qmac}

\input{application_proofs/tt-mac}
\input{application_proofs/t1-mac}
\input{application_proofs/ucf-mac}
\subsection{Omitted Proofs in~\Cref{sec:prs_extension}}
\label{app:proof_extension}

\input{application_proofs/length_ext}
\else
\fi

%% file: construction_proofs/comp_to_info_proof.tex
\begin{proof}[Proof of~\Cref{lem:comp_to_info}] 
We prove this by a standard hybrid argument. Consider the following hybrids: 

\begin{itemize}
\item Hybrid $\hybrid_0$: The oracle is $\prfsi_{\secparam}$ defined in~\Cref{fig:prp_prs}.
\item Hybrid $\hybrid_1$: The oracle is the same as $\prfsi_{\secparam}$ except that $\prf$ is replaced by a random function.
\item Hybrid $\hybrid_2$: The oracle is $G_{(f,\pi)}$ defined in~\Cref{fig:info_prp_prs}.
\end{itemize}

\newcommand{\oracle}{{\cal O}}
\begin{claim}
Assuming the quantum-query security of $\prf$, the output distributions of the hybrids $\hybrid_0$ and $\hybrid_1$ are computationally indistinguishable. 
\end{claim}
\begin{proof}
Suppose there exists some QPT algorithm $\adversary$ that distinguishes Hybrid~0 from Hybrid~1 with a non-negligible advantage $\nu$. We'll construct a reduction $\distinguisher$ that given oracle access to $\oracle$ distinguishes whether $\oracle$ is either the $\prf$ oracle or a random function with the same advantage $\nu$ by using $\adversary$. Upon receiving a query $\ket{\psi}$ from $\adversary$, the reduction $\distinguisher$ responds by first applying  $H^{\otimes m}\otimes I_{n}$ on $\ket{\psi}\ket{0^m}$, querying $\oracle$ and finally, computing $g(k_2,\cdot)$, where $k_2$ is sampled uniformly at random from $\{0,1\}^{\secparam_1}$. Since $\distinguisher$ perfectly simulates the distributions of oracles in hybrids $\hybrid_0$ and $\hybrid_1$, it has the same distinguishing advantage as that of $\adversary$. However, this contradicts the post-quantum security of the underlying $\prf$.
\end{proof}

\begin{claim}
Assuming the quantum-query security of $\prp$, the output distributions of the hybrids $\hybrid_1$ and $\hybrid_2$ are computationally indistinguishable. 
\end{claim}
\begin{proof}
Similarly, suppose there exists some QPT algorithm $\adversary'$ that distinguishes hybrids $\hybrid_1$ from $\hybrid_2$ with a non-negligible advantage $\nu'$. We'll construct a reduction $\distinguisher'$ that given access to an oracle $\oracle$ distinguishes where $\oracle$ implements $\prp$ or a random permutation with the same advantage $\nu'$ by using $\adversary'$. Suppose the number of queries made by $\adversary'$ is $q = \poly(\secp)$. Since each query to the oracle needs to invoke the random function once, the number of queries to the random function is also $q$. Upon receiving a query $\ket{\psi}$ from $\adversary'$, the reduction $\distinguisher'$ responds by first applying  $H^{\otimes m}\otimes I_{n}$ on $\ket{\psi}\ket{0^m}$, applying a $2q$-wise independent hash function and finally, querying $\oracle$. From~\Cref{thm:zha12}, it follows that 
a $2q$-wise independent hash function perfectly simulates a random function. Thus, $\distinguisher'$ perfectly simulates the distributions of the oracles in the hybrids $\hybrid_1$ and $\hybrid_2$ and thus has the same distinguishing advantage as that of $\adversary'$. However, this contradicts the post-quantum security of the underlying $\prp$. This completes the proof.
\end{proof} 
\end{proof}

%% file: construction_proofs/invar_to_sec_proof.tex
\begin{proof}[Proof of~\Cref{lem:comp_to_invar}]
    We start by proving the following lemma,
    \begin{lemma}
        \label{lem:rand_iso_is_invar}
         Let $C_{U}$ be the channel that appends $\ket{0^m}$ and applies the unitary $U$. Let for any input $\rho\in {\cal D}(\C^{2^{nq+\ell}})$, $\rho_{G}$ be as defined in \Cref{lem:comp_to_invar}. Then, for any input $\rho\in {\cal D}(\C^{2^{nq+\ell}})$,
         $$\E_{U\leftarrow\overline{\Haar_{m+n}}}\left[\left(I_{\ell}\otimes C_{U}^{\otimes q}\right)(\rho)\right] = \E_{ U\leftarrow\overline{\Haar_{m+n}}} \left[\left(I_{\ell}\otimes (U)^{\otimes q}\right)(\rho_G)\right].$$
    \end{lemma}
    \begin{proof}
        Notice that $G_{(f,\pi)}$ can be seen as the following operations, 
        \begin{enumerate}
            \item Apply $C_{\ket{0^{m}}}$ where $C_{\ket{0^{m}}}$ appends $\ket{0^{m}}$.
            \item Apply $I\otimes H^{\otimes m}$.
            \item Apply the unitary $O_f$, where the action of $O_f$ on any $\ket{x}$ is $O_{f}\ket{x} = \omega_p^{x}\ket{x}$.
            \item Apply the unitary $O_{\pi}$, where the action of $O_\pi$ on any $\ket{x}$ is $O_{\pi}\ket{x} = \ket{\pi(x)}$.
        \end{enumerate}
        Similarly, $C_U$ can be seen as the following operations, 
        \begin{enumerate}
            \item Apply $C_{\ket{0^{m}}}$ where $C_{\ket{0^{m}}}$ appends $\ket{0^{m}}$.
            \item Apply $U$.
        \end{enumerate}
        Hence, if we look at 
        $$\xi = \E_{ U\leftarrow\overline{\Haar_{m+n}}} \left[\left(I_{\ell}\otimes (U)^{\otimes q}\right)(\rho_G)\right].$$
        Then from the definition of $\rho_G$,
        $$\xi = \E_{\substack{ U\leftarrow\overline{\Haar_{m+n}}\\ (f,\pi)\leftarrow (\mathcal{F}_{2^{n+m},p},S_{2^{n+m}})}} \left[\left(I_{\ell}\otimes (UG_{(f,\pi)})^{\otimes q}\right)(\rho)\right],$$
        From the definition of $G_{(f,\pi)}$, 
        $$\xi = \E_{\substack{U\leftarrow\overline{\Haar_{m+n}}\\ (f,\pi)\leftarrow (\mathcal{F}_{2^{n+m},p},S_{2^{n+m}})}} \left[\left(I_{\ell}\otimes (UO_{\pi}O_{f}(I\otimes H^{\otimes m})C_{\ket{0^m}})^{\otimes q}\right)(\rho)\right].$$
        By unitary invariance of the Haar measure, we get 
        $$\xi = \E_{\substack{U\leftarrow\overline{\Haar_{m+n}}\\ (f,\pi)\leftarrow (\mathcal{F}_{2^{n+m},p},S_{2^{n+m}})}} \left[\left(I_{\ell}\otimes (UC_{\ket{0^m}})^{\otimes q}\right)(\rho)\right],$$
        or 
        $$\xi = \E_{U\leftarrow\overline{\Haar_{m+n}}} \left[\left(I_{\ell}\otimes (C_U)^{\otimes q}\right)(\rho)\right].$$
        Hence, 
        $$\E_{U\leftarrow\overline{\Haar_{m+n}}}\left[\left(I_{\ell}\otimes C_{U}^{\otimes q}\right)(\rho)\right] = \E_{ U\leftarrow\overline{\Haar_{m+n}}} \left[\left(I_{\ell}\otimes (U)^{\otimes q}\right)(\rho_G)\right].$$
    \end{proof}
    Now we prove the above theorem using the Hybrid argument. Let $C_{\prfsi_k}$ be the channel that applies $\prfsi(k,\cdot)$ as given in \Cref{fig:prp_prs}.
    \paragraph{Hybrid 1.} Sample a random key $k\in\{0,1\}^{\secparam}$ and apply $I_{\ell}\otimes C_{\prfsi_k}^{\otimes q}$ on $\rho$ and output the result.
    \paragraph{Hybrid 2.} Sample a random function $f\leftarrow\mathcal{F}_{2^{n+m},p}$ and $\pi\leftarrow S_{2^{n+m}}$. Apply $I_{\ell}\otimes G_{(f,\pi)}^{\otimes q}$ on $\rho$ and output the result.
    \paragraph{Hybrid  3.} Sample a unitary $U\leftarrow\overline{\Haar_{n+m}}$. Apply $I_{\ell}\otimes C_{U}^{\otimes q}$ on $\rho$ and output the result. \\
    The output of Hybrid 1 and Hybrid 2 is computationally indistinguishable by \Cref{lem:comp_to_info}. 
    \noindent Notice that by \Cref{lem:rand_iso_is_invar}, the output of Hybrid 3 is $$\E_{ U\leftarrow\overline{\Haar_{m+n}}} \left[\left(I_{\ell}\otimes (U)^{\otimes q}\right)(\rho_G)\right].$$ Then by assumption, the output of Hybrid 3 is negligibly close to the output of Hybrid 2. 
    Hence, combining the above results, we get that \Cref{fig:prp_prs} is a $\qclass$-secure Pseudorandom Isometry.
\end{proof}

%% file: construction_proofs/single_input_type_proof.tex
\input{construction_proofs/comp_outer_query_proof}
\input{construction_proofs/multicopy_proof}

%% file: construction_proofs/haar_type_proof.tex
\input{construction_proofs/tuni_info_proof}
\input{construction_proofs/haar_info_proof}

%% file: PRIimpliedPRFS.tex

%% file: tcs.bib
@article{KNPPZ21,
  title={Generating random quantum channels},
  author={Kukulski, Ryszard and Nechita, Ion and Pawela, {\L}ukasz and Pucha{\l}a, Zbigniew and {\.Z}yczkowski, Karol},
  journal={Journal of Mathematical Physics},
  volume={62},
  number={6},
  year={2021},
  publisher={AIP Publishing}
}

@article{HBK23,
  title={Pseudorandom unitaries are neither real nor sparse nor noise-robust},
  author={Haug, Tobias and Bharti, Kishor and Koh, Dax Enshan},
  journal={arXiv preprint arXiv:2306.11677},
  year={2023}
}

@article{ZS00,
  title={Truncations of random unitary matrices},
  author={Zyczkowski, Karol and Sommers, Hans-J{\"u}rgen},
  journal={Journal of Physics A: Mathematical and General},
  volume={33},
  number={10},
  pages={2045},
  year={2000},
  publisher={IOP Publishing}
}

@article{Ras12,
  title={Relations for certain symmetric norms and anti-norms before and after partial trace},
  author={Rastegin, Alexey E},
  journal={Journal of Statistical Physics},
  volume={148},
  pages={1040--1053},
  year={2012},
  publisher={Springer}
}

@article{BS20,
  title={Almost public quantum coins},
  author={Behera, Amit and Sattath, Or},
  journal={arXiv preprint arXiv:2002.12438},
  year={2020}
}

@inproceedings{BZ13,
  title={Quantum-secure message authentication codes},
  author={Boneh, Dan and Zhandry, Mark},
  booktitle={Advances in Cryptology--EUROCRYPT 2013: 32nd Annual International Conference on the Theory and Applications of Cryptographic Techniques, Athens, Greece, May 26-30, 2013. Proceedings 32},
  pages={592--608},
  year={2013},
  organization={Springer}
}

@inproceedings{Portmann17,
  title={Quantum authentication with key recycling},
  author={Portmann, Christopher},
  booktitle={Advances in Cryptology--EUROCRYPT 2017: 36th Annual International Conference on the Theory and Applications of Cryptographic Techniques, Paris, France, April 30--May 4, 2017, Proceedings, Part III 36},
  pages={339--368},
  year={2017},
  organization={Springer}
}

@inproceedings{DNS12,
  title={Actively secure two-party evaluation of any quantum operation},
  author={Dupuis, Fr{\'e}d{\'e}ric and Nielsen, Jesper Buus and Salvail, Louis},
  booktitle={Annual Cryptology Conference},
  pages={794--811},
  year={2012},
  organization={Springer}
}

@inproceedings{AM17,
  title={Quantum non-malleability and authentication},
  author={Alagic, Gorjan and Majenz, Christian},
  booktitle={Advances in Cryptology--CRYPTO 2017: 37th Annual International Cryptology Conference, Santa Barbara, CA, USA, August 20--24, 2017, Proceedings, Part II 37},
  pages={310--341},
  year={2017},
  organization={Springer}
}

@inproceedings{AGM18,
  title={Unforgeable quantum encryption},
  author={Alagic, Gorjan and Gagliardoni, Tommaso and Majenz, Christian},
  booktitle={Annual international conference on the theory and applications of cryptographic techniques},
  pages={489--519},
  year={2018},
  organization={Springer}
}

@inproceedings{AMRS20,
  title={Quantum-access-secure message authentication via blind-unforgeability},
  author={Alagic, Gorjan and Majenz, Christian and Russell, Alexander and Song, Fang},
  booktitle={Advances in Cryptology--EUROCRYPT 2020: 39th Annual International Conference on the Theory and Applications of Cryptographic Techniques, Zagreb, Croatia, May 10--14, 2020, Proceedings, Part III 39},
  pages={788--817},
  year={2020},
  organization={Springer}
}

@inproceedings{GYZ17,
  title={New security notions and feasibility results for authentication of quantum data},
  author={Garg, Sumegha and Yuen, Henry and Zhandry, Mark},
  booktitle={Advances in Cryptology--CRYPTO 2017: 37th Annual International Cryptology Conference, Santa Barbara, CA, USA, August 20--24, 2017, Proceedings, Part II 37},
  pages={342--371},
  year={2017},
  organization={Springer}
}

@article{BHH16,
  title={Local random quantum circuits are approximate polynomial-designs},
  author={Brandao, Fernando GSL and Harrow, Aram W and Horodecki, Micha{\l}},
  journal={Communications in Mathematical Physics},
  volume={346},
  pages={397--434},
  year={2016},
  publisher={Springer}
}

@inproceedings{AE07CCC,
  title={Quantum t-designs: t-wise independence in the quantum world},
  author={Ambainis, Andris and Emerson, Joseph},
  booktitle={Twenty-Second Annual IEEE Conference on Computational Complexity (CCC'07)},
  pages={129--140},
  year={2007},
  organization={IEEE}
}

@article{Harrow13church,
  title={The church of the symmetric subspace},
  author={Harrow, Aram W},
  journal={arXiv preprint arXiv:1308.6595},
  year={2013}
}

@article{BCHKP21,
  title={Models of quantum complexity growth},
  author={Brand{\~a}o, Fernando GSL and Chemissany, Wissam and Hunter-Jones, Nicholas and Kueng, Richard and Preskill, John},
  journal={PRX Quantum},
  volume={2},
  number={3},
  pages={030316},
  year={2021},
  publisher={APS}
}

@inproceedings{DN02,
  title={Expanding pseudorandom functions; or: From known-plaintext security to chosen-plaintext security},
  author={Damg{\aa}ard, Ivan and Nielsen, Jesper Buus},
  booktitle={Annual International Cryptology Conference},
  pages={449--464},
  year={2002},
  organization={Springer}
}

@inproceedings{ABGKR14,
  title={Candidate weak pseudorandom functions in $\mathsf{AC}^0$ $\circ$ $\mathsf{Mod}_2$},
  author={Akavia, Adi and Bogdanov, Andrej and Guo, Siyao and Kamath, Akshay and Rosen, Alon},
  booktitle={Proceedings of the 5th conference on Innovations in theoretical computer science},
  pages={251--260},
  year={2014}
}

@article{HarrowThesis,
  title={Applications of coherent classical communication and the schur transform to quantum information theory},
  author={Harrow, Aram W.},
  year={2005},
  journal={PhD thesis, Massachusetts Institute of Technology}
}

@article{AGM21quantum,
  doi = {10.22331/q-2021-12-16-603},
  url = {https://doi.org/10.22331/q-2021-12-16-603},
  title = {Can you sign a quantum state?},
  author = {Alagic, Gorjan and Gagliardoni, Tommaso and Majenz, Christian},
  journal = {{Quantum}},
  issn = {2521-327X},
  publisher = {{Verein zur F{\"{o}}rderung des Open Access Publizierens in den Quantenwissenschaften}},
  volume = {5},
  pages = {603},
  month = dec,
  year = {2021}
}

@inproceedings{BCG+02FOCS,
  title={Authentication of quantum messages},
  author={Barnum, Howard and Cr{\'e}peau, Claude and Gottesman, Daniel and Smith, Adam and Tapp, Alain},
  booktitle={The 43rd Annual IEEE Symposium on Foundations of Computer Science, 2002. Proceedings.},
  pages={449--458},
  year={2002},
  organization={IEEE}
}

@book{WatrousBook,
  title={The theory of quantum information},
  author={Watrous, John},
  year={2018},
  publisher={Cambridge university press}
}

@article{KNY08,
  title={The efficiency of quantum identity testing of multiple states},
  author={Kada, Masaru and Nishimura, Harumichi and Yamakami, Tomoyuki},
  journal={Journal of Physics A: Mathematical and Theoretical},
  volume={41},
  number={39},
  pages={395309},
  year={2008},
  publisher={IOP Publishing}
}

@article{BBBS23,
  title={Pseudorandomness with proof of destruction and applications},
  author={Behera, Amit and Brakerski, Zvika and Sattath, Or and Shmueli, Omri},
  journal={Cryptology ePrint Archive},
  year={2023}
}

@article{BBA+97,
  title={Stabilization of quantum computations by symmetrization},
  author={Barenco, Adriano and Berthiaume, Andre and Deutsch, David and Ekert, Artur and Jozsa, Richard and Macchiavello, Chiara},
  journal={SIAM Journal on Computing},
  volume={26},
  number={5},
  pages={1541--1557},
  year={1997},
  publisher={SIAM}
}

@inproceedings{GJMZ23,
  title={Commitments to quantum states},
  author={Gunn, Sam and Ju, Nathan and Ma, Fermi and Zhandry, Mark},
  booktitle={Proceedings of the 55th Annual ACM Symposium on Theory of Computing},
  pages={1579--1588},
  year={2023}
}

@article{GGH+15,
 author = {Gutoski, Gus and Hayden, Patrick and Milner, Kevin and Wilde, Mark M.},
 title = {Quantum Interactive Proofs and the Complexity of Separability Testing},
 year = {2015},
 pages = {59--103},
 doi = {10.4086/toc.2015.v011a003},
 publisher = {Theory of Computing},
 journal = {Theory of Computing},
 volume = {11},
 number = {3},
 URL = {https://theoryofcomputing.org/articles/v011a003},
}

@article{Zha16,
  title={A note on quantum-secure PRPs},
  author={Zhandry, Mark},
  journal={arXiv preprint arXiv:1611.05564},
  year={2016}
}

@book{Meckes19,
  title={The random matrix theory of the classical compact groups},
  author={Meckes, Elizabeth S},
  volume={218},
  year={2019},
  publisher={Cambridge University Press}
}

@inproceedings{AGQY22,
  title={Pseudorandom (Function-Like) Quantum State Generators: New Definitions and Applications},
  author={Ananth, Prabhanjan and Gulati, Aditya and Qian, Luowen and Yuen, Henry},
  booktitle={Theory of Cryptography Conference},
  pages={237--265},
  year={2022},
  organization={Springer}
}

@misc{MTdW00,
      title={Private quantum channels and the cost of randomizing quantum information}, 
      author={Mosca, Michele and Tapp, Alain and de Wolf, Ronald},
      year={2000},
      eprint={quant-ph/0003101},
      archivePrefix={arXiv}
}

@inproceedings{JLS18,
  author    = {Zhengfeng Ji and
               Yi{-}Kai Liu and
               Fang Song},
  editor    = {Hovav Shacham and
               Alexandra Boldyreva},
  title     = {Pseudorandom Quantum States},
  booktitle = {Advances in Cryptology - {CRYPTO} 2018 - 38th Annual International
               Cryptology Conference, Santa Barbara, CA, USA, August 19-23, 2018,
               Proceedings, Part {III}},
  series    = {Lecture Notes in Computer Science},
  volume    = {10993},
  pages     = {126--152},
  publisher = {Springer},
  year      = {2018},
  doi       = {10.1007/978-3-319-96878-0_5},
  bibsource = {dblp computer science bibliography, https://dblp.org}
}

@book{nielsen_chuang_2010, place={Cambridge}, title={Quantum Computation and Quantum Information: 10th Anniversary Edition}, DOI={10.1017/CBO9780511976667}, publisher={Cambridge University Press}, author={Nielsen, Michael A. and Chuang, Isaac L.}, year={2010}}

@inproceedings{BS19,
  author    = {Zvika Brakerski and
               Omri Shmueli},
  editor    = {Dennis Hofheinz and
               Alon Rosen},
  title     = {(Pseudo) Random Quantum States with Binary Phase},
  booktitle = {Theory of Cryptography - 17th International Conference, {TCC} 2019,
               Nuremberg, Germany, December 1-5, 2019, Proceedings, Part {I}},
  series    = {Lecture Notes in Computer Science},
  volume    = {11891},
  pages     = {229--250},
  publisher = {Springer},
  year      = {2019},
  doi       = {10.1007/978-3-030-36030-6_10},
  bibsource = {dblp computer science bibliography, https://dblp.org}
}

@inproceedings{AQY21,
      title={Cryptography from Pseudorandom Quantum States.}, 
      author={Ananth, Prabhanjan and Qian, Luowen and Yuen, Henry},
      booktitle={CRYPTO},
      year={2022}
}

@article{GLGEYQ23,
  title={A little magic means a lot},
  author={Gu, Andi and Leone, Lorenzo and Ghosh, Soumik and Eisert, Jens and Yelin, Susanne and Quek, Yihui},
  journal={arXiv preprint arXiv:2308.16228},
  year={2023}
}

@article{LQSYZ23,
  title={Quantum Pseudorandom Scramblers},
  author={Lu, Chuhan and Qin, Minglong and Song, Fang and Yao, Penghui and Zhao, Mingnan},
  journal={arXiv preprint arXiv:2309.08941},
  year={2023}
}

@article{ABKKK23,
  title={A quantum tug of war between randomness and symmetries on homogeneous spaces},
  author={Arvind, Rahul and Bharti, Kishor and Khoo, Jun Yong and Koh, Dax Enshan and Kong, Jian Feng},
  journal={arXiv preprint arXiv:2309.05253},
  year={2023}
}

@inproceedings{MY22,
  title = {Quantum commitments and signatures without one-way functions},
  author={Morimae, Tomoyuki and Yamakawa, Takashi},
  booktitle = {CRYPTO},
  year = {2022}
}

@inproceedings{BrakerskiS20,
  author    = {Zvika Brakerski and
               Omri Shmueli},
  editor    = {Daniele Micciancio and
               Thomas Ristenpart},
  title     = {Scalable Pseudorandom Quantum States},
  booktitle = {Advances in Cryptology - {CRYPTO} 2020 - 40th Annual International
               Cryptology Conference, {CRYPTO} 2020, Santa Barbara, CA, USA, August
               17-21, 2020, Proceedings, Part {II}},
  series    = {Lecture Notes in Computer Science},
  volume    = {12171},
  pages     = {417--440},
  publisher = {Springer},
  year      = {2020},
  doi       = {10.1007/978-3-030-56880-1_15},
  bibsource = {dblp computer science bibliography, https://dblp.org}
}

@inproceedings{Kretschmer21,
  author    = {William Kretschmer},
  editor    = {Min{-}Hsiu Hsieh},
  title     = {Quantum Pseudorandomness and Classical Complexity},
  booktitle = {16th Conference on the Theory of Quantum Computation, Communication
               and Cryptography, {TQC} 2021, July 5-8, 2021, Virtual Conference},
  series    = {LIPIcs},
  volume    = {197},
  pages     = {2:1--2:20},
  publisher = {Schloss Dagstuhl - Leibniz-Zentrum f{\"{u}}r Informatik},
  year      = {2021},
  doi       = {10.4230/LIPIcs.TQC.2021.2},
  bibsource = {dblp computer science bibliography, https://dblp.org}
}

@article{HBCLMNBKP22,
  title={Quantum advantage in learning from experiments},
  author={Huang, Hsin-Yuan and Broughton, Michael and Cotler, Jordan and Chen, Sitan and Li, Jerry and Mohseni, Masoud and Neven, Hartmut and Babbush, Ryan and Kueng, Richard and Preskill, John and others},
  journal={Science},
  volume={376},
  number={6598},
  pages={1182--1186},
  year={2022},
  publisher={American Association for the Advancement of Science}
}

@inproceedings{BFV19,
  author    = {Adam Bouland and
               Bill Fefferman and
               Umesh V. Vazirani},
  editor    = {Thomas Vidick},
  title     = {Computational Pseudorandomness, the Wormhole Growth Paradox, and Constraints
               on the AdS/CFT Duality (Abstract)},
  booktitle = {11th Innovations in Theoretical Computer Science Conference, {ITCS}
               2020, January 12-14, 2020, Seattle, Washington, {USA}},
  series    = {LIPIcs},
  volume    = {151},
  pages     = {63:1--63:2},
  publisher = {Schloss Dagstuhl - Leibniz-Zentrum f{\"{u}}r Informatik},
  year      = {2020},
  doi       = {10.4230/LIPIcs.ITCS.2020.63},
  bibsource = {dblp computer science bibliography, https://dblp.org}
}

@misc{ABFGVZZ23,
      title={Quantum Pseudoentanglement}, 
      author={Scott Aaronson and Adam Bouland and Bill Fefferman and Soumik Ghosh and Umesh Vazirani and Chenyi Zhang and Zixin Zhou},
      year={2023},
      eprint={2211.00747},
      archivePrefix={arXiv},
      primaryClass={quant-ph}
}

@misc{BBSS23,
      author = {Amit Behera and Zvika Brakerski and Or Sattath and Omri Shmueli},
      title = {Pseudorandomness with Proof of Destruction and Applications},
      howpublished = {Cryptology ePrint Archive, Paper 2023/543},
      year = {2023},
      note = {\url{https://eprint.iacr.org/2023/543}},
      url = {https://eprint.iacr.org/2023/543}
}

@misc{Zha12,
      author = {Mark Zhandry},
      title = {Secure Identity-Based Encryption in the Quantum Random Oracle Model},
      howpublished = {Cryptology ePrint Archive, Paper 2012/076},
      year = {2012},
      note = {\url{https://eprint.iacr.org/2012/076}},
      url = {https://eprint.iacr.org/2012/076}
}

@misc{mele2023introduction,
      title={Introduction to Haar Measure Tools in Quantum Information: A Beginner's Tutorial}, 
      author={Antonio Anna Mele},
      year={2023},
      eprint={2307.08956},
      archivePrefix={arXiv},
      primaryClass={quant-ph}
}
